\documentclass[10pt,reqno]{amsart}

\usepackage{amsfonts}
\usepackage[utf8]{inputenc}
\usepackage{graphicx}
\usepackage{csquotes} % needed when using babel and biblatex
\usepackage{mathtools}
\usepackage{amsmath}
\usepackage{amssymb}
\usepackage{amstext}
\usepackage{amsthm}
\usepackage{tikz-cd}
\usepackage[english]{babel}
\usepackage{mathrsfs}%the package for \mathscr
\usepackage[top=3cm,bottom=2.5cm,right=2cm,left=2cm]{geometry}
\usepackage[backend=biber, style=alphabetic, sorting=nyt ]{biblatex}
\usepackage{caption}
\usepackage{hyperref}

\theoremstyle{plain}

\newtheorem{theorem}{Theorem}[section]
\newtheorem{theorem*}{Theorem}
\newtheorem{corollary}[theorem]{Corollary}
\newtheorem{lemma}[theorem]{Lemma}
\newtheorem{proposition}[theorem]{Proposition}

\theoremstyle{definition}
\newtheorem{definition}[theorem]{Definition}
\newtheorem{remark}[theorem]{Remark}
\newtheorem*{remark*}{Remark}

\makeatletter
\newcommand{\proofstep}[1]{%
  \par% ensure starting on a new paragraph
  \addvspace{\medskipamount}% some vertical space
  \textit{#1\@addpunct{.}}\enspace\ignorespaces
}
\makeatother


\renewbibmacro{in:}{}

\begin{document}

\title{Non-trivial bundles and Algebraic Classical Field Theory}
\author[R. Brunetti]{Romeo Brunetti}
\address{Dipartimento di Matematica, Universit\`a di Trento, 38123 Povo (TN), Italy}
\email{romeo.brunetti@unitn.it}
\author[A. Moro]{Andrea Moro}
\address{Dipartimento di Matematica, Universit\`a di Trento, 38123 Povo (TN), Italy}
\email{andrea.moro-1@unitn.it}
%\date{}
\maketitle

\begin{abstract}
\noindent Inspired by the recent algebraic approach to classical field theory, we propose a more general setting based on the manifold of smooth sections of a non-trivial fiber bundle. Central is the notion of observables over such sections, \textit{i.e.} appropriate smooth functions on them. The kinematics will be further specified by means of the Peierls brackets, which in turn are defined via the causal propagators of linearized field equations. We shall compare the formalism we use with the more traditional ones.
\end{abstract}

\tableofcontents

%\mainmatter

%\pagestyle{fancy}

%\addcontentsline{toc}{chapter}{Introduction}

\section{Introduction}

The aim of the present paper is to generalize the recent treatment of relativistic classical field theory \cite{acftstructure}, seen as a Lagrangian theory based on (non linear) functionals over the infinite dimensional configuration space, to the case where the latter is made of sections of general bundles.

There are few mathematical treatments of classical field theories, all finding inspirations and drawing ideas from two main sources, the Hamiltonian and Lagrangian formalism of classical mechanics. If one intends also to include relativistic phenomena then there remain essentially only two rigorous frameworks, both emphasizing the geometric viewpoint: the multisymplectic approach (see \cite{gotay1998momentum}, \cite{gotay2004momentum}, \cite{forger2005covariant}) and another related to the formal theory of partial differential equations (see \cite{fati}, \cite{krupka}). They have several points in common and there is now a highly developed formalism leading to rigorous calculus of variations. 

Physicists look at Lagrangian classical field theories with more interest. This clearly amounts to developing a formalism in which one treats the intrinsic infinite dimensional degrees of freedom of the configuration spaces. Both previously cited frameworks use ingenuous ideas to avoid the direct treatment of infinite dimensional situations.  However, there exists another treatment of classical mechanics which emphasises more the algebraic and the analytic structures and is intrinsically infinite dimensional, which is named after the pioneering works of von Neumann (\cite{neumann1932operatorenmethode}) and Koopman (\cite{koopman1931hamiltonian}) and works directly in Hilbert spaces. If one is willing to generalise this last setting to field theories, one finds immediately an insurmountable difficulty namely, a result by Eells and Elworthy (see \cite{ely}, \cite{ely2}) constrains a configuration space, viewed as a second countable Hilbert manifold, to be smoothly embedded into its ambient space, i.e., it is just an open subset of the Hilbert space. Hence, we need to bypass this fact of life and find a clever replacement.

This task was done in the last decade, in which another treatment was developed that drew inspirations from perturbative quantum field theories in the algebraic fashion \cite{renormalizationgroup} and which is closer in spirit to the von Neumann-Koopman formalism. Based on these ideas, and moreover using as inputs also some crucial notions belonging to microlocal analysis, one of the authors (RB) in collaboration with Klaus Fredenhagen and Pedro Lauridsen Ribeiro \cite{acftstructure}, have formalised the new treatment for the case of scalar fields on globally hyperbolic spacetimes. It emphasises more the observables' point of view and deals directly with the configuration space as an infinite dimensional manifold but modelled now over locally convex spaces. This new approach gives a strong structural feedback to quantum field theories in the functional-algebraic formalism using deformation quantization, which is especially clear in a recent treatment given by the authors \cite{morophd,brunettimoro2}. In particular, we shall define a class of functionals named  ``microcausal functionals'' that are extremely important objects in perturbation theory for interacting quantum field theories since they are the image of local functionals under the chronological products. To work efficiently with such functionals one relies heavily on the clarifications given in the last thirty years about the most appropriate calculus on locally convex spaces (see, \textit{e.g.}, \cite{mb,kriegl1997convenient,glockneeb}).  

As advertised above, it is one of the main aims of the present paper to generalise such treatment to the more complicated situation in which fields are sections of fiber bundles. There are plenty of examples in the physics literature dealing with such structures, one for all being the case of nonlinear $\sigma$-models (wave maps in the mathematicians' language). At first sight the idea looks straightforward to implement, however it contains some subtleties whose treatment needs a certain degree of care. Indeed, in our general setting, images of the fields are never linear spaces and moreover, the global configuration space has only a manifold structure. This forces us to generalize many notions like the support of functionals, or the central notion of locality/additivity, over configuration space, which can be given in two different formulations: a global formulation that uses the notion of relative support already used in \cite{gravbrunetti} and a local one that uses the notion of charts over configuration space seen as a infinite dimensional manifold. It is gratifying that both notions give equivalent results, as shown \textit{e.g.}\ in Proposition~\ref{prop_1_additivity}. 

We summarize the content of this article.

Section~$2$ is devoted to the geometrical tools used in the rest of the article: we introduce the classical geometrical formalism based on jet manifolds and then the infinite dimensional formalism. 

Section~$3$ focuses on the definition of observables, their support, and the introduction to various classes of observables depending on their regularity. In particular two of this classes admit a ultralocal characterization, i.e. local in the sense of the manifold structure of the space of sections. In the end we introduce the notion of \textit{generalized Lagrangian}, essentially showing that each Lagrangian in the standard geometric approach is a Lagrangian in the algebraic approach as well. We then discuss how linearized field equations are derived from generalized Lagrangians. A crucial point here is the way we characterize microlocal functionals in Proposition \ref{porop_1_muloc_charachterization}, by applying the nonlinear version of Peetre theorem.

Section~$4$ begins with some preliminaries  about normal hyperbolicity and normally hyperbolic operators. Here, generalized Lagrangians of second order which are normally hyperbolic, as it is in the case of wave maps, play a major role. We then show the existence of the causal propagator which in turn is used in Definition \ref{def_1_Peierls} to define the Poisson bracket on the class of \textit{microlocal functionals}. Then we enlarge the domain of the bracket to the previously recalled \textit{microcausal functionals}, defined by requiring a specific form of the wave front set of their derivatives. Finally, Theorems \ref{thm_1_mucaus_1}, \ref{thm_1_peierls_closedness} and \ref{thm_1_jacobi} establish the Poisson $*$-algebra of microcausal functionals. 

Section~$5$ presents results that culminate in Theorem \ref{thm_1_mucaus_top}, which establishes that microcausal functionals can be given the topology of a nuclear locally convex space. Furthermore Propositions \ref{prop_1_C-infty_ring} and \ref{prop_1_mu_caus_top_prop} give additional properties concerning this space and its topology. We conclude the section by defining the on-shell ideal with respect to the Lagrangian generating the Peierls bracket and the associated Poisson $*$-algebraic ideal. 

In Section~$6$ we elucidate the previous results by showing how to adapt them to the case of wave maps. In particular we shall write the expression for the causal propagator in $4$-dimensional spacetime.  

Finally, in the Appendix we give details on the topology of the space of sections $\Gamma^{\infty}(M\leftarrow B)$ and its manifold structure, gathering a number of results that cannot be found in a single reference. We hope that this may make the paper more enticing for newcomers in the field.

\section{Geometrical setting} 

\subsection{Preliminaries}

Let $M$ be a smooth $m$-dimensional manifold, suppose that there is a smooth section $g$ of $T^*M \vee T^*M \rightarrow M )$, where $\vee$ denotes the symmetric tensor product, such that its signature is $(-,+,\ldots,+)$; then $(M,g)$ is called a Lorentz manifold and $g$ its Lorentzian metric. If we take local coordinates $(U,\lbrace x^{\mu} \rbrace)$ and denote $d\sigma(x)\doteq dx^1 \wedge \ldots \wedge dx^m$, then $d\mu_g(x)=\sqrt{\vert g(x) \vert} d\sigma(x)$ is the canonical volume element of $M$, where as usual we denote by $g(x)$ the determinant of $g$ calculated at $x \in M$. Any Lorentzian metric $g$ induces the so called musical isomorphisms $g^{\sharp}:TM \rightarrow T^*M:(x,v)\mapsto (x,g_{\mu \nu}v^{\mu} dx^{\nu})$, $g^{\sharp}:T^*M \rightarrow TM:(x,\alpha)\mapsto (x,g^{\mu \nu}\alpha_{\mu} \partial_{\nu})$, where $\lbrace dx^{\mu} \rbrace_{\mu=1,\ldots,m}$ is the standard basis of $T_x^*M$ in local coordinates $(U,\lbrace x^{\mu} \rbrace)$, and $\lbrace \partial_{\mu} \rbrace_{\mu=1,\ldots,m}$ the dual basis of $T_xM$. 

Given any Lorentzian manifold $(M,g)$, a non zero tangent vector $v_x \in T_xM$ is \textit{timelike} if $g_x(v_x,v_x)<0$, \textit{spacelike} if $g_x(v_x,v_x)>0$, \textit{lightlike} if $g_x(v_x,v_x)=0$; similarly a curve $\gamma:\mathbb{R}\rightarrow M: t \mapsto \gamma(t)$ is called \textit{timelike} (resp. \textit{lightlike}, resp. \textit{spacelike}) if at each $t \in \mathbb{R}$ its tangent vector is \textit{timelike} (resp. \textit{lightlike}, resp. \textit{spacelike}), a curve that is either timelike or lightlike is called \textit{causal}. We denote the cone of timelike vectors tangent to $x\in M$ by $V_g(x)$. A Lorentzian manifold admits a \textit{time orientation} if there is a global timelike vector field $T$, then timelike vectors $v\in T_xM$ that are in the same connected component of $T(x)$ inside the light cone, are called \textit{future directed}. When an orientation is present we can consistently split, for each $x\in M$, the set $V_g(x)$ into two disconnected components $V^+_x(x)\cup V^-_g(x)$ calling them respectively the sets of \textit{future directed} and \textit{past directed} tangent vectors at $x$. Given $x,y\in M$, we say that $x\ll y$ if there is a future directed timelike curve joining $x$ to $y$, and $x \leq y$ if there is a causal curve joining $x$ to $y$. We denote $I^{+}_M(x)=\lbrace y \in M : x \ll y\rbrace$, $I^{-}_M(x)=\lbrace y \in M : x \gg y\rbrace$, $J^{+}_M(x)=\lbrace y \in M : x \leq y\rbrace$, $J^{-}_M(x)=\lbrace y \in M : x \geq y\rbrace$ and call them respectively the \textit{chronological future}, \textit{chronological past}, \textit{causal future}, \textit{causal past} of $x$. $(M,g)$ is said to be \textit{globally hyperbolic} if $M$ is causal, i.e. there are not closed causal curves on $M$ and the sets $J_M(x,y)\doteq J_M^+(x)\cap J_M^-(y)$ are compact for all $x,y \in M$. Equivalently, $M$ is globally hyperbolic if there is a smooth map $\tau:M \rightarrow \mathbb{R}$ called \textit{temporal function} such that its level sets, $\Sigma_t$ are Cauchy hypersurfaces, that is every inextensible causal curve intersects $\Sigma_t$ exactly once. A notable consequence is that any globally hyperbolic manifold $M$ has the form $\Sigma \times \mathbb{R}$ for some, hence any, Cauchy hypersurface $\Sigma$. We point to \cite{semiriem} for details on the geometric structure and to \cite{geroch1970domain}, \cite{bernal2003smooth, bernal2005smoothness} for details on globally hyperbolic manifolds.\\

 A \textit{fiber bundle} is a quadruple $(B,\pi,M,F)$, where $B$, $M$, $F$ are smooth manifold called respectively the \textit{bundle}, the \textit{base} and the \textit{typical fiber}, such that:
\begin{enumerate}
    \item[$(i)$] $\pi:B \rightarrow M$ is a smooth surjective submersion;
    \item[$(ii)$] there exists an open covering of the base manifold $M$, $\lbrace U_{\alpha} \rbrace_{\alpha \in A}$ admitting, for each $\alpha \in A$, diffeomorphisms $t_{\alpha}: \pi^{-1}(U_{\alpha})\rightarrow U_{\alpha}\times F$, called trivializations, which are fiber respecting \textit{i.e.} $\mathrm{pr}_1 \circ t_{\alpha} = \pi|_{\pi^{-1}(U_{\alpha})}$.
\end{enumerate}
Given $U_{\alpha}$, $U_{\beta}$ subsets of $M$, we can define the \textit{transition mappings} $g_{\alpha\beta}: U_{\alpha\beta}\times F \to F :(x,y) \mapsto g_{\alpha\beta}(x,y)\doteq \mathrm{pr}_2 \circ t_{\alpha}\circ t_{\beta}^{-1}(x,y)$. We remark that $\mathrm{pr}_1 \circ t_{\alpha} = \pi|_{\pi^{-1}(U_{\alpha})}$ implies $\mathrm{pr}_1 = \pi \circ t_{\alpha}^{-1}$.\\

Using trivialization it is possible to construct charts of $B$ via those of $M$ and $F$. We call those \textit{fibered coordinates} and we denote them by $(x^{\mu},y^i)$ with the understanding that Greek indices denote the base coordinates and Latin indices the fiber coordinates. Given two fiber bundles $(B_i,\pi_i,M_i,F_i)$, $i=1,2$, we define a \textit{fibered morphism} as a pair $(\Phi,\phi)$, where $\Phi:B_1 \rightarrow B_2$, $\phi:M_1 \rightarrow M_2$ are smooth mappings, such that $\pi_2 \circ \Phi= \phi \circ \pi_1$. We denote by 
\begin{equation}\label{eq_1_space_of_sections}
	\Gamma^{\infty}(M\leftarrow  B)=\lbrace \varphi:M \rightarrow B, \space\ \mathrm{smooth} \space\ : \pi \circ \varphi =id_M \rbrace
\end{equation}
the space of \textit{sections} of the bundle. \\

\textit{Vector bundles} are a particular kind of fiber bundle whose standard fibers are vector spaces and their transition mappings act on fibers as transformations of the general Lie group associated to the standard fibers. We denote coordinates of these bundles by $\lbrace x^{\mu}, v^i\rbrace$. Suppose that $(E,\pi,M,V)$, $(F,\rho,M,W)$ be vector bundles over the same base manifold, then it is possible to construct a third vector bundle $(E\otimes F,\pi\otimes \rho,M,V\otimes W)$ called the tensor product bundle whose standard fiber is the tensor product of the standard fibers of the starting bundles. For example, $k$-forms over $M$ are sections of the vector bundle $(\Lambda_k(M),\tau_{\Lambda},M,\Lambda^k\mathbb R^m)$.\\

In the sequel we will use particular vector fields of $B$: they are called \textit{vertical vector fields} and belong to $\mathfrak{X}_{\mathrm{vert}}(B)=\lbrace X \in \Gamma^{\infty}(B \leftarrow TB) : T \pi(X)=0 \rbrace \subset \mathfrak X (B) $, where $\pi: B \rightarrow M$ is the bundle projection and $T$ denotes the tangent functor. We will denote by $\Phi^{X}_t : B \rightarrow B$ the flow of any vector field on $B$, and assume in the rest of this work that the parameter $t$ varies in an appropriate interval which has been maximally extended. Note that if $X \in \mathfrak{X}_{\mathrm{vert}}(B)$, then $\Phi^{X}_t$ is a fibered morphism whose base projection is the identity over $M$. Vertical vector fields can be seen as a sections of the \textit{vertical vector bundle}, $(VB\doteq \mathrm{ker}(T\pi),\tau_V,B)$ which is easily seen to carry a vector bundle structure over $B$. Another construction that we shall often use is that of \textit{pullback bundle}: given a fiber bundle $B$ and a smooth map $\psi: M\to N$, we can describe another bundle over $N$ with the same typical fibers as the original fiber bundle and with total space defined by $\psi^\ast B\doteq\{(n,b)\in N\times B :\ \psi(n)=\pi(b)\}$ and projection $\psi^\ast\pi \doteq \mathrm{pr}_1|_{\psi^\ast B}$. In particular, we call $\psi^*B$ the \emph{pullback bundle of $B$ along $\psi$}.\\

Another important notion necessary to the geometric framework of classical field theories are jet bundles. Heuristically they geometrically formalize PDEs. For general references see \cite[Chapter IV, $\S \ 12$]{kolar2013natural} or \cite{jet}. Rather then giving the most general definition, we simply recall the bundle case. Given any fiber bundle $(B,\pi,M,F)$, two sections, $\varphi_1$, $\varphi_2$ are $k$th-order equivalent in $x \in M$, which we write, $\varphi_1 \sim_{x}^k \varphi_2$ if for all $f \in C^{\infty}(B)$, $\gamma \in C^{\infty}(\mathbb{R},M)$ having $\gamma(0)=x$, the Taylor expansions at $0$ of order $k$ of $f\circ \varphi_1 \circ \gamma$ and $f\circ \varphi_2 \circ \gamma$ coincide. The relation $ \sim_{x}^k $ becomes an equivalence relation and we denote by $j^k_x\varphi$ the equivalence class with respect to $\varphi$. Letting $J^k_xB \doteq \Gamma^{\infty}_x(M\leftarrow B)/\sim_{x}^k$, where $\Gamma^{\infty}_x(M\leftarrow B)$ are the germs of local sections of $B$ defined on a neighborhood of $x$, the $k$th order \textit{jet bundle} is then 
$$
	J^kB \doteq \bigsqcup_{x\in M} J^k_xB\ . 
$$
The latter inherits the structure of a fiber bundle with base either $M$, $B$ or any $J^lB$ with $l<k$. If $\lbrace x^{\mu},y^j \rbrace$ are fibered coordinates on $B$, then we induce fibered coordinates $ \lbrace x^{\mu},y^j,y^j_{\mu},\ldots,y^j_{\mu_1 \ldots \mu_k} \rbrace$ on $J^kB$ where Greek indices are understood to be symmetric. The latter coordinates embody the geometric notion of PDEs. The family $\lbrace (J^rB,\pi^r) \rbrace_{r \in \mathbb{N}}$ with $\pi^r:J^rB\rightarrow M$ allows an inverse limit $(J^{\infty}B,\pi^{\infty},\mathbb{R}^{\infty})$ called the \textit{infinite jet bundle} over $M$, it can be seen as fiber bundle whose standard fiber $\mathbb{R}^{\infty}$ is a Fréchet topological vector space. Its sections denoted by $j^{\infty}\varphi$ are called infinite jet prolongations.\\

Given a vector bundle $E \to M$, we define its space of distributional sections $\Gamma^{-\infty}(M\leftarrow E)$ as the strong topological dual of $\Gamma^{\infty}_c(M\leftarrow E)$ equipped with the standard limit Fréchet topology. Notice that, if we denote by $E'\to M$ the dual bundle of $E\to M$, given $\mu\in\Gamma^{\infty}\big(M\leftarrow E'\otimes \Lambda_m(M)\big)$, we can define
$$
    X \in \Gamma^{\infty}(M\leftarrow E) \mapsto \int_M \langle X, \mu\rangle (x)\in \mathbb{R}\ .
$$
Thus $\Gamma^{\infty}\big(M\leftarrow E'\otimes \Lambda_m(M)\big)$ embeds into $\Gamma^{-\infty}(M\leftarrow E)$. Let $U \subseteq M$ be open and $s \in  \Gamma^{-\infty}(M\leftarrow E) $, the \textit{restriction of} $s$ \textit{to} $U$ is the distributional section
$$ 
   s|_{U}(X)= s(X), \ \ X \in \Gamma^{\infty}_{c}(U\leftarrow E|_U)\ .
$$
The \textit{support} of $s$ is the set
$$
    \mathrm{supp}(s)= \bigcap_{ \substack {A \subset M \mathrm{closed} \\ \left. s\right|_{M \backslash A}=0}}A\ .
$$

We remark that playing with partitions of unity it is possible to endow $ \Gamma^{-\infty}(M\leftarrow E) $ with the structure of a \textit{fine sheaf}. This in particular implies the principle of localization: a distributional section is the zero section if and only if for every point $x\in M$ there is an open neighborhood $U \ni x$ such that $s|_U=0$. If we denote by $\Gamma^{-\infty}_c(M \leftarrow E)$ the space of \textit{compactly supported distributional sections}, then one can show that it is isomorphic to the space of continuous linear mappings $s :\Gamma^{\infty}(M\leftarrow E) \to \mathbb{R}$, \textit{i.e.} the dual of the Fréchet space $\Gamma^{\infty}(M\leftarrow E)$.

In the sequel, we will usually employ distributional sections in $\Gamma^{-\infty}(M\leftarrow \varphi^{*}VB)$, where $VB \to B$ is the vertical bundle of $B$ and $\varphi : M \to B$ a section of the bundle $B \to M$.

To estimate singularities of distributional sections we shall use the notion of wave front set: let $\pi:E\to M$ be a vector bundle and let $\{(\pi^{^-1}(U_{\alpha}),t_{\alpha})\}_{\alpha}$ be a family of trivializations of $E$. If $s\in \Gamma^{-\infty}_c(M\leftarrow E)$, then $(t_{\alpha})_*s=(s^1,\ldots,s^k)$ where $k$ is the dimension of the fiber of $E$ and each $s^i\in \mathcal{D}'(M)$. Then we set
\begin{equation}\label{eq_WF_distributional_sections}
    \mathrm{WF}(s)\doteq\bigcup_{i=1}^k \mathrm{WF}(s^i)\ .
\end{equation}
The above definition does not depend on the chosen trivialization for diffeomorphisms. This entails that we can straightforwardly generalize the results of \cite[Chapter 8]{hormanderI} to distributional sections in $\Gamma^{-\infty}_c(M\leftarrow E)$.

\subsection{Topology and geometry of field configurations}

We here give a synthetic exposition concerning the topology and manifold structure of $\Gamma^{\infty}(M\leftarrow B)$. For further details, we defer to the Appendix and the references given therein. We stress that for a generic bundle the space of global smooth sections might be empty (\textit{e.g.} for nontrivial principal bundles), therefore we assume that our bundles do possess them. Indeed that is the case whenever we are considering trivial bundles, vector bundles, or bundles of geometric objects such as natural bundles (see \textit{e.g.}\ \cite[$\S$ 14]{kolar2013natural}).

Let $M$, $N$ be Hausdorff topological spaces and let $C(M,N)$ the space of continuous mappings between the two spaces. The \textit{compact open} topology $\tau_{CO}$ or CO-topology is the topology generated by a basis whose elements have the form 
$$
    N(K,V)=\{ \varphi \in C(M,N) : \varphi(K)\subset V\}\ , 
$$
where $K\subset M$ is a compact subset and $V\subset N $ is open. The \textit{wholly open topology} or $WO$-topology is the one generated by the subbasis
$$
    \lbrace \varphi \in C(M,N) : f(M) \subseteq U \rbrace
$$ 
where $U\subset N$ is open. The \textit{graph topology} or $WO^0$-topology on $C(M,N)$ is generated by requiring that
$$
        G:C(M,N) \ni \varphi \mapsto G_{\varphi} \in \big( C(M,M\times N),\tau_{\mathrm{WO}}\big)\ ,
$$
with $G_{\varphi} : M \to M \times N, \quad x \mapsto (x, \varphi(x)) $ being the graph mapping, is an embedding. This topology is Hausdorff. Finally, the \textit{Whitney} $C^k$ \textit{topology}, or $WO^k$-topology is defined by requiring 
$$
	j^k :C^{\infty}(M,N) \rightarrow \left( C(M,J^k(M \times N)),WO-\mathrm{topology}\right)
$$
to be an embedding. When $k=\infty$ we call the latter \textit{Whitney topology} or $WO^{\infty}$-topology. If $M$ is second countable and finite dimensional, $N$ is metrizable, given an exhaustion of compact subsets $\{K_n\}_{n\in\mathbb{N}}$ in $M$ and a sequence of natural numbers $\{k_n\}_{n\in \mathbb{N}}\nearrow \infty $, then a basis of open subset is given by subsets of the form
$$
	W(K_n,U_n)=\left\{ \varphi \in C^{\infty}(M,N) : j^{k_n}\varphi(M\backslash K_n)\subset U_n \right\}
$$
where each $U_n\subset J^{k_n}(M,N)$ is an open subset. This topology enjoys several properties collected in Propositions \ref{prop_1_WO^k_properties}, \ref{prop_1_WO^0-convergence} and Corollary \ref{coro_1_WO^0-curves}. We here stress two facts:
\begin{itemize}
    \item[$(i)$] the convergence of a sequence $\varphi_n$ to $\varphi\in C^{\infty}(M,N)$ can be characterized as follows: there is a compact subset, $K$ for which $\varphi_n|_{M \backslash K}=\varphi|_{M \backslash K}$, and $j^k\varphi_n \to j^k\varphi$ uniformly on $K$ for all $k \in \mathbb N$.
    \item[$(ii)$] continuous curves $\gamma : \mathbb R \to C^{\infty}(M,N)$ have the property that for each compact interval $I\subset \mathbb R$, there is a compact subset $K$ of $M$, for which $\gamma(t_1)|_{M \backslash K} \equiv \gamma(t_2)|_{M \backslash K}$ for all $t_1, \ t_2 \in I$.
\end{itemize}

Let $\varphi, \psi \in C^{\infty}(M,N)$, we define the \textit{relative support} of $\psi$ with respect to $\varphi$ as the subset
\begin{equation}\label{eq_1_rel_support}
    \mathrm{supp}_{\varphi}(\psi)\doteq \{ x \in M : \ \varphi(x) \neq \psi(x) \}\ .
\end{equation}
We say that $\varphi \sim \psi$ whenever $\mathrm{supp}_{\varphi}(\psi)$ is compact. The \emph{refined Whitney topology} is the coarsest topology on $C^{\infty}(M,N)$ finer than the $WO^{\infty}$-topology and for which the subsets 
$$
    \mathcal{V}_{\varphi}\doteq \lbrace \psi \in C^{\infty}(M,N): \psi \sim \varphi \rbrace
$$ 
are open. This is easily achieved by adding to the basis of open subsets defined above, the family generated by finite intersections between the elements $\mathcal{V}_{\varphi}$ and $W(K_n,U_n)$. We equip $\Gamma^{\infty}(M\leftarrow B)$ with the subspace topology of $C^{\infty}(M,B)$ with the refined Whitney topology; notice that $\varphi \in \Gamma^{\infty}(M\leftarrow B)  \subset C^{\infty}(M,B)$ if and only if $\pi_{*}(\varphi)=\pi\circ \varphi = \mathrm{id}_M$, by Proposition 7.1 in \cite{michor2011manifolds}, $\pi_{*}$ is continuous, so the equation $\pi_{*}(\cdot) =\mathrm{id}_M $ defines a closed subset in $C^{\infty}(M,B)$ with the refined Whitney topology.

The smooth structure of the manifold $C^{\infty}(M,N)$ is the modelled on the locally convex spaces $\Gamma^{\infty}_c(M\leftarrow \varphi^*TN)$, with Bastiani calculus (see Definition \ref{def_A_Bastiani_smooth_map} for the precise notion).  

The manifold structure of $C^{\infty}(M,N)$ is generated by \textit{ultralocal charts} $\{ \mathcal{U}_{\varphi}, u_{\varphi}\}$ as follows: 

\begin{definition}
Let $\exp$ be any Riemannian exponential\footnote{In the proof of Theorem \ref{thm_1_mfd_mappings} we show that the induced smooth structure does not depend on the chosen Riemannian exponential.} on $N$  and denote by $\widetilde U \subset N \times N$ the neighborhood of the diagonal where $\exp : \widetilde V \subset TN \to N \times N$ is a diffeomorphism; then, an \textit{ultralocal chart} $\{ \mathcal{U}_{\varphi}, u_{\varphi}\}$, is determined by the following choices
$$
    \mathcal{U}_{\varphi} = \{ \psi \in \mathcal{V}_{\varphi} : \ (\varphi,\psi)(M)\subset \widetilde{U} \}
$$
and
\begin{equation}\label{eq_1_chart_mapping_2}
\begin{aligned}
    u_{\varphi}: & \  \mathcal{U}_{\varphi} \ni \psi \mapsto u_{\varphi}(\psi)\in \Gamma^{\infty}_c(M\leftarrow \varphi^*TN)    \\
    u_{\varphi}(\psi) : \  & M \ni x \mapsto  u_\varphi(\psi)(x)=\exp^{-1} (\varphi(x),\psi(x))\simeq \big(\varphi(x),\exp^{-1}_{\varphi(x)}(\psi(x))\big)\ .
\end{aligned}
\end{equation}
\end{definition}

Notice that by Theorem \ref{thm_A_Bastiani smooth_pushforward}, compositions $u_{\varphi} \circ u_{\psi}^{-1}$ of charts mappings are Bastiani smooth.

\begin{remark}
    We stress that a similar manifold structure can be obtained by using a different notion of calculus: \textit{convenient calculus}. In extreme synthesis, a mapping $f$ between locally convex spaces $X, \ Y$ is conveniently smooth if and only if $f_{*}(C^{\infty}(\mathbb R, X)\subset C^{\infty}(\mathbb R,Y)$. The charts of the manifold structure remain the same, however the topology will be finer. For details see \textit{e.g.} \cite[Chapter I and $\S$ 42]{kriegl1997convenient}. However, it has been shown (see $\S$2 and Proposition 2.2 of \cite{glockner2005discontinuous}) that not all conveniently smooth mappings can be continuous, thereby not allowing the use of Schwartz kernel theorem for the derivatives of (conveniently) smooth functions $C^{\infty}(M,N) \to \mathbb R$. Motivated by this need we shall adopt the notion of Bastiani calculus. However, remind that for Fr\'echet spaces the two calculi coincide.
\end{remark}

This differential structure of $\Gamma^{\infty}(M \leftarrow B)$ is modeled on the locally convex spaces $\Gamma^{\infty}_c(M\leftarrow \varphi^*VB)$ induced as a submanifold of $C^{\infty}(M,B)$ with ultralocal charts $\{ \mathcal{U}_{\varphi}, u_{\varphi}\}$, $\mathcal{U}_{\varphi} = \{ \psi \in \mathcal{V}_{\varphi} : \ (\varphi,\psi)(M)\subset \widetilde{U} \}$,

\begin{equation}\label{eq_1_chart_mapping_non_trivial}
\begin{aligned}
    u_{\varphi}: \ & \  \mathcal{U}_{\varphi} \ni \psi \mapsto u_{\varphi}(\psi)\in \Gamma^{\infty}_c(M\leftarrow \varphi^*VB)    \\
    u_{\varphi}(\psi) : \  & M \ni x \mapsto  u_\varphi(\psi)(x)=\widetilde\exp^{-1} (\varphi(x),\psi(x))\simeq \big(\varphi(x),\widetilde\exp^{-1}_{\varphi(x)}(\psi(x))\big)\ ;
\end{aligned}
\end{equation}
where $\widetilde\exp$ is a suitably modified version of some Riemannian exponential on $B$ and $\widetilde U \subset B \times B$ the neighborhood of the diagonal where the latter mapping becomes a diffeomorphism. 

The kinematic tangent space $T_{\varphi}\Gamma^{\infty}(M\leftarrow B)$ at point $\varphi$ is canonically isomorphic to $ \Gamma^{\infty}_c(M\leftarrow \varphi^{*}VB)$. The kinematic tangent bundle $(T\Gamma^{\infty}(M\leftarrow B),\tau_{\Gamma}, \Gamma^{\infty}(M\leftarrow B))$ is therefore defined in analogy with the finite dimensional case, and carries a canonical infinite dimensional bundle structure with trivializations
$$
	t_{\varphi}: \tau^{-1}_{\Gamma}(\mathcal{U}_{\varphi})\rightarrow \mathcal{U}_{\varphi} \times \Gamma^{\infty}_c(M\leftarrow \varphi^{*}VB)\ .
$$
Through $t_{\varphi}$ we can identify points of $T\Gamma^{\infty}(M\leftarrow B)$ with pairs $(\varphi,\vec{X}_{\varphi})$. Notice that an element of $T_{\varphi}\Gamma^{\infty}(M\leftarrow B)$, can equivalently be seen as a section of the (finite dimensional) vector bundle $\Gamma^{\infty}_c(M\leftarrow \varphi^{*}VB)$. When using the latter interpretation, we will write the section in local coordinates as $\vec{X}_{\varphi}(x)=\vec{X}^i(x)\partial_i\big\vert_{\varphi(x)}$. Vector fields for $\Gamma^{\infty}(M\leftarrow B)$ are therefore Bastiani smooth sections $\Vec{X} : \varphi \to X_{\varphi}$ of the tangent bundle $T\Gamma^{\infty}(M\leftarrow B) \to \Gamma^{\infty}(M\leftarrow B)$.   %To avoid confusion between the two we adopt different notations: we shall write the arrow symbol over capital Roman letters, e.g. $(\vec{X}, \vec{Y}, \vec{Z}, \dots)$, to denote elements in $T_{\varphi}\Gamma^{\infty}(M\leftarrow B)$; we shall use instead use $(\mathcal{X}, \mathcal{Y}, \mathcal{Z}, \dots) $ to denote sections of $\Gamma^{\infty}_c(M\leftarrow \varphi^{*}B)$. 
We will use Roman letters, e.g. $(s,\vec{u}, \dots)$ to denote distributional sections in $\Gamma^{-\infty}(M\leftarrow \varphi^{*}VB)$.

Finally, we recall that we can always choose a connection on $T\Gamma^{\infty}(M\leftarrow B)$, that is a fiber respecting splitting of $TT\Gamma^{\infty}(M\leftarrow B)$ into  horizontal and vertical part. This induces a notion of covariant derivative which is intrinsic and helps provide a chart independent classification for certain classes of functionals (\textit{c.f.} Definition \ref{def_1_func_classes}). Such a connection, as detailed in \eqref{eq_1_def_connection} and \eqref{eq_1_connection_induced}, can be always induced via a connection $\Gamma$ on the fiber $F$ of $B$ by setting 
$$
    \Gamma_{\varphi}(\vec{X},\vec{Y})(x)\doteq \widetilde{\Gamma}_{jk}^i({\varphi}(x))\vec{X}_{\varphi}^j(x)\vec{Y}_{\varphi}^k(x) \partial_i\vert_{{\varphi}(x)}\ .   
$$

%%%%%%%%%%%%%%%%%%%%%%%%%% OBSERVABLES %%%%%%%%%%%%%%%%%%%%%%%%%

\section{Observables}

By \textit{functional} we mean a smooth mapping 
$$
	F: \mathcal{U}\subset \Gamma^{\infty}(M\leftarrow B) \to \mathbb{R}\ ,
$$
where $\mathcal{U}$ is an open set in the $CO$-topology generated by \eqref{eq_1_CO_open}. Since smoothness is tested on ultralocal charts, a functional $F$ is smooth if and only if, given any ultralocal atlas $\lbrace \mathcal{U}_{\varphi}, u_{\varphi} \rbrace_{\varphi \in \mathcal{U}}$ its localization
\begin{equation}\label{eq_1_localization_of_F}
	F_{\varphi} \doteq F \circ u_{\varphi}^{-1} : \Gamma^{\infty}_c(M\leftarrow \varphi^*VB) \rightarrow \mathbb{R}\ ,
\end{equation}
is smooth for all $\varphi \in \mathcal{U}$ in the sense of Definition \ref{def_A_Bastiani_smooth_map}. 

The first notion we introduce is the spacetime support of a functional. The idea is to follow the definition of support given in \cite{acftstructure}, and account for the lack of linear structure on the fibers of the configuration bundle $B$.

\begin{definition}\label{def_1_func_support}
Let $F$ be a functional over $\mathcal{U}$, $CO$-open, then its support is the closure in $M$ of the subset $x\in M$ such that for all $V \subset M$ open neighborhood of $x$, there is $ \varphi \in \mathcal{U}$, $\vec{X}_{\varphi} \in \Gamma^{\infty}_c(M\leftarrow \varphi^{*}VB)$ having $\mathrm{supp}(\vec{X}_{\varphi})\subset V$, for which $F_{\varphi}(\vec{X}_{\varphi})\neq F_{\varphi}(0)$. The set of functionals over $\mathcal{U}$ with compact spacetime support will be denoted by $\mathcal{F}_c(B,\mathcal{U})$ and its elements called \emph{observables}.
\end{definition}

Let us display some examples of functionals. Given $\alpha \in C^{\infty}( B,\mathbb{R})$, consider
\begin{equation}\label{eq_1_generic_func}
	F_{\alpha}:\Gamma^{\infty}(M\leftarrow B) \rightarrow \mathbb{R}: \varphi \mapsto F_{\alpha}(\varphi) \doteq  \left		     \lbrace \begin{array}{lr}
      \frac{1}{1+\sup_{M}(\alpha(\varphi))} & \alpha(\varphi) \space\ \text{bounded },\\
       0 & \text{otherwise }.
             \end{array} \right. 
\end{equation}
If $f\in C^{\infty}_c(M)$ and $\lambda \in \Omega_m(J^rB)$ define
\begin{equation} \label{eq_1_gen_Lag_1}
	\mathcal{L}_{f,\lambda} :\Gamma^{\infty}(M\leftarrow B) \rightarrow \mathbb{R}:  \varphi \mapsto \mathcal{L}_{f,\lambda}(\varphi) \doteq  \int_M f(x) j^r\varphi^*\lambda(x) \mathrm{d}\mu_g(x)\ .
\end{equation}
On the other hand if $f, \lambda$ are as above and $\chi:\mathbb{R}\rightarrow \mathbb{R}$ with $0\leq \chi \leq 1$, $\chi(t)=1 \space\ \forall  \vert t \vert \leq 1/2$ and $\chi(t)=0 \space\ \forall  \vert t \vert \geq 1/2$ define
\begin{equation}\label{eq_1_reg_func}
	G_{f,\lambda,\chi}:\Gamma^{\infty}(M\leftarrow B) \rightarrow \mathbb{R}:  \varphi \mapsto G_{f,\lambda,\chi}(\varphi) \doteq e^{1-\chi\big( (\mathcal{L}_{f,\lambda}(\varphi))^2\big)}\ .
\end{equation}
\vspace{0.1cm}

We can endow $\mathcal{F}_c(B,\mathcal{U})$ with the following operations
\begin{equation}
    (F,G) \mapsto (F+G)(\varphi)\doteq F(\varphi)+G(\varphi)\ ;
\end{equation}
\begin{equation}
    (\alpha \in \mathbb{R}, F) \mapsto (\alpha F)(\varphi) \doteq \alpha F(\varphi)\ ;
\end{equation}
\begin{equation}\label{eq_1_classical_algebra_product}
    (F,G) \mapsto (F \cdot G)(\varphi) \doteq F(\varphi)G(\varphi)\ ;
\end{equation}
\begin{equation}
    F \mapsto F^{*}, \ F^{*}(\varphi) \doteq {F(\varphi)}\footnote{In adherence to standard classical field theory, we use real functionals, which makes involution a trivial operation; we remark though that one could repeat \textit{mutatis mutandis} everything with $\mathbb{R}$ replaced by $\mathbb{C}$, then involution becomes $F^{*}(\varphi) \doteq \overline{F(\varphi)}\in \mathbb C$ and the algebraic content is that of a $*$-algebra with unity.}.
\end{equation}

It can be shown that those operation preserve the compactness of the support, turning $\mathcal{F}_c(B,\mathcal{U})$ into a \textit{commutative algebra with unit} where the unit element is given by $\varphi \mapsto 1\in \mathbb{R}$. That involution and scalar multiplication are support preserving is trivial, to see that for multiplication and sum we use

\begin{lemma}
Let $F$, $G$ be functionals over $\mathcal{U} \subset \Gamma^{\infty}(M\leftarrow B)$ $CO$-open subset, then
\begin{itemize}
\item[$(i)$] $\mathrm{supp}(F+G)\subset \mathrm{supp}(F) \cup \mathrm{supp}(G)$\ ,
\item[$(ii)$] $\mathrm{supp}(F\cdot G)\subset \mathrm{supp}(F) \cup \mathrm{supp}(G)$\ .
\end{itemize}
\end{lemma}

Notice that the more restrictive version of $(ii)$ with the intersection of domains does not hold in general, this can be checked by taking the constant functional $G(\varphi) \equiv 1 \in \mathbb R \space\ \forall \varphi \in \mathcal{U}$, then $\mathrm{supp}(G)=\emptyset$ while $\mathrm{supp}(F+G)$, $\mathrm{supp}(F\cdot G)=\mathrm{supp}(F)$.

\begin{proof}
Suppose that $x \notin \mathrm{supp}(F) \cup \mathrm{supp}(G) $, then there is an open neighborhood $V$ of $x$ such that for any $X \in \Gamma^{\infty}_c(M\leftarrow VB)$ with $\mathrm{supp}(X)\subset V$, and any $\varphi \in \mathcal{U}$ we have $(F+G)_{\varphi}(\vec{X}_{\varphi})=F_{\varphi}(\vec{X}_{\varphi})+G_{\varphi}(\vec{X}_{\varphi})=F_{\varphi}(0)+G_{\varphi}(0)$, so $x \notin \mathrm{supp}(F+G)$. The other follows analogously.
\end{proof}

% Using the notion of Bastiani differentiability we can induce a related differentiability for functionals over $\Gamma^{\infty}(M\leftarrow B)$, in the same spirit as done for mappings between manifolds.

\begin{definition}\label{def_1_func_differentiability}
Let $\mathcal{U}$ be $CO$-open, a functional $F\in \mathcal{F}_c(B,\mathcal{U})$ is differentiable of order $k$ at $\varphi \in \mathcal{U}$ if for all $0 \leq j \leq k$ the functionals $d^jF_{\varphi}[0]:  \otimes^j\left(\Gamma^{\infty}_c(M\leftarrow \varphi^*VB)\right) \rightarrow \mathbb{R}: (\vec{X}_1,\ldots, \vec{X}_j) \mapsto d^jF_{\varphi}[0](\vec{X}_1,\ldots, \vec{X}_j)$ are linear and continuous with 
$$
\begin{aligned}
	d^jF_{\varphi}[u_{\varphi}(\varphi)](\vec{X}_1,\ldots, \vec{X}_j)\doteq &\left.\frac{d^j}{dt_1 \ldots dt_j}\right|_{t_1= \ldots=t_j=0}F_{\varphi}(t_1\vec{X}_1+ \cdots+ t_j \vec{X}_j) \\ &= \left\langle F^{(j)}_{\varphi}[0],\vec{X}_1 \otimes \cdots \otimes \vec{X}_j\right\rangle\ .
\end{aligned}
$$
If $F$ is differentiable of order $k$ at each $\varphi \in \mathcal{U}$ we say that $F$ is differentiable of order $k$ in $\mathcal{U}$. Whenever $F$ is differentiable of order $k$ in $\mathcal{U}$ for all $k \in \mathbb{N}$ we say that $F$ is smooth and denote the set of smooth functionals as $\mathcal{F}_0(B,\mathcal{U})$.
\end{definition}

Since $F$ is Bastiani smooth, $d^jF_{\varphi}[0]$ is continuous, therefore $d^jF_{\varphi}[0] \in \Gamma^{-\infty}(M^j\leftarrow \boxtimes^j \varphi^*VB) $ and by Schwartz kernel theorem, we can represent it as an integral kernel $F^{(j)}_{\varphi}[0]$ which we write as
\begin{equation}\label{eq_1_kernel_notation}
    \left\langle F^{(j)}_{\varphi}[0],\vec{X}_1 \otimes \ldots \otimes \vec{X}_j\right\rangle = \int_{M^j} f^{(j)}_{\varphi}[0]_{i_1 \cdots i_j}(x_1,\ldots,x_j) \vec X_1^{i_1}(x_1) \cdots \vec X_j^{i_j}(x_j) d\mu_g(x_1,\ldots,x_j)\ ,
\end{equation}
where $\vec X_p=\vec X_p^{i_p}\frac{\partial}{\partial y^{i_p}}\Big|_{y=\varphi(x)} \in T_{\varphi}\Gamma^{\infty}(M\leftarrow B)$. We stress that repeated indices denote summation of vector components as usual with Einstein notation.

When $F$ is smooth the condition of Definition \ref{def_1_func_differentiability} is independent from the chart we use to evaluate the B differential: suppose we take charts $(\mathcal{U}_{\varphi}, u_{\varphi})$, $(\mathcal{U}_{\psi}, u_{\psi})$ with $\varphi \in \mathcal{U}_{\psi}$, then by Faà di Bruno's formula
\begin{align}\label{eq_1_gamma_loc_change}
\begin{aligned}
	&d^jF_{\psi}[u_{\psi}(\varphi)](\vec{X}_1,\ldots, \vec{X}_j) 
	\\
	& \quad = \sum_{\pi \in \mathscr{P}(\{1,\ldots,j\})} F_{\varphi}^{|\pi|}[0]\left( d^{\vert I_1 \vert}u_{\varphi \psi}[u_{\psi}(\varphi)]\Big(\bigotimes_{i\in I_1}\vec{X}_{i}\Big), \ldots,  d^{\vert I_{|\pi|} \vert}u_{\varphi \psi}[u_{\psi}(\varphi)]\Big(\bigotimes_{i'\in I_{|\pi|}}\vec{X}_{i'}\Big) \right)\ ,
\end{aligned}
\end{align}
where $\pi$ is a partition of $\{1,\ldots,j\}$ into $|\pi|$ smaller subsets $I_1,\ldots,I_{|\pi|}$ and we denote by $u_{\varphi \psi}$ the transition function $u_{\varphi}\circ u_{\psi}^{-1}$. We immediately see that the right hand side is Bastiani smooth by the smoothness of the transition function, therefore the left hand side ought to be Bastiani smooth as well. Incidentally the same kind of reasoning shows Definition \ref{def_1_func_differentiability} is independent from the ultralocal atlas used for practical calculations. 

Although this is enough to ensure Bastiani differentiability, in the sequel we shall introduce a connection on the bundle $T\Gamma^{\infty}(M\leftarrow B) \rightarrow \Gamma^{\infty}(M\leftarrow B) $ 
so that \eqref{eq_1_gamma_loc_change} can be written as an equivalence between two single terms involving the covariant derivatives. In particular, as explained in \eqref{eq_1_def_connection} and \eqref{eq_1_connection_induced}, we will choose a smooth connection ${\Gamma}$ on the typical fiber $F$ of the bundle $B$, the latter will induce a linear connection $\varphi^*{\Gamma}$ on the vector bundle $M\leftarrow \varphi^*VB$, and, in turn, a connection on $T\Gamma^{\infty}(M\leftarrow B)$
\begin{equation}\label{eq_1_conn}
\begin{aligned}
    \Gamma_{\varphi} &:  \Gamma^{\infty}_c(M\leftarrow \varphi^{*}VB)\times  \Gamma^{\infty}_c(M\leftarrow \varphi^{*}VB)
    \to  \Gamma^{\infty}_c(M\leftarrow \varphi^{*}VB)\\
    &\quad (\vec{X}_{\varphi},\vec{Y}_{\varphi})\mapsto  \Gamma_{\varphi}(\vec{X}_{\varphi},\vec{Y}_{\varphi}),\\
    &\quad\Gamma_{\varphi}(\vec{X},\vec{Y})(x)= \Gamma(\varphi(x))^i_{jk}\vec{X}_{\varphi}^j(x)\vec{Y}_{\varphi}^k(x)\partial_i\big|_{\varphi(x)}\ ;
\end{aligned}
\end{equation}
where $\vec{X}^j\partial_j\big|_{\varphi}$, $\vec{Y}_{\varphi}^k\partial_k\big|_{\varphi}$ are the expressions in local coordinates of $\vec{X}_{\varphi}$, $\vec{Y}_{\varphi}\in \Gamma^{\infty}_c(M\leftarrow \varphi^{*}VB)$. Armed with \eqref{eq_1_conn} we can define the notion of \textit{covariant differential} recursively setting
\begin{align}\label{eq_1_cov_der}
\begin{aligned}
	\nabla^{1}F_{\varphi}[0](\vec{X})  &\doteq dF_{\varphi}(\vec{X}) ,\\
	\nabla^{n}F_{\varphi}[0](\vec{X}_1,\ldots, \vec{X}_n)  &\doteq  F^{(n)}_{\varphi}(\vec{X}_1,\ldots, \vec{X}_n)  \\
	&\quad + \sum_{j=1}^n \frac{1}{n!} \sum_{\sigma \in \mathcal{P}(n)} \nabla^{n-1}F_{\varphi} (\Gamma_{\varphi} (\vec{X}_{\sigma(j)}, \vec{X}_{\sigma(n)}), \vec{X}_{\sigma(1)}, \ldots, \widehat{\vec{X}_{\sigma(j)}}, \ldots \vec{X}_{\sigma(n-1)})\ ,
\end{aligned}
\end{align}
where $\mathcal{P}(n)$ denotes the set of permutations of $n$ elements. In this way we can intrinsically extend properties of iterated derivatives, which are ultralocal chart dependent, to the whole manifold. The price we pay is that, a priori, the property might depend on the connection chosen.

\begin{lemma}\label{lemma_1_support_property}
Let $\mathcal{U}$ be a $CO$-open subset such that for all $\varphi\in \Gamma^{\infty}(M\leftarrow B)$, $u_{\varphi}(\mathcal{U}\cap \mathcal{U}_{\varphi})$ is convex. If $F:\mathcal{U} \rightarrow \mathbb{R}$ a differentiable functional of order one, then
$$
	\mathrm{supp}(F)= \overline{\bigcup_{\varphi \in \mathcal{U}} \mathrm{supp}\left(F^{(1)}_{\varphi}[0]\right)}\ .
$$
\end{lemma}

\begin{proof}
Suppose that $x \in \mathrm{supp}(F)$, then, for all open neighborhoods $V$ of $x$ there is $\varphi \in \mathcal{U}$ and $\vec{X}_{\varphi} \in \Gamma^{\infty}_c(M\leftarrow \varphi^{*}VB)$ with $\mathrm{supp}(\vec{X}) \subset V$ having $F_{\varphi}(\vec{X}_{\varphi}) \neq F_{\varphi}(0)$, using the convexity of $u_{\varphi}(\mathcal{U}\cap \mathcal{U}_{\varphi})$ and the fundamental theorem of calculus we obtain
$$
	F_{\varphi}( \vec{X}_{\varphi}) - F_{\varphi}(0)= \int_0^1 F_{\varphi}^{(1)}[\lambda  \vec{X}_{\varphi}]( \vec{X}_{\varphi}) \mathrm{d}\lambda \neq 0\ .
$$
Thus there is some $\lambda_0 \in (0,1)$ for which the integrand is not zero, setting $\psi=u_{\varphi}^{-1}(\lambda_0  \vec{X}_{\varphi} )$, we obtain 
$$
	dF_{\psi}[0]\left(d^1u_{\varphi \psi}[\lambda_0 \vec{X}_{\varphi}]( \vec{X}_{\varphi})\right)\neq 0\ .
$$
On the other hand, if $x \in \mathrm{supp}\left( F^{(1)}_{\varphi}[0] \right)$ for some $\varphi \in \mathcal{U}$, then there is $\vec{X}_{\varphi} \in \Gamma^{\infty}_c(M\leftarrow \varphi^*VB)$ having $\vec{X}_{\varphi}(x)\neq \vec{0}$ for which $dF_{\varphi}[0](\vec{X}_{\varphi}) \neq 0$, as a result, choosing $\epsilon$ small enough,
$$
	 F_{\varphi}(\epsilon \vec{X}_{\varphi}) = F_{\varphi}(0)+ \int_0^{\epsilon} F_{\varphi}^{(1)}[\lambda \vec{X}_{\varphi}]( \vec{X}_{\varphi}) \mathrm{d}\lambda \neq F_{\varphi}(0)\ .
$$
\end{proof}

%Among smooth functionals we select certain classes:

\begin{definition}\label{def_1_func_classes}
Let $\mathcal{U}$ be $CO$-open. We select the following classes of $\mathcal{F}_c(B,\mathcal{U})$.
\begin{itemize}
    \item[$(i)$] \textbf{Regular Functionals}: the subset of functionals $F \in \mathcal{F}_c(B,\mathcal{U})$ such that for each $\varphi\in \mathcal{U}$, the integral kernel 
    $$
	    \nabla^{k}F_{\varphi}[0](\vec{X}_1,\ldots,\vec{X}_k) = \int_{M^k} \nabla^{k} f_{\varphi}[0](x_1,\ldots,x_k)\vec{X}_1(x_1)\cdots\vec{X}_k(x_k) d\mu_{g}(x_1,\ldots,x_k)
    $$ 
     associated to $\nabla^{k}F_{\varphi}[0]$, has $\nabla^{k} f_{\varphi}[0]\in \Gamma^{\infty}_c\big(M^k\leftarrow \boxtimes^k \big(\varphi^{*}VB' \big)\big)$; we denote this set by $\mathcal{F}_{\mathrm{reg}}(B,\mathcal{U})$.

    \item[$(ii)$] \textbf{Local Functionals}: the subset of functionals $F \in \mathcal{F}_c(B,\mathcal{U})$ such that for each $\varphi\in \mathcal{U}$, $\mathrm{supp}\big(\nabla^{2}F_{\varphi}[0]\big) \subset \triangle_2(M)=\{(x,x)\in M\times M\}$; we denote this set by $\mathcal{F}_{loc}(B,\mathcal{U})$.
    
    \item[$(iii)$] \textbf{Microlocal Functionals}: the subset of functionals $F \in \mathcal{F}_{loc}(B,\mathcal{U})$ such that the integral kernel associated to $\nabla^{1}F_{\varphi}[0]\equiv dF_{\varphi}[0]$ has $f^{(1)}_{\varphi}[0] \in \Gamma^{\infty}\big(M\leftarrow (\varphi^{*}VB)' \big)$ for each $\varphi\in \mathcal{U}$; we denote this set by $\mathcal{F}_{\mu loc}(B,\mathcal{U})$.
\end{itemize}
\end{definition}

Using the Schwartz kernel theorem, we can equivalently define microlocal functionals by requiring $\{ F \in \mathcal{F}_{loc}(B,\mathcal{U}) \space\ \mathrm{:} \  \mathrm{WF}\big(F^{(1)}_{\varphi}[0]\big)=\emptyset \ \forall \varphi \in \mathcal{U}\}$. Other authors also add further requirements: for example, in \cite{BDGR14} microlocal functionals have the additional property that given any $\varphi\in \mathcal{U}$ there exists an open neighborhood $\mathcal{V}\ni \varphi$ in which $f^{(1)}_{\varphi'}[0]\in \Gamma^{\infty}\left(M\leftarrow (\varphi^{*}VB)' \right)$ depends on the $k$th order jet of $\varphi'$ for all $\varphi'\in \mathcal{V}$ and some $k\in \mathbb{N}$. We choose to give a somewhat more general description which however will turn out to be almost equivalent by Proposition \ref{porop_1_muloc_charachterization}. Finally we stress that the definition of local functionals together with Lemma \ref{lemma_1_support_property} shows that that $\mathrm{supp}\big(\nabla^{k}F_{\varphi}[0]\big) \subset \Delta_k(M) \equiv \{(x,\ldots,x) \in M^k: \ x\in M\}$ for each $k\in \mathbb{N}$.\\

As remarked earlier, writing differentials with a connection does yield a  definition which is independent from the ultralocal chart chosen to perform the calculations, however, we have to check that Definition \ref{def_1_func_classes} is independent from the chosen connection. 

\begin{lemma}
    Suppose that $\Phi$, $\widehat{\Phi}$ are two connections on $T\Gamma^{\infty}(M\leftarrow B)$ according to \eqref{eq_1_def_connection}. Then the definition of regular (resp. local, microlocal) functionals do not depend on the chosen connection.
\end{lemma}
\begin{proof}
Denote by $\nabla$, $\widehat{\nabla}$ the covariant derivatives induced by $\Phi$, $\widehat{\Phi}$ respectively. If $F \in \mathcal{F}_c(B,\mathcal{U})$ is local with respect to the second connection,
$$
	\left( \nabla^{2} F_{\varphi} - \widehat{\nabla}^{2}F_{\varphi} \right) [0] (\vec{X}_1,\vec{X}_2) = dF_{\varphi}[0] \left( \Gamma_{\varphi}(\vec{X}_1,\vec{X}_2) - \widehat{\Gamma}_{\varphi}(\vec{X}_1,\vec{X}_2) \right)\ .
$$
Due to linearity of the connection in both arguments, when the two sections $\vec{X}_1$, $\vec{X}_2$ have disjoint support the resulting vector field is identically zero, so that by linearity of $dF_{\varphi}[0](\cdot)$ the expression is zero and locality is preserved. Moreover, since $\nabla^{1}F_{\varphi}[0]\equiv dF_{\varphi}[0]$, we immediately obtain that microlocality is independent as well. Regular functionals do not depend on the connection used to perform calculations either: this is easily seen by induction. The case $k=1$ is trivial, for arbitrary $k$ one simply notes that $\left( \nabla^{k} F_{\varphi} - \widehat{\nabla}^{k}F_{\varphi} \right) [0] (\ldots)$ depends on terms of order $l \leq k-1$ and applies the induction hypothesis. 
\end{proof}

We stress that in particular cases, such as when $B= M \times \mathbb{R}$, $TC^{\infty}(M) \simeq C^{\infty}(M) \times C^{\infty}_c(M)$, we are allowed to choose a trivial connection, in which case the differential and the covariant derivative coincide. It is also possible to formulate Definition \ref{def_1_func_classes} in terms of differentials instead of covariant derivatives, then the above arguments can be used backward to show that regular and local functionals are intrinsic.\\

If we go back to the examples of functionals given earlier we find that \eqref{eq_1_generic_func} does not belong to any class, while \eqref{eq_1_reg_func} is a regular functional that however fails to be local. If $D\subset M$ is a compact subset and $\chi_D$ its characteristic function then 
$$
    \varphi \mapsto \mathcal{L}_{\chi_D,\lambda}(\varphi) \doteq  \int_M \chi_D(x)\lambda(j^r\varphi)(x)\mathrm{d}\mu_g(x)\ .
$$
is a local functional which however, is not microlocal due to the possible singularities localized in the boundary of $D$. Finally we claim that \eqref{eq_1_gen_Lag_1} is a microlocal functional. To see it, let us consider a particular example where $r=1$,
$$
	\mathcal{L}_{f,\lambda}(\varphi)=\int_M f \cdot j^1 \varphi^{*}\lambda= \int_M f(x) \lambda(j^1\varphi)(x) d\mu_g(x)
$$
taking the first derivative and integrating by parts yields
\begin{equation}\label{eq_1_derivative_lag_wavemaps}
    d\mathcal{L}_{f,\lambda,\varphi}[0](\vec{X}_{\varphi})=\int_{M}f(x)\bigg\{\frac{\partial\lambda}{\partial y^i}-d_{\mu}\bigg( \frac{\partial \lambda}{\partial y^i_{\mu}} \bigg)   \bigg\}(x) X_{\varphi}^i(x) \mathrm{d}\mu_g(x)\ ,
\end{equation}
setting 
\begin{equation}
	\lambda_{f,\varphi}^{(1)}[0](x) \doteq f(x)\left\lbrace  \frac{\partial \lambda}{\partial y^i}-d_{\mu}\left( \frac{\partial \lambda}{\partial y^i_{\mu}} \right)   \right\rbrace(x) dy^i \wedge \mathrm{d}\mu_g(x)\ ,
\end{equation}
we see that the integral kernel of the first derivative in $\varphi$ of \eqref{eq_1_gen_Lag_1}, $\lambda_{f,\varphi}^{(1)}[0]$, belongs to $\Gamma^{\infty}\left(M\leftarrow \varphi^{*}VB' \otimes \Lambda_m(M)\right) $. For generic orders $r\neq 1$, multiple integration by parts will yield the desired result, for details on those calculations see \cite[Chapter 6]{fati}. This last example is important because it shows that functionals obtained by integration of pull-backs of $m$-forms $\lambda$ are microlocal. One could ask whether the converse can hold, \textit{i.e.} if all microlocal functionals have this form; this matter will be further analysed in Proposition \ref{porop_1_muloc_charachterization}.

\begin{definition} \label{def_1_additivity}
Let $\mathcal{U}$ be $CO$-open, a functional $F\in \mathcal{F}_c(B,\mathcal{U})$ is called:
\begin{itemize}
\item[$(i)$]$\varphi_0$-additive if for all $\varphi_1, \ \varphi_{-1} \in \mathcal{U}_{\varphi_0}\cap\, \mathcal{U}$ having $\mathrm{supp}_{\varphi_0}(\varphi_1) \cap \mathrm{supp}_{\varphi_0}(\varphi_{-1}) = \emptyset$, setting $\vec{X}_j=u_{\varphi_0}(\varphi_j)$, $j=1,-1$, we have\footnote{For \eqref{eq_1_loc_additivity} to be defined, one should additionally require $\vec{X}_1+\vec{X}_{-1}\in u_{\varphi_0}(\mathcal{U}_{\varphi_0}\cap\mathcal{U})$. This can be done \textit{e.g.} by restriction of the chart open convex $u_{\varphi_0}(\mathcal{U}_{\varphi_0})\subset \Gamma^{\infty}_c\left(M\leftarrow \varphi^{*}VB\right) $ to a subset $u_{\varphi_0}(\mathcal{W}_{\varphi_0})$ which exists by continuity of the additivity in topological vector spaces and satisfies $u_{\varphi_0}(\mathcal{W}_{\varphi_0}) + u_{\varphi_0}(\mathcal{W}_{\varphi_0}) \subset u_{\varphi_0}(\mathcal{U}_{\varphi_0})$.}
\begin{equation}\label{eq_1_loc_additivity}
	F_{\varphi_0}(\vec{X}_1 {+} \vec{X}_{-1})= F_{\varphi_0}(\vec{X}_1) - F_{\varphi_0}(0) + F_{\varphi_0}(\vec{X}_{-1})\ .
\end{equation}
\item[$(ii)$] additive if for all $\varphi_j\in \mathcal{U}$, $j=1,0,-1$, with $\mathrm{supp}_{\varphi_0}(\varphi_1)\cap \mathrm{supp}_{\varphi_0}(\varphi_{-1})=\emptyset$, setting 
$$ \varphi=
\begin{cases}
	\varphi_1 \ \ \ \mathrm{in} \space\ \mathrm{supp}_{\varphi_0}(\varphi_{-1})^c\\
\varphi_{-1} \hspace{0.17cm}  \mathrm{in} \space\ \mathrm{supp}_{\varphi_0}(\varphi_1)^c
\end{cases}
$$
we have 
\begin{equation}\label{eq_1_glob_additivity}
	F(\varphi)=F(\varphi_1)+F(\varphi_0)-F(\varphi_{-1})\ .
\end{equation}
\end{itemize}
\end{definition}
We remark that $(ii)$ is equivalent to the definition of additivity present in \cite{gravbrunetti}. Before the proof of the equivalence of those two relations, we prove a technical lemma.

\begin{lemma} \label{lemma_1_interpolation_of_sections}
Let $\varphi_1$, $\varphi_0$, $\varphi_{-1} \in \Gamma^{\infty}(M\leftarrow B)$ have $\mathrm{supp}_{\varphi_0}(\varphi_1)\cap \mathrm{supp}_{\varphi_0}(\varphi_{-1})=\emptyset$, then there exist $n\in \mathbb{N}$, a finite family of sections  
\begin{equation}\label{eq_1_interpolation_0}
	\left\{\varphi_{(k,l)}\right\}_{k,l \in \{1,\ldots,n\}}
\end{equation}
for which the following conditions holds:
\begin{itemize}
\item[$(a)$] For each $k$, $l\in \mathbb{N}$
\begin{equation}\label{eq_1_intepolation_2}
	\left\{\varphi_{(k-1,l-1)},\varphi_{(k,l-1)},\varphi_{(k-1,l)},\varphi_{(k+1,l)},\varphi_{(k,l+1)},\varphi_{(k+1,l+1)} \right\}\in \mathcal{U}_{\varphi_{(k,l)}}\ ,
\end{equation}
\item[$(b)$]  Moreover for each $k$, $l\in \mathbb{N}$ we can define elements $\vec{X}_{k}$, $\vec{Y}_{l} \in \Gamma^{\infty}_c(M\leftarrow \cdot^{*}VB)$, where 
\begin{equation}\label{eq_+_vector}
	\vec{X}_{k} \doteq \widetilde\exp^{-1}_{\varphi_{(k-1,l)}}\left( \varphi_{(k,l)}\right)\ ,
\end{equation}
\begin{equation}\label{eq_-_vector}
	\vec{Y}_{l} \doteq \widetilde\exp^{-1}_{\varphi_{(k,l-1)}}\left( \varphi_{(k,l)} \right)\ ;
\end{equation}
whose exponential flows generate all the above sections:
$$
	\varphi_1=\widetilde\exp\big(\vec{X}_{n}\big)\circ\cdots \circ \widetilde\exp\big(\vec{X}_{1}\big) \circ \varphi_0\equiv \varphi_{(n,0)}\ ;
$$
$$
	\varphi_{-1}=\widetilde\exp\big(\vec{Y}_{n}\big)\circ\cdots \circ \widetilde\exp\big(\vec{Y}_{1}\big)\circ \varphi_0\equiv \varphi_{(0,n)}\ ;
$$
and
$$
    \varphi= \widetilde\exp\big(\vec{X}_{n}\big)\circ \cdots \circ \widetilde\exp\big(\vec{X}_{1}\big) \circ \widetilde\exp\big(\vec{Y}_{n}\big)\circ \cdots \circ \widetilde\exp\big(\vec{Y}_{1}\big)\circ  \varphi_0\ .
$$
\end{itemize}
\end{lemma}
\begin{proof}
We are considering $\varphi_0$ as a background section, then application of a number of exponential flows of the above fields will generate new sections interpolating between $\varphi_0$ and $\varphi,\varphi_1,\varphi_{-1}$, such that each section in the interpolation procedure has the adjacent sections in the same chart as described in \eqref{eq_1_intepolation_2}. For a pair of generic sections, this is not trivial; however, due to the requirement of mutual compact support between sections, our case is special. Due to the relative compact support of $\varphi_1, \ \varphi_{-1} $ with respect to 
$\varphi_0$, let $K$ be any compact containing the compact subsets  $\mathrm{supp}_{\varphi_0}(\varphi_{1}),\  \mathrm{supp}_{\varphi_0}(\varphi_{-1}) $. Since $B$ is itself a paracompact manifold, it admits an exhaustion by compact subsets and a Riemannian metric compatible with the fibered structure. The exponential mapping $\exp$ of this metric will have a positive injective radius throughout any compact subset of $B$. Thus let $H$ be any compact subset of B containing the bounded subset 
$$
    \Big\{b\in \pi^{-1}(K)\subset B : \sup_{x\in K}d(\varphi_0(x),b) < 2\max\big(\sup_{x\in K} d(\varphi_0(x),\varphi_1(x)), \sup_{x\in K} d(\varphi_0(x),\varphi_{-1}(x))\big)\Big\}\ ,
$$
where $d$ is the distance induced by the metric chosen. Let $\delta>0$ be the injective radius of the metric on the compact $H$. If $r=\max\big(\sup_{x\in K} d(\varphi_0(x),\varphi_1(x)), \sup_{x\in K} d(\varphi_0(x),\varphi_{-1}(x))\big)$ there will be some finite $n \in \mathbb{N}$ such that $n\delta < r <(n+1)\delta$, and thus we can select a finite family of sections $\left\{\varphi_{(k,l)}\right\}_{k,l =1,\ldots,n}$ interpolating between $\varphi_0=\varphi_{(0,0)}$ and $\varphi_1=\varphi_{(n,0)}$, $\varphi_{-1}=\varphi_{(0,n)}$, $\varphi=\varphi_{(n,n)}$ such that 
$$
	(|k-k'|-1)\frac{\delta}{2}+(|l-l'|-1)\frac{\delta}{2}<\sup_{x\in K}d\left(\varphi_{(k,l)}(x),\varphi_{(k',l')}(x)\right)< (|k-k'|)\frac{\delta}{2}+(|l-l'|)\frac{\delta}{2},
$$
This property ensures that we are interpolating in the right direction that is, as $k$ (resp. $l$) grows, new sections are nearer to $\varphi_1$ (resp. $\varphi_{-1})$ and further away from $\varphi_0$. Eventually modifying $\exp$ to $\widetilde\exp$ as done in \eqref{eq_1_tilde_exp}, set
$$
    \vec{X}_{(k,l)} \doteq \widetilde\exp^{-1}_{\varphi_{(k-1,l)}}\left( \varphi_{(k,l)}\right)\ ,
$$
$$
	\vec{Y}_{(k,l)} \doteq \widetilde\exp^{-1}_{\varphi_{(k,l-1)}}\left( \varphi_{(k,l)} \right)\ .
$$
We claim that those are the vector fields interpolating between sections. They are always well defined because, by construction, we choose adjacent sections to be separated by a distance where $\widetilde\exp$ is still a diffeomorphism. Due to the mutual disjoint support of $\varphi_1$ and $\varphi_{-1}$, we can identify $\vec{X}_{(k,l)}$ (resp. $\vec{Y}_{(k,l)}$) with each other $\vec{X}_{(k,l')}$ (resp. $\vec{Y}_{(k',l)}$) for all $\leq l, \ l' \leq n$ (resp. $\leq k, \ k' \leq n$), therefore it is justified to use one index to denote the vector fields as done in \eqref{eq_+_vector} and \eqref{eq_-_vector}. Moreover, for each $k$, $l\in \mathbb{N}$, we have 
$$
    \widetilde\exp\big(\vec{X}_{k}\big)\circ \widetilde\exp\big(\vec{Y}_{l}\big)=\widetilde\exp\big(\vec{Y}_{l}\big)\circ \widetilde\exp\big(\vec{X}_{k}\big)\ ;
$$
which provides a uniquely defined section $\varphi$.
\end{proof}

\begin{proposition} \label{prop_1_additivity}
Let $F \in \mathcal{F}_c(B,\mathcal{U})$ then the following statements are equivalent:
\begin{itemize}
\item[$(i)$] $F$ is additive; 
\item[$(ii)$] $F$ is $\varphi_0$-additive for all $\varphi_0 \in \mathcal{U}$;
\item[$(iii)$] $F \in \mathcal{F}_{\mathrm{loc}}(B,\mathcal{U)}$.
\end{itemize}
\end{proposition}
\begin{proof}
Let us start proving the equivalence between $(i)$ and $(ii)$.\\
%\begin{itemize}
%\item[$(i) \Rightarrow (ii)$]
$(i) \Rightarrow (ii) $ If $\varphi_j \in \mathcal{U}\cap\mathcal{U}_{\varphi_0}$ with $j=1,0,-1$ are as in $(ii)$ above, take $\vec{X}_j$ such that $u_{\varphi_0}^{-1}(\vec{X}_j)=\varphi_j$. Writing \eqref{eq_1_glob_additivity} in terms of $F_{\varphi_0}$ yields \eqref{eq_1_loc_additivity}.\\
$(ii) \Rightarrow (i) $ Let us take sections $\varphi_j$ with $j=1,0,-1$ such that $\mathrm{supp}_{\varphi_0}(\varphi_1)\cap \mathrm{supp}_{\varphi_0}(\varphi_{-1})=\emptyset$, combining Lemma \ref{lemma_1_interpolation_of_sections} with $\varphi$-additivity for each section, yields
$$
\begin{aligned}
	F(\varphi)& = F_{\varphi_{(n-1,n-1)}}\big(\vec{X}_n+ \vec{Y}_n\big)=
	F_{\varphi_{(n-1,n-1)}}\big(\vec{X}_n\big)+F_{\varphi_{(n-1,n-1)}}\big( \vec{Y}_n\big)-F_{\varphi_{(n-1,n-1)}}(0) \\ 
    &= F\left(\varphi_{(n,n-2)}\right)+\textcolor{red}{F\left(\varphi_{(n-1,n-1)}\right)}-F\left(\varphi_{(n-1,n-2)}\right)	\\ 
    &\quad+ F\left(\varphi_{(n-2,n)}\right)+F\left(\varphi_{(n-1,n-1)}\right)-F\left(\varphi_{(n-2,n-1)}\right)-	\textcolor{red}{F\left(\varphi_{(n-1,n-1)}\right)}\\ 
     &= F\left(\varphi_{(n,n-2)}\right)-\textcolor{blue}{F\left(\varphi_{(n-1,n-2)}\right)} + F\left(\varphi_{(n-2,n)}\right)+\textcolor{violet}{F\left(\varphi_{(n-2,n-1)}\right)}+\textcolor{blue}{F\left(\varphi_{(n-1,n-2)}\right)}\\
    &\quad-F\left(\varphi_{(n-2,n-2)}\right)-\textcolor{violet}{F\left(\varphi_{(n-2,n-1)}\right)} \\
    &=  F\left(\varphi_{(n,n-2)}\right)+F\left(\varphi_{(n-2,n)}\right)-F\left(\varphi_{(n-2,n-2)}\right)\ .
\end{aligned}
$$
%\end{itemize}
Repeating the above argument an extra $n-2$ times we arrive at
$$
	F(\varphi)= F\left(\varphi_{(n,0)}\right)+F\left(\varphi_{(0,n)}\right)-F\left(\varphi_{(0,0)}\right)\equiv F(\varphi_1)+F(\varphi_{-1})-F(\varphi_0) \	.
$$
We conclude proving that $(ii)$ and $(iii)$ are equivalent.\\
%\begin{itemize}
$(iii) \Rightarrow (ii)\ $Take $\varphi_j$, $\vec{X}_j \doteq u_{\varphi_0}(\varphi_j) $, with $j=1,0,-1$ as in $(i)$ Definition \ref{def_1_additivity}. Then 
$$
\begin{aligned}
	 & F_{\varphi_0}(\vec{X}_1 +\vec{X}_{-1})-F_{\varphi_0}(\vec{X}_1)+F_{\varphi_0}(0)-F_{\varphi_0}(\vec{X}_{-1})=
	 \int_{0}^1 \frac{d}{dt}\left( F_{\varphi_0}(\vec{X}_1+ t\vec{X}_{-1})- F_{\varphi_0}(t\vec{X}_{-1})\right)dt  \\& =
	 \int_{0}^1  \frac{d}{dt}  \left( \int_0^1 \frac{d}{dh}  F_{\varphi_0}(h\vec{X}_1 + t\vec{X}_{-1})dh\right)dt= \int_{0}^1  \int_{0}^1 d^2F{\varphi_0}[h\vec{X}_1+ t\vec{X}_{-1}](\vec{X}_1,\vec{X}_{-1})dhdt\ .
\end{aligned}
$$
By locality we have that $\mathrm{supp}\big(d^2F{\varphi_0}\big) \subset \triangle_2 M$, however, $\mathrm{supp}(\vec{X}_1) \cap \mathrm{supp}(\vec{X}_{-1})=\emptyset$ implying that the integrand on the right hand side of the above equation is identically zero.\\
$(ii) \Rightarrow (iii)\ $Fix any $\varphi_0 \in \mathcal{U}$, consider two vector fields $\vec{X}_1$, $\vec{X}_{-1}\in \Gamma^{\infty}\big(M\leftarrow  \varphi_0^{*}VB\big)$ such that $\mathrm{supp}(\vec{X}_1) \cap\mathrm{supp}(\vec{X}_{-1})=\emptyset$. Additionally, let $\varphi_{j}\doteq u_{\varphi_0}^{-1}(\vec{X}_j)$ for $j=1,-1$. As a result, $\mathrm{supp}_{\varphi}(\varphi_1) \cap \mathrm{supp}_{\varphi}(\varphi_{-1})=\emptyset$. By direct computation we get 
$$
\begin{aligned}
	 F_{\varphi_0}^{(2)}[0](\vec{X}_1,\vec{X}_{-1}) & =\left.\frac{d^2}{dt_1dt_2} \right|_{t_1=t_2=0} F_{\varphi_0}(t_1\vec{X}_1 + t_2\vec{X}_{-1}) \\ & = \left.\frac{d^2}{dt_1dt_2} \right|_{t_1=t_2=0} \left( F_{\varphi_0}(t_1\vec{X}_1)-F_{\varphi_0}(0)+F_{\varphi_0}( t_2\vec{X}_{-1}) \right) \equiv 0\ .
\end{aligned}
$$
\end{proof}

As a result, we have shown that locality and additivity are consistent concepts in a broader generality than done in \cite{acftstructure}. Of course, additivity strongly relates to Bogoliubov's formula for S-matrices, therefore we expect that as long as we can formulate different concepts consistently, those must be equivalent formulations. We also mention that when the exponential map used to construct ultralocal charts is a global diffeomorphism, then additivity and $\varphi$ additivity becomes trivially equivalent since the chart can be enlarged to $\mathcal{V}_{\varphi}\equiv \{\psi\in \Gamma^{\infty}(M\leftarrow B): \mathrm{supp}_{\varphi}(\psi) \subset M \space\ \mathrm{is} \space\ \mathrm{compact} \}$.

\begin{remark}
   The ultralocal notion of additivity \textit{i.e.} $(i)$ in Definition \ref{def_1_additivity} is independent from the ultralocal chart: suppose that $F$ is $\varphi_0$-additive in $\lbrace \mathcal{U}_{\varphi_0},u_{\varphi_0} \rbrace$, take another chart $\lbrace \mathcal{U}'_{\varphi_0},u'_{\varphi_0} \rbrace$ such that $\mathcal{U}'_{\varphi_0} \cap \mathcal{U}_{\varphi}\neq \emptyset$, set $\vec{X}_j=u_{\varphi_0}(\varphi_j)$, $\vec{Y}_j=u_{\varphi_0}'(\varphi_j)$, for $j=1,-1$, we have\footnote{In the subsequent calculations we can assume, without loss of generality, that $\vec{Y}_1+\vec{Y}_{-1}\in u'_{\varphi_0}(\mathcal{U}'_{\varphi_0})$, for if this is not the case we can use an argument involving Lemma \ref{lemma_1_interpolation_of_sections} to make this expression meaningful.}
$$
\begin{aligned}
	F\circ u'^{-1}_{\varphi_0}(\vec{Y}_1+\vec{Y}_{-1}) & =F\circ u^{-1}_{\varphi_0}\circ u_{\varphi_0}\circ u'^{-1}_{\varphi_0} (\vec{Y}_1+\vec{Y}_{-1})=F\circ u^{-1}_{\varphi_0}(\vec{X}_1+\vec{X}_{-1})\\
	& = F\circ u^{-1}_{\varphi_0}(\vec{X}_1 ) - F\circ u^{-1}_{\varphi_0}(0) + F\circ u^{-1}_{\varphi_0}(\vec{X}_{-1})\\ 
	& = F\circ u'^{-1}_{\varphi_0}(\vec{Y}_1 ) - F\circ u'^{-1}_{\varphi_0}(0) + F\circ u'^{-1}_{\varphi_0}(\vec{Y}_{-1})
\end{aligned}
$$
where $u_{\varphi_0}\circ u'^{-1}_{\varphi_0} (\vec{Y}_1+\vec{Y}_{-1})=\vec{X}_1+\vec{X}_{-1}$ is due to the fact that the two vector fields have mutually disjoint supports. We then see that $\varphi_0$-additivity does not depend upon the chosen chart. 
\end{remark}
% The ultralocal notion of additivity \textit{i.e.} $(i)$ in Definition \ref{def_1_additivity} is independent from the ultralocal chart used: suppose that $F$ is $\varphi_0$-additive in $\lbrace \mathcal{U}_{\varphi_0},u_{\varphi_0} \rbrace$, take another chart $\lbrace \mathcal{U}'_{\varphi_0},u'_{\varphi_0} \rbrace$ such that $\mathcal{U}'_{\varphi_0} \cap \mathcal{U}_{\varphi}\neq \emptyset$, set $\vec{X}_j=u_{\varphi_0}(\varphi_j)$, $\vec{Y}_j=u_{\varphi_0}'(\varphi_j)$, for $j=1,-1$, we have\footnote{In the subsequent calculations we can assume, without loss of generality, that $\vec{Y}_1+\vec{Y}_{-1}\in u'_{\varphi_0}(\mathcal{U}'_{\varphi_0})$, for if this is not the case we can use an argument involving Lemma \ref{lemma_1_interpolation_of_sections} to make this expression meaningful.}
% $$
% \begin{aligned}
% 	F\circ u'^{-1}_{\varphi_0}(\vec{Y}_1+\vec{Y}_{-1}) & =F\circ u^{-1}_{\varphi_0}\circ u_{\varphi_0}\circ u'^{-1}_{\varphi_0} (\vec{Y}_1+\vec{Y}_{-1})=F\circ u^{-1}_{\varphi_0}(\vec{X}_1+\vec{X}_{-1})\\
% 	& = F\circ u^{-1}_{\varphi_0}(\vec{X}_1 ) - F\circ u^{-1}_{\varphi_0}(0) + F\circ u^{-1}_{\varphi_0}(\vec{X}_{-1})\\ 
% 	& = F\circ u'^{-1}_{\varphi_0}(\vec{Y}_1 ) - F\circ u'^{-1}_{\varphi_0}(0) + F\circ u'^{-1}_{\varphi_0}(\vec{Y}_{-1})
% \end{aligned}
% $$
% where $u_{\varphi_0}\circ u'^{-1}_{\varphi_0} (\vec{Y}_1+\vec{Y}_{-1})=\vec{X}_1+\vec{X}_{-1}$ is due to the fact that the two vector fields have mutually disjoint supports. We then see that $\varphi_0$-additivity does not depend upon the chosen chart.\\

\begin{remark}\label{rmk_1_extension_of_functionals}
    Functionals are generally defined in $CO$-open subsets instead of more general Whitney open sets since we can always extend them 
 from $WO^{\infty}$-open domains to a $CO$-open subsets. To wit, suppose $F:\mathcal{U} \to \mathbb{R}$ is a smooth functional with compact spacetime support, then consider the function $\chi\in C^{\infty}_c(M)$ having $0\leq \chi \leq 1$, $\chi\equiv 1$ inside $K \supset \mathrm{supp}(F)$ and, given any $\varphi_0 \in \mathcal{U}$, define
    \begin{equation}\label{eq_1_CO_extension_mapping}
        i_{\chi,\varphi_0}: \Gamma^{\infty}(M \leftarrow B) \to \mathcal{V}_{\varphi_0}, \quad \widetilde{\psi} \mapsto \psi.
    \end{equation}
    The mapping can be constructed using with Lemma \ref{lemma_1_interpolation_of_sections}: starting with $\varphi_0$, we can modify the latter inside $K$ so that $\widetilde\exp(\Vec{X}_n)\circ \cdots \circ \widetilde\exp(\Vec{X}_1)\varphi_0|_K=\widetilde{\psi}|_K$, then setting $\psi=\widetilde\exp(\chi\Vec{X}_n)\circ \cdots \circ \widetilde\exp(\chi\Vec{X}_1)\varphi_0$, we have that $\psi=\widetilde{\psi}$ inside $\mathrm{supp}(F)$, $\psi=\varphi_0$ outside $K$. $i_{\chi,\varphi_0}$ is a continuous and smooth mapping, and when $\mathcal{U}$ is a $\mathrm{WO}^{\infty}$ open neighborhood of $\varphi_0$, $i_{\chi,,\varphi_0}^{-1}(\mathcal{U})$ is a $CO$-open. Then we can seamlessly extend the functional $F$ to $\widetilde{F}: \widetilde{\mathcal{U}}\equiv \cup_{\varphi_0 \in \mathcal{U}}i_{\chi,\varphi_0}^{-1}(\mathcal{U})\to \mathbb{R}$. The functional will remain smooth, and all its derivatives will not be affected by the cutoff function $\chi$.
\end{remark}

We now give the characterization of microlocality; we will find that, contrary to additivity, the latter representation will be limited to a chart domain, in the sense that the functional can be represented as an integral provided we shrink its domain to a chart, this representation however will not be independent from the chosen chart.

We recall that a mapping $T:U\subset X\to Y$ between locally convex spaces is \textit{locally bornological} if for any $x\in X$ there is a neighborhood $V\ni x$ contained in $U$ such that $T|_V$ maps bounded subsets of $V$ into bounded subsets of $Y$. From this definition, it follows a technical result:

\begin{lemma}\label{lemma_1_loc_bornology}
    Let $\mathcal{U}\subset \Gamma^{\infty}(M\leftarrow B)$ be $CO$-open, then a smooth, spacetime compactly supported functional $F$ satisfies: $dF:\mathcal{U}_{\varphi} \to \Gamma^{\infty}_c(M\leftarrow \varphi^*VB'\otimes \Lambda_m(M))$ is Bastiani smooth if and only if it is locally bornological. 
\end{lemma}
\begin{proof}[Sketch of a proof]
The proof of this result when $B=M\times \mathbb R$ can be found in \cite[Lemma 2.6]{acftstructure}. Since we are allowing for a bit more generality (\textit{i.e.} we are considering distributional sections of the bundle $\varphi^*VB\to M$), we will just highlight the minor changes to the argument presented in the aforementioned Lemma 2.6. From Bastiani smoothness of $F$ we can see $F^{(1)}$ as a Bastiani smooth mapping $\mathcal{U}_{\varphi}\to \Gamma^{-\infty}_c(M\leftarrow \varphi^*VB)$, combining the support property of $F$ with the fact that it is microlocal, we obtain that $F^{(1)}$ can be viewed as a mapping $T:\mathcal{U}_{\varphi}\to \Gamma^{\infty}_K(M\leftarrow \varphi^*VB\otimes \Lambda_m(M))$ for some compact subset $K$ of $M$. To prove the lemma it is enough to show that $T$ is Bastiani smooth if and only if it is locally bornological. Necessity follows from the fact that composing $T$ with (Bastiani smooth) chart mappings yields a Bastiani smooth, hence continuous, mapping $\Gamma^{\infty}_c(M\leftarrow \varphi^*VB)\to\Gamma^{\infty}_K(M\leftarrow \varphi^*VB\otimes \Lambda_m(M))$. Both spaces are semi-Montel, \textit{i.e.} every bounded subset is relatively compact. Thus let ${W}\subset \overline{{W}}\subset {V} \subset u_{\varphi}(\mathcal{U}_{\varphi})$, with $\overline{W}$ bounded, then $T(\overline{W})$ is compact, hence bounded, due to continuity of $F$. As a result $T|_{V}$ is locally bornological. The sufficiency condition is guaranteed if, 
\begin{itemize}
    \item given any compact subset $K\subset M$ and a finite cover of $K$ of the form ${\psi_{\alpha}^{-1}(Q)}$ where $Q\subset \mathbb{R}^n$ is the open $m$-cube and $\psi_{\alpha}$ are the charts of $M$;
    \item given any partition of unity ${f_{\alpha}}$ of the above cover, the induced conveniently smooth\footnote{We recall, as in \cite[Definition 3.6]{kriegl1997convenient}, that a mapping is conveniently smooth if it maps smooth curves in $\mathcal{U}_{\varphi}$ to smooth curves of $\mathcal{D}\big(E|_Q\big)$.} mappings $T_{\alpha}: \mathcal{U}_{\varphi} \to \mathcal{E}'\big(Q;V\big), \quad \varphi \mapsto (\psi_{\alpha})_*\big(f_{\alpha}T(\varphi)\big)$, where $V \simeq \mathbb{R}^d$ is the typical fiber of the bundle $\varphi^*VB'\otimes \Lambda_m(M) \to M$, can be equivalently seen as conveniently smooth mappings $T_{\alpha}: \mathcal{U}_{\varphi} \to \mathcal{D}\big(Q;V\big)$.
\end{itemize}

By hypothesis we know that $T_{\alpha}: \mathcal{U}_{\varphi} \to \mathcal{D}\big(Q;V\big)$ is locally bornological, to complete the proof we note that each projection mapping $\pi_i:V \to \mathbb{R}$, $i=1,\ldots ,n$, induces continuous and also bounded mappings $(\pi_i)_{*} : \mathcal{D}\big(Q;V\big) \to \mathcal{D}(Q)$. Thus, by claims (i), (ii) in the proof of Lemma 2.6 in \cite{acftstructure}, each $(\pi_i)_{*} T_{\alpha} :\mathcal{U}_{\varphi} \to \mathcal{D}(Q) $, which remains locally bornological, maps smooth curves of $\mathcal{U}_{\varphi}$ to smooth curves of $\mathcal{D}(Q)$ implying that $T_{\alpha}: \mathcal{U}_{\varphi} \to \mathcal{D}\big(Q;V\big)$ is convenient smooth as well for any $\alpha$.
\end{proof}

\begin{proposition}\label{porop_1_muloc_charachterization}
Let $\mathcal{U}\subset\Gamma^{\infty}(M\leftarrow B)$ be $CO$-open and $F \in \mathcal{F}_{\mu loc}(B,\mathcal{U})$ , then $f^{(1)}:\mathcal{U}\subset \Gamma^{\infty}(M\leftarrow B)\ni \varphi \mapsto f^{(1)}_{\varphi}[0]\in \Gamma^{\infty}_c(M\leftarrow \varphi^{*}VB'\otimes \Lambda_m(M))$ is locally bornological if and only if for each $\mathcal{U}_{\varphi_0}\subset \mathcal{U}$ there is a $m$-form $\lambda_{F,\varphi_0}\equiv \lambda_{F,0}$ with $\lambda_{F,0}(j^{\infty}\varphi)$ having compact support for all $\varphi \in \mathcal{U}_{\varphi_0}$ such that
\begin{equation}\label{eq_1_muloc_formula}
	F(\varphi)=F(\varphi_0)+ \int_M (j^r_x\varphi)^{*}\lambda_{F,0}\ .
\end{equation}
\end{proposition}
\begin{proof}
Suppose $ F(\varphi)=F(\varphi_0)+ \int_M (j^{r}\varphi)^{*}\lambda_{F,0}$ for all $\varphi\in \mathcal{U}_{\varphi_0}$, we evaluate $dF_{\varphi}[0]$ and find that its integral kernel may always be recast in the form
$$
	f^{(1)}_{\varphi}[0](x)=e_i[\lambda,\varphi_0](j^{2r}_x\varphi)dy^i\otimes d\mu_g(x)
$$
where $e_i[\lambda,F,\varphi_0]dy^i\otimes d\mu_g:\Gamma^{\infty}(M\leftarrow B) \to \Gamma^{\infty}_c(M\leftarrow \varphi^{*}VB'\otimes \Lambda_m(M))$ are the Euler-Lagrange equations associated to $\lambda_{F,\varphi_0}$ evaluated at some field configuration. Using the ultralocal differential structure of the source space, and keeping in mind that $e_i[\lambda,F]: \mathcal{U}_{\varphi_0}\to \Gamma^{\infty}_c(M\leftarrow \varphi^{*}VB'\otimes \Lambda_m(M))$ is an operator of bounded order, we can apply Theorem \ref{thm_A_Bastiani smooth_pushforward} and by Lemma \ref{lemma_1_loc_bornology} to get that
$f^{(1)}$ is locally bornological.

Conversely suppose that $f^{(1)}$ is locally bornological, by Lemma \ref{lemma_1_loc_bornology} it is Bastiani smooth as a mapping as well. Fix $\varphi_0\in \mathcal{U}$ and call $\vec{X}=u_{\varphi_0}(\varphi)$, by microlocality combined with Schwartz kernel theorem,
$$
	F(\varphi)-F(\varphi_0)=F_{\varphi_0}(\vec{X})-F_{\varphi_0}(0)= \int_0^1 dF_{\varphi_0}[t\vec{X}](\vec{X}) dt=\int_0^1 dt \int_M f^{(1)}_{\varphi_0}[t\vec{X}](\vec{X})\ .
$$
Applying the Fubini-Tonelli theorem to exchange the integrals in the above relation yields our candidate for $j^{r}\varphi^{*}\lambda_{F,0}$: the $m$-form $x \mapsto \theta[\varphi](x)\equiv \int_0^1 f^{(1)}_{\varphi_0}[t\vec{X}](\vec{X})(x)dt$. We have to show that this element depends at most on $j^{r}_x\varphi$. Notice that, a priori, $\theta[\varphi](x)$ might not depend on $j^r_x\varphi$, however we can say that if $\varphi_1$, $\varphi_2$ possess the same germ at $x\in M$, then $\theta[\varphi_1]\vert_{V}=\theta[\varphi_2]\vert_{V}$ for a suitably small neighborhood $V$ of $x$. To see this, set $\vec{X}_1=u_{\varphi_0}(\varphi_1)$, $\vec{X}_2=u_{\varphi_0}(\varphi_2)$, by ultralocality of the charts, they agree in a suitably small neighborhood $V'\subset V$ of $x$, moreover 
$$
\begin{aligned}
	 \theta[\varphi_1](x)-\theta[\varphi_2](x)& =\int_0^1 dt \left(f^{(1)}_{\varphi_0}[t\vec{X}_{1}](\vec{X}_{1})(x)-f^{(1)}_{\varphi_0}[t\vec{X}_{2}](\vec{X}_{2})(x)\right)\\ 
	&=\int_0^1 dt \left(f^{(1)}_{\varphi_0}[t\vec{X}_{1}]_i(x)\vec{X}_{1}^i(x)-f^{(1)}_{\varphi_0}[t\vec{X}_{2}]_i(x)\vec{X}_{2}^i(x)\right)\\ 
	&= \int_0^1 dt \int_0^1 dh f^{(2)}_{\varphi_0}[t\vec{X}_{2}+th\vec{X}_{1}-th\vec{X}_{2}]_{ij}(x)(t\vec{X}_{1}-t\vec{X}_{2})^i(x) \vec{X}_1^j(x)\ ; 
\end{aligned}
$$
where in the last equality we used locality of $F$ and linearity of the derivative. The last line of the above equation identically vanishes in $V'$ due the support properties of $f^{(2)}_{\varphi_0}$ and the fact that $\vec{X}_1\vert_{V'}=\vec{X}_2\vert_{V'}$. Therefore $\theta[\varphi](x)\in \varphi^{*}VB'\otimes \Lambda_m(M)$ depends at most on $\mathrm{germ}_x(\varphi)$. 

We wish to apply Peetre-Slov\'ak's theorem %\footnote{See \textit{e.g.}\ Theorem 19.10 pp. 180 in \cite{kolar2013natural} for the precise statement of the theorem. Also notice that due to the fact that our functionals have compact spacetime support, it is enough to show a weaker regularity condition then the one employed in the original statement of the theorem, thus the regularity hypothesis therein can be weakened by testing it with a family of mappings with variation happening on compact set. More precisely an operator $D:C^{\infty}(M,X)\to C^{\infty}(M,Y)$ is weakly regular if, given any compactly supported variation, \textit{i.e.} a jointly smooth family of mapping $\Phi:\mathbb{R}\times M\to N: t \mapsto \varphi_t \in C^{\infty}(M,X)$ such that there is a compact $K\subset M$ with $\varphi_t\vert_{M\backslash K}$ constant in the variation parameter $t\in \mathbb{R}$; the family $D(\varphi_t)$ is again a compactly supported variation.} 
to $\theta$; the germ dependence hypothesis has been verified above, so one has to show that $\theta$ is also weakly regular, that is, if $\mathbb{R}\times M \ni (t,x ) \mapsto \varphi_t(x) \in B$ is compactly supported variation, then $(t,x)\mapsto\theta[\varphi_t](x)$ is again a compactly supported variation. $\theta$ is a compactly supported form, thus it maps compactly supported variations into compactly supported variations. Moreover it is Bastiani smooth since 
$$
    \varphi \mapsto F(\varphi_0)+ \int_M\theta[\varphi]d\mu_g(x)=F(\varphi)
$$
is an observable. Then we can conclude by Lemma \ref{lemma_A_locality&bastiani_implies_w-reg}.

% Suppose therefore that $\mathbb{R}\ni t\to \varphi_t \in \Gamma^{\infty}(M\leftarrow B)$ is the mapping described above, then $\varphi_t$ is a smooth curve for the smooth structure of $\Gamma^{\infty}(M\leftarrow B)$ (see Proposition \ref{prop_1_smooth_Bastiani_curve}). By Bastiani smoothness of $\theta:\varphi\mapsto \theta[\varphi]$,
% $$
% 	t\mapsto \theta[\varphi_t] \in \Gamma^{\infty}_c(M\leftarrow \varphi^{*}_tVB'\otimes \Lambda_m(M))
% $$
% is smooth. For each $t\in \mathbb{R}$, we can define a canonical fibered isomorphism $\varphi^{*}_tVB'\otimes \Lambda_m(M)\simeq \varphi_0^{*}VB'\otimes \Lambda_m(M) $ which, by Theorem \ref{thm_A_Bastiani smooth_pushforward}, induces a smooth diffeomorphism $\Gamma^{\infty}_c(M\leftarrow \varphi^{*}_tVB'\otimes \Lambda_m(M)) \simeq \Gamma^{\infty}_c(M\leftarrow \varphi^{*}VB'\otimes \Lambda_m(M)) $. As a result we have that $\theta$ is a smooth mapping
% $$
% 	t\mapsto \theta[\varphi_t] \in \Gamma^{\infty}_c(M\leftarrow \varphi^{*}_0VB'\otimes \Lambda_m(M)).
% $$
% We infer by $(ii)$ in Proposition \ref{prop_A_conv_sections_of_vector_bndl} and the fact that $\theta[\varphi_t]$ is a constant variation outside a compact subset of $M$, that $\theta[\varphi]: M\times \mathbb{R}: (t,x) \mapsto \theta[\varphi_t](x) \in B $ is smooth. Note that the topology limit-Fréchet topology on $\Gamma^{\infty}_c(M\leftarrow \varphi^{*}_0VB'\otimes \Lambda_m(M))$ makes it a topological embedding in $\Gamma^{\infty}(M\leftarrow \varphi^{*}_0VB'\otimes \Lambda_m(M))$ with the $\mathrm{WO}^{\infty}$ topology (see Theorem \ref{thm_1_Gamma_c_TVS}).
Finally, applying the Peetre-Slovak we deduce that for each neighborhood of $x$ there exists $r=r(x,\varphi_0) \in \mathbb{N}$, an open neighborhood $U^r\subset J^rB$ of $j^r\varphi_0$ and a mapping $\lambda_{F,0}:J^{r}B\supset U^r \to \Gamma^{\infty}_c(M\leftarrow \varphi^{*}_0VB'\otimes \Lambda_m(M))$ such that $\lambda_{F,0}(j^r_x\varphi)=\theta[\varphi](x)$ for each $\varphi$ with $j^r\varphi\in U^r$. Due to compactness of $\mathrm{supp}(\theta)$ we can take the order $r$ to be independent from the point $x$ on $M$; then
$$
	F(\varphi)=F(\varphi_0)+\int_M \lambda_{F,0}(j^r_x\varphi)\ .
$$
\end{proof}

\begin{remark}    
One could also strengthen the hypothesis of Proposition \ref{porop_1_muloc_charachterization}, for example, by requiring that for every $k \in \mathbb{N} $, $R>0$ and every $\varphi_0 \in\mathcal{U}$, 
whenever
$$
    \sup_{\substack{x\in K \\ j\leq k}} \Big\vert\nabla^j\big(u_{\varphi_0}(\varphi_1)-u_{\varphi_0}(\varphi_2) \big)(x) \Big\vert \equiv \big|\big|u_{\varphi_0}(\varphi_1)-u_{\varphi_0}(\varphi_2) \big|\big|_{K,k+r} < R\ ,
$$
there exists a positive constant $C$ for which
\begin{equation}\label{eq_1_lipschitz_bound}
    \big|\big|f_{\varphi_0}^{(1)}[\varphi_1]-f_{\varphi_0}^{(1)}[\varphi_2]\big|\big|_{K,k} \leq C \big|\big|u_{\varphi_0}(\varphi_1)-u_{\varphi_0}(\varphi_2) \big|\big|_{K,k+r}\ .
\end{equation}
This condition implies that $f^{(1)}$ is locally bornological, moreover it is sufficient (see Lemma 1 together with (B) of Theorem 1 in \cite{zajtz1999nonlinear}) to imply that the order $r$ from Proposition \ref{porop_1_muloc_charachterization} is independent from the section $\varphi$ and thus globally constant.
\end{remark}

As mentioned above, this characterization is limited to the ultralocal chart chosen: given charts $\lbrace \mathcal{U}_{\varphi_j},u_{\varphi_j} \rbrace$, $j=1,2$ such that $\varphi \in \mathcal{U}_{\varphi_1} \cap \mathcal{U}_{\varphi_2}$, let $\vec{X}_j = u_{\varphi_j}(\varphi)$ and suppose $F$ satisfies the hypothesis of Proposition \ref{porop_1_muloc_charachterization}, then according to \eqref{eq_1_muloc_formula}
$$
	F(\varphi)=  F(\varphi_1)+ \int_M (j^{r_1}\varphi)^{*}\lambda_{F,1}= F(\varphi_2)+ \int_M (j^{r_2}\varphi)^{*}\lambda_{F,2}\ .
$$
Assuming $r_1=r_2\equiv r$ and using the same argument as in the proof of Proposition \ref{porop_1_muloc_charachterization}\ ,
$$
\begin{aligned}
	(j^{r}\varphi)^{*}\lambda_{F,1}(x) &=  \int_0^1 f_{\varphi_1}^{(1)}[t\vec{X}_1]\left(\vec{X}_1\right)(x) \mathrm{d}t \\ 
	& = \int_0^1 f_{\varphi_2}^{(1)}[u_{\varphi_1 \varphi_2}(t\vec{X}_1)]\left( u_{\varphi_1 \varphi_2}(\vec{X}_1)\right)(x) \mathrm{d}t	\\ 
	& = \int_0^1 f_{\varphi_2}^{(1)}[u_{\varphi_1 \varphi_2}(t\vec{X}_1)]\left(\vec{X}_2\right)(x) \mathrm{d}t\ ,
\end{aligned}
$$
whereas
$$
	(j^{r}\varphi)^{*}\lambda_{F,2}(x)= \int_0^1 f_{\varphi_2}^{(1)}[t\vec{X}_2]\left(\vec{X}_2\right)(x) \mathrm{d}t\ ,
$$
we therefore see that the lack of linearity of the transition mapping $u_{\varphi_1\varphi_2}$, namely $u_{\varphi_1 \varphi_2}(t\vec{X}_1)\neq t u_{\varphi_1 \varphi_2}(\vec{X}_1)= t\vec{X}_2 $, does not allow us to conclude $(j^{\infty}\varphi)^{*}\lambda_{F,1}(x)= (j^{\infty}\varphi)^{*}\lambda_{F,2}(x)$.
\begin{remark}
    
We give another argument that prevents the characterization in Proposition \ref{porop_1_muloc_charachterization} from being intrinsic. This relies on the variational sequence\footnote{A complete exposition can be found in \textit{e.g.} \cite{krupka}.}: a cohomological sequence of forms over $J^rB$ for some finite $r\in \mathbb{N}$
\begin{center}
\hspace{-1.45cm}
\begin{tikzcd}
		 & 0 \arrow[r]\arrow[d, phantom, ""{coordinate, name=Z}] & \mathbb{R}  \arrow[r,"E_0"] & \Omega^1(J^rB)/\sim \arrow[r,"E_1"]& \ldots \arrow[r] & \Omega^m(J^rB)/\sim \arrow[dllll,"E_m",rounded corners,to path={ -- ([xshift=2ex]\tikztostart.east)|- (Z) [near end]\tikztonodes
-| ([xshift=-2ex]\tikztotarget.west)-- (\tikztotarget)}] \\
   & \Omega^{m+1}(J^rB)/\sim \arrow[r,"E_{m+1}"] & \Omega^{m+2}_{h}(B)/\sim  \arrow[r]
		  & \ldots \arrow[r] &\Omega^{N}(B) \arrow[r,"E_{N}"]  &  0 \ ,
\end{tikzcd}
\end{center}
where each element of the sequence is the quotient of the space of $p$-forms in $J^rB$ modulo some relation that cancel the exact forms (in the sense of the de-Rham differential on the manifold $J^rB$) and accounts for integration by parts when the order is greater then $m=\mathrm{dim}(M)$. In particular the $m$th differential $E_m$ is the operator which, given a horizontal $m$-form, calculates its \textit{Euler-Lagrange form} and the $(m+1)$th differential $E_{m+1}$ is the operator which associates to each Euler-Lagrange form its \textit{Helmholtz-Sonin form}. By the Poincar\'e lemma, if $\sigma\in \mathrm{ker}(E_{m+1})$, there exists a local chart $(V^r,\psi^r)$ in $J^rB$ and a horizontal $m$-form $\lambda \in \Omega^m(V^r)$ having $E_m\vert_{V^r}(\lambda)=\sigma\vert_{V^r}$. Establishing whether this condition holds globally amounts to determine, given some Euler-Lagrange equations satisfying certain conditions (the associated Helmholtz-Sonin form vanishes), whether they arise from the variation of some Lagrangian. The variational sequence implies that a sufficient condition is the vanishing of the $m$-th cohomology group, \textit{i.e.} whenever topological obstructions are not present. 

Now, if Proposition \ref{porop_1_muloc_charachterization} could somehow reproduce \eqref{eq_1_muloc_formula} for each $\varphi \in \mathcal{U}$ with an integral over the same $m$-form $\lambda_F$, we would have found a way to circumvent the topological obstructions that ruin the exactness of the variational sequence. Furthermore in the derivation of $\lambda_{F,0}$ we did not even require that the associated Euler-Lagrange equations had vanishing Helmholtz-Sonin form, but instead a Bastiani smoothness requirement that, due to Proposition \ref{thm_A_Bastiani smooth_pushforward}, will always be met by integral functionals constructed from smooth geometric objects. It appears therefore that the two approaches bears some kind of duality: given a representative $F_{\varphi}^{(1)}[0] \in \Omega^{m+1}(J^rB)/ \hspace{-0.1cm}\sim$ one can, on the one hand, get a \textit{ultralocal chart dependent} Lagrangian via Proposition \ref{porop_1_muloc_charachterization} \textit{i.e.} a global $m$-form on the bundle $J^rB$ which however describe the functional only when evaluated in a small neighborhood of a background section $\varphi_0$; on the other hand, prioritize \textit{ultralocal chart independence}, therefore having a local $m$-form defined on the bundle $J^r(\pi^{-1}U)$ for some open subset $U$ of $M$, which however describes the functional for all sections of $\Gamma^{\infty}(U\leftarrow \pi^{-1}(U))$.
\end{remark}

\begin{proposition}\label{prop_1_variational_sqn}
Let $\mathcal{U}\subset\Gamma^{\infty}(M\leftarrow B)$ be $CO$-open, $F \in \mathcal{F}_{\mu loc}(B,\mathcal{U})$ satisfying the hypothesis of Proposition \ref{porop_1_muloc_charachterization} and the bound \eqref{eq_1_lipschitz_bound}. Fix $\varphi \in \mathcal{U}_{\varphi_0}$ and suppose that
$$
	F(\varphi)=F(\varphi_0)+ \int_M (j^{r}\varphi)^{*}\lambda_{F,0}=F(\varphi_0)+ \int_M (j^{r}\varphi)^{*}\lambda'_{F,0}\ .
$$
then $\lambda_{F,0}-\lambda'_{F,0}=d_h\theta$ for some $\theta \in \Omega^{m-1}_{\mathrm{hor}}(J^rB)$ if and only if the $m$-th de Rham cohomology group $H_{\mathrm{dR}}^m(B)=0$. In particular the above condition is verified whenever $B$ is a vector bundle with finite dimensional fiber and $M$ is orientable non-compact and connected.
\end{proposition}
\begin{proof}
Using the notation introduced above for the variational sequence, we have that $E_m(\lambda_{F,0})= f^{(1)}_{\varphi_0}[0] = E_m(\lambda'_{F,0})$; thus their difference is zero and $\lambda_{f,\psi}-\lambda'_{f,\psi}\in \Omega^m(J^rB)/\sim$. The latter cohomology group is isomorphic, by the abstract de Rham Theorem, to $H_{\mathrm{dR}}^m(B)$, therefore $\lambda_{F,0}-\lambda'_{F,0}=d_h\theta$ if and only if $H_{\mathrm{dR}}^m(B)=0$. When $B$ is a vector bundle over $M$ its de Rham cohomology groups are isomorphic to those of $M$. Finally, if $M$ is orientable, non-compact and connected it has $H^m_{\mathrm{dR}}(M)=0$. The latter claim can be established using Poincar\'e duality, \textit{i.e.} $H^m_{\mathrm{dR}}(M)\simeq H^0_{\mathrm{dR},c}(M)$. If $M$ is non-compact and connected, \textit{e.g.} when it is globally hyperbolic, there are no compactly supported functions with vanishing differential other then the zero function, so the $m$-th cohomology group is zero.
\end{proof}

We shall conclude this section by introducing generalized Lagrangians, which, as the name suggests, will be used to select a dynamic on $\Gamma^{\infty}(M\leftarrow B)$. We stress that unlike the usual notion of Lagrangian - either a horizontal $m$-form over $J^rB$ or a morphism $J^rB\to \Lambda_m(M)$ - this definition will allow us to evaluate the action functional as an integral over the whole manifold (instead of a compact subset) as well as avoid the presence of singularities which might arise when the cut-off function $f$ is the characteristic function $\chi_K$ of a compact subset $K\subset M$.

\begin{definition}\label{def_1_gen_lag}
Let $\mathcal{U}\subset\Gamma^{\infty}(M\leftarrow B)$ be $CO$-open. A generalized Lagrangian $\mathcal{L}$ on $\mathcal{U}$ is a mapping $$\mathcal{L}:C^{\infty}_c(M) \rightarrow \mathcal{F}_{c}(B,\mathcal{U}),$$ such that 
\begin{itemize}
\item[$(i)$] $\mathrm{supp}(\mathcal{L}(f))\subseteq \mathrm{supp}(f)$ and $\mathcal{L}(f)$ is Bastiani smooth for all $f\in C^{\infty}_c(M)$, 
\item[$(ii)$] for each $f_1$, $f_2$, $f_3 \in C^{\infty}_c(M)$ with $\mathrm{supp}(f_1)\cap \mathrm{supp}(f_3)=\emptyset$, $$\mathcal{L}(f_1+f_2+f_3)=\mathcal{L}(f_1+f_2)-\mathcal{L}(f_2)+\mathcal{L}(f_2+f_3).$$
\end{itemize}
\end{definition}

Given the properties of the above Definition we immediately get:  

\begin{proposition} \label{prop_1_gen_lag_1}
Let $\mathcal{U}\subset\Gamma^{\infty}(M\leftarrow B)$ be $CO$-open, $\mathcal{L}$ a generalized Lagrangian on $\mathcal{U}$. Then
\begin{itemize}
\item[$(i)$] $\mathrm{supp}(\mathcal{L}(f+f_0)-\mathcal{L}(f_0)) \subseteq \mathrm{supp}(f)$ for all $f$, $f_0 \in C^{\infty}_c(M)$,
\item[$(ii)$] for all $f\in C^{\infty}_c(M)$, $\mathcal{L}(f)$ is a local functional.
\end{itemize}
\end{proposition}

The proof is the same as those of \cite[Lemma 3.1 and 3.2]{acftstructure}.

Combining the linearity of $C^{\infty}_c(M)$, property $(ii)$ Definition \ref{def_1_gen_lag} and Proposition \ref{prop_1_gen_lag_1} we obtain that each generalized Lagrangian can be written as a suitable sum of arbitrarily small supported generalized Lagrangians. To see it, fix $\epsilon>0$ and consider $\mathcal{L}(f)$. By compactness $\mathrm{supp}(f)$ admits a finite open cover of balls, $\lbrace B_i \rbrace_{i \in I}$ of radius $\epsilon$ such that none of the open balls is completely contained in the union of the others. Let $\lbrace g_i \rbrace_{i \in I}$ be a partition of unity subordinate to the above cover of $\mathrm{supp}(f)$, set $f_i \doteq g_i \cdot f$. Using $(ii)$ Definition \ref{def_1_gen_lag}
$$
	\mathcal{L}(f)=\mathcal{L}\left(\sum_i f_i \right)=\sum_{J\subset I} c_J \mathcal{L}\left(\sum_{j \in J} f_j \right),
$$
where $J\subset I$ contains the indices of all balls $B_i$ having non empty intersection with a fixed ball (the latter included), and $c_J= \pm 1$ are suitable coefficients determined by the application of $(ii)$ Definition \ref{def_1_gen_lag}. By construction each index $J$ has at most two elements and $\mathrm{supp}(\sum_{j \in J} f_j)$ is contained at most in a ball of radius $2\epsilon$. We have thus split $\mathcal{L}(f)$ as a sum of generalized Lagrangians with arbitrarily small supports.

\begin{definition}\label{def_1_euler_der}
Let $\mathcal{U}\subset\Gamma^{\infty}(M\leftarrow B)$ be $CO$-open, $\mathcal{L}$ a generalized Lagrangian on $\mathcal{U}$. The $k$-th Euler-Lagrange derivative of $\mathcal{L}$ in $\varphi \in \mathcal{U}$ along $ (\vec{X}_1,\ldots, \vec{X}_k) \in \Gamma^{\infty}_c(M\leftarrow \varphi^*VB)^k $ is 
\begin{equation}\label{eq_1_variational_der}
	\delta^{(k)} \mathcal{L}(1)_{\varphi}[0](\vec{X}_1,\ldots, \vec{X}_k)\doteq \left.\frac{d^k}{dt_1 \ldots dt_k}\right|_{t_1= \ldots=t_k=0} \mathcal{L}(f)_{\varphi}[0] (t_1 \vec{X}_1+ \ldots+ t_k\vec{X}_k)
\end{equation}
where $\left. f \right|_{K} \equiv 1$ on a suitable compact $K$ containing all compacts $\mathrm{supp}(\vec{X}_i)$.
\end{definition}

%One can see how the compact supports of the $\Gamma^{\infty}$-tangent vectors, \textit{i.e.} the sections of $\Gamma^{\infty}_c(M\leftarrow \varphi^*VB)^k $ allow us to perform an adiabatic limit and consider the cutoff function $f$ to be identically $1$ throughout $M$.

From now on we will assume that generalized Lagrangian used are microlocal, \textit{i.e.} $\mathcal{L}(f)\in \mathcal{F}_{\mu loc}(B,\mathcal{U})$ for each $f\in C^{\infty}_c(M)$; this means that the first Euler-Lagrange derivative can be written as
\begin{equation}\label{eq_1_euler_der}
	\delta^{(1)} \mathcal{L}(1)_{\varphi}[0](\vec{X})=\int_M E(\mathcal{L})_{\varphi}[0](\vec{X})\ ,
\end{equation}
where by microlocality $E(\mathcal{L})_{\varphi}[0]\in\Gamma^{\infty}_c(M\leftarrow \varphi^{*}VB' \otimes \Lambda_m(M))$.

A generalized Lagrangian $\mathcal{L}$ is \textit{trivial} whenever $\mathrm{supp}\big(\mathcal{L}(f)\big)\subset \mathrm{supp}(df)$ for each $f\in C^{\infty}_c(M)$. Triviality induces an equivalence relation on the space of generalized Lagrangians, namely two $\mathcal{L}_1$, $\mathcal{L}_{2}$ are equivalent whenever their difference is trivial. We can show that if two Lagrangians $\mathcal{L}_1$, $\mathcal{L}_2$ are equivalent then they end up producing the same first variation \eqref{eq_1_euler_der}. For instance, suppose that $\mathcal{L}_1(f)-\mathcal{L}_2(f) =\Delta\mathcal{L}(f)$ with $\Delta\mathcal{L}(f)$ a trivial generalized Lagrangian for each $f\in C^{\infty}_c(M)$. To evaluate $\delta^{(1)}\Delta\mathcal{L}(1)_{\varphi}[0](\vec{X})$ one has to choose some $f$ which is identically $1$ in a neighborhood of $\mathrm{supp}(\vec{X})$, however, by $(i)$ Definition \ref{def_1_gen_lag} $\mathrm{supp}\big(\Delta\mathcal{L}(f)\big)\subset \mathrm{supp}(df)\cap \mathrm{supp}(\vec{X})=\emptyset$. By Lemma \ref{lemma_1_support_property}, $E(\Delta\mathcal{L})_{\varphi}[0](\vec{X})=0$ and
$$
\begin{aligned}
     \delta^{(1)} \mathcal{L}_1(f)_{\varphi}[0](\vec{X}) &=  \delta^{(1)} \mathcal{L}_2(1)_{\varphi}[0](\vec{X})+\delta^{(1)} \Delta\mathcal{L}(1)_{\varphi}[0](\vec{X})\\
     &= \int_M E(\mathcal{L}_2)_{\varphi}[0](\vec{X})+\int_M E(\Delta\mathcal{L})_{\varphi}[0](\vec{X})\\
     &=\delta^{(1)} \mathcal{L}_2(1)_{\varphi}[0](\vec{X})\ .
\end{aligned}
$$
Finally we compare our generalized action functional with the \textit{standard} action which is generally used in classical field theory (see \textit{e.g.} \cite{fati}, \cite{krupka}). One generally introduce the \textit{standard geometric} Lagrangian $\lambda$ of order $r$, as a bundle morphism
\begin{center}
\begin{tikzcd}
		 & J^rB \arrow[r,"\lambda"] \arrow[d,"\pi^r"]  & \Lambda_m(M) \arrow[d, "\rho"] \\
		 & M\arrow[r,equal] & M 
\end{tikzcd}
\end{center}	
between $(J^rB,\pi^r,M)$ and $(\Lambda_m(M),\tau_{\Lambda}, M, \wedge^m T_{\cdot}^*M )$, where the latter is the vector bundle whose sections are $m$-forms. Two Lagrangian morphisms $\lambda_1$, $\lambda_2$ are equivalent whenever their difference is an exact form. Its associated \textit{standard geometric} action functional will therefore be 
\begin{equation}\label{eq_classical_lagrangian}
	\mathcal{A}_D(\varphi)= \int_M \chi_D(x) \lambda(j^r_x\varphi)
\end{equation}
where $\lambda$ an element of the equivalence class of Lagrangian morphisms, $D$ is a compact region of $M$ whose boundary $\partial D$ is an orientable $(m-1)$-manifold and $\chi_D$ its characteristic function. One could be tempted to draw a parallel with a generalized Lagrangian by considering the mapping
\begin{equation}\label{eq_classical_gen_lag}
	\chi_D \mapsto \mathcal{A}(\chi_D)= \int_M \chi_D(x) \lambda(j^r_x\varphi)\ .
\end{equation}
However \eqref{eq_classical_gen_lag} differs from Definition \ref{def_1_gen_lag} in the singular character of the cutoff function. Indeed the functional $\mathcal{A}_D \in \mathcal{F}_{loc}(B,\mathcal{U})$ for each choice of compact $D$ but it is \textit{never} microlocal, for the integral kernel of $\mathcal{A}(\chi_D)^{(1)}_{\varphi}[0]$ has always singularities localized in $\partial D$. This is a severe problem when attempting to calculate the Peierls bracket for local functionals, a way out is to extend this bracket to less regular functionals (see Definition \ref{def_1_WF_mucaus}) maintaining the closure of the operation (see Theorem \ref{thm_1_peierls_closedness}); however, we cannot outright extend the bracket to all local functionals. Therefore, in order to accommodate those less regular functionals such as \eqref{eq_classical_lagrangian}, one would need to place severe restrictions on the possible compact subsets $D$ which cut off possible integration divergences. Restrictions, however, are not consistent with the derivation of Euler-Lagrange equations by the usual variation technique which requires to consider \textit{each} $D\subset M$ compact. 

Of course, given a Lagrangian morphism $\lambda$ of order $r$ we can always define a generalized microlocal Lagrangian as a microlocal-valued distribution, \textit{i.e.} 

\begin{equation*}
    C^{\infty}_c(M)\times\mathcal{U} \ni (f,\varphi) \mapsto \mathcal{L}(f)(\varphi)= \int_M f(x) \lambda(j^r_x\varphi) \ .
\end{equation*}

%%%%%%%%%%%%%%%%%%%%%%%%%%%%%% PEIERLS %%%%%%%%%%%%%%%%%%%%%%%%%%%%

\section{The Peierls bracket}\label{section_peierls_bracket}

% Heuristically speaking the Peierls bracket is a duality relating two observables, $F$, $G$, that accounts for the effect of the (antisymmetric) influence of $F$ on $G$ when the latter is perturbed around a solution of certain equations. 
We will define the Peierls bracket using the linearized field equations associated with the second derivative of a generalized microlocal Lagrangian, which are assumed to be normally hyperbolic. We start by reviewing some basic notions from the theory of normally hyperbolic (NH) operators.\\

Let $E, \ F \to M$ be vector bundles over a globally hyperbolic Lorentzian manifold $M$, a \textit{differential operator} $D:\Gamma^{\infty}(M\leftarrow E) \rightarrow \Gamma^{\infty}(M\leftarrow F)$ is of \textit{second order} if locally, it can be written as second order partial differential operator. $D$ is called \textit{normally hyperbolic} if its principal symbol $\sigma_D\in \Gamma^{\infty}(S^2 TM \otimes E^{*} \otimes F)$ can be written as 

\begin{equation}\label{eq_1_normal_hyp}
    \sigma_2(D) =\frac{1}{2} g^{-1}\otimes \mathrm{id}_E
\end{equation}
where $g$ is the Lorentzian metric on $M$.

\begin{theorem} \label{thm_1_properties_of_Green_functions}
Let $E\to M$ be a vector bundle over a globally hyperbolic Lorentzian manifold $M$ and let $D$ be a normally hyperbolic differential operator. Then $D$ admits global Green operators 
$$
    G_M^{\pm}: \Gamma^{\infty}_c(M\leftarrow E) \rightarrow \Gamma^{\infty}(M\leftarrow E)
$$
and their causal propagator $G = G^+-G^-$, satisfying the following properties:
\begin{itemize}
\item[$(i)$] \textbf{Continuity.} $G^{\pm}$, $G$ are continuous linear operators where $\Gamma^{\infty}_c(M\leftarrow E)$ has the LF-topology and $\Gamma^{\infty}(M\leftarrow E) $ the standard Fréchet topology. Moreover, each operator admit a continuous and linear extension to the spaces $\Gamma^{-\infty}_{\pm}(M\leftarrow E)$ topological dual to $\Gamma^{\infty}_{\mp}(M\leftarrow E)=\lbrace \sigma \in \Gamma^{\infty}(M\leftarrow E) :\space\ \forall x\in M \space\ \mathrm{supp}(\vec{u})\cap J^{\mp}(x) $ is compact $\rbrace$. 
\item[$(ii)$] \textbf{Support Properties.} 
$$
	\mathrm{supp}(G^{\pm} u) \subset J^{\pm}(\mathrm{supp}(u))
$$
for all $u \in \Gamma^{-\infty}_{\pm}(M\leftarrow E)$.
\item[$(iii)$] \textbf{Cauchy problem.} For every $v \in \Gamma_c^{-\infty}(M\leftarrow E)$ and spacelike Cauchy hypersurface $i:\Sigma \hookrightarrow M$ with $\mathrm{supp}(v) \subseteq I^{\pm}(\Sigma)$ and all $u_0$, $\dot{u}_0$ $\in \Gamma_c^{\infty}(M\leftarrow i^{*}E)$ there is a unique $u_{\pm} \in \Gamma^{-\infty}(M\leftarrow E)$ with
\begin{gather*}
	Du_{\pm}=v ,\\
	\mathrm{supp}(u_{\pm}) \subseteq J^{\pm}\big(\mathrm{supp}(u_0) \cup \mathrm{supp}(\dot{u_0})\big) \cup J^{\pm}(\mathrm{supp}(v)).
\end{gather*}
The section $u_{\pm}$ also depends continuously on $v$, $u_0$ and $\dot{u}_0$.
\item[$(iv)$] \textbf{Wave Front Sets.} 
\begin{align*}
	\mathrm{WF}(G^{\pm}) &= \{ (x,x;\xi,-\xi)\in \dot{T}^{*}(M \times M) \, \, \mathrm{with} \ (x;\xi) \in \dot{T}^{*}M  \} \\ &\quad \cup \{ (x_1,x_2;\xi_1,\xi_2) \in \dot{T}^{*}(M\times M) \, \mathrm{with} \, (x_1,x_2;\xi_1,-\xi_2) \in \mathrm{BiCh}^{g}_{\gamma} \}, \\
	\mathrm{WF}(G_M) &= \{ (x_1,x_2;\xi_1,\xi_2)\in \dot{T}^{*}(M\times M)\, \, \mathrm{with} \,\,  (x_1,x_2;\xi_1,-\xi_2) \in \mathrm{BiCh}^{g}_{\gamma} \}\ ;
\end{align*}
where $ \mathrm{BiCh}^{g}_{\gamma}$ is the bicharacteristic strip of the lightlike geodesic $\gamma$, \textit{i.e.} the set of points $(x_1,x_2;\xi_1,\xi_2)$ such that there is an interval $[0,\Lambda] \subset \mathbb{R}$ for which $(x_1,\xi_1)=(\gamma(0), g^{\flat}\dot{\gamma}(0))$ and $(x_2,-\xi_2)=(\gamma(\Lambda), g^{\flat}\dot{\gamma}(\Lambda))$.
\item[$(v)$] \textbf{Propagation of Singularities.} Given $u \in \Gamma^{-\infty}_{\pm}(M \leftarrow E)$, $(x,\xi) \in \mathrm{WF}(G^{\pm}(u))$ if either $(x,\xi) \in \mathrm{WF}(u)$ or there is a lightlike geodesic $\gamma$ and some $(y;\eta) \in \mathrm{WF}(u)$ such that $(x,y;\xi,-\eta) \in \mathrm{BiCh}^{g}_{\gamma}$. Similarly, $(x,\xi) \in \mathrm{WF}(G(u))$ if there is a lightlike geodesic $\gamma$ and some $(y,\eta) \in \mathrm{\mathrm{WF}}(u)$ such that $(x,y;\xi,-\eta) \in \mathrm{BiCh}^{g}_{\gamma}$.
\end{itemize}
\end{theorem}

We recall that in the above theorem the notation $\Gamma^{-\infty}(M\leftarrow E)$ denotes distributional sections of the vector bundle $E$, \textit{i.e.} continuous linear mappings $\Gamma^{\infty}_c(M\leftarrow E) \rightarrow \mathbb{C} $, where the first space is endowed with the usual limit Fréchet topology. We also recall that the wave front set at $x\in M$ of a distributional section ${u}$ of a vector bundle $E$ of rank $k$ is calculated as follows: fix a trivialization $(U_{\alpha},t_{\alpha})$ on $E$, ${u}$ is then locally represented by $k$ distributions ${u}^i\in \mathcal{D}'(U_{\alpha})$, each of which will have its own wave front set. 
Then we set
\begin{equation}\label{eq_1_def_WF_sections}
\mathrm{WF}({u})\doteq \bigcup_{i=1}^k \mathrm{WF}\big({u}^i\big).
\end{equation}

We can view the second Euler-Lagrange derivative (\textit{c.f.} Definition \ref{def_1_euler_der}) of a microlocal generalized Lagrangian $\mathcal{L}$ as a mapping
\begin{equation}\label{eq_1_2-th_derivative_gen_Lag}
	\delta^{(1)} E(\mathcal{L})_{\varphi}[0]: \Gamma^{\infty}_c(M\leftarrow \varphi^{*}VB) \rightarrow  \Gamma^{\infty}_c\big(M\leftarrow \varphi^{*}VB' \otimes \Lambda_m(M)\big)\ .
\end{equation}
Equivalently, we can say that the \textit{linearized field equations} around $\varphi$ induce a differential operator of second order. If we fix an auxiliary metric $h$ on the standard fiber of $B$ and use the Lorentzian metric $g$ of $M$ together with its Hodge isomorphism $*_g$, we define 
\begin{equation}\label{eq_1_diffop_1}
	D_{\varphi} \doteq (\varphi^*h)^{\sharp} \circ (\mathrm{id}_{\varphi^{*}VB' } \otimes *_g)\circ  \delta^{(1)} E(\mathcal{L})_{\varphi}[0]: \Gamma^{\infty}(M\leftarrow \varphi^{*}VB ) \rightarrow \Gamma^{\infty}(M\leftarrow \varphi^{*}VB)\ .
\end{equation}
% For each $\varphi$ we can see that $D_{\varphi}$ is a differential operator and determine the principal symbol. 

\begin{remark}
Notice that if $D_{\varphi}$ is a differential operator induced by the linearized equations of some microlocal generalized Lagrangian $\mathcal{L}$ as in \eqref{eq_1_diffop_1}, then its principal symbol is independent form the section $\varphi$ chosen.

Indeed, by \eqref{eq_1_gamma_loc_change},
\begin{align*}
	\delta^{(1)} E(\mathcal{L})_{\psi}[0](\vec{X}_1,\vec{X_2}) & = \delta^{(1)} E(\mathcal{L})_{\varphi}[0]\left(d^1u_{\varphi \psi}[u_{\psi}(\varphi)](\vec{X}_1),d^1u_{\varphi \psi}[u_{\psi}(\varphi)](\vec{X_2})\right) \\
	&\quad + E(\mathcal{L})_{\varphi}[0]\left(d^2u_{\varphi \psi}[u_{\psi}(\varphi)](\vec{X}_1,\vec{X_2})\right)\ .
\end{align*}
While the second piece modifies the expression of the differential operator, it does not alter its principal symbol since the local form of $d^2u_{\varphi \psi}[u_{\psi}(\varphi)](\vec{X}_1,\vec{X_2})$, due to the ultralocal nature of the chart mapping, does not yield extra derivatives. We therefore conclude that if we use a generalized Lagrangian $\mathcal{L}$ whose linearized equations differential operator, $D_\varphi$, is normally hyperbolic for some $\varphi_0 \in \mathcal{U}$, then it is normally hyperbolic (with the same principal symbol) for all $\varphi\in \mathcal{U}_{\varphi_0}$.
\end{remark}

Let us give a more specific example on how to calculate the principal symbol from a microlocal generalized Lagrangian. Recalling formula \eqref{eq_1_derivative_lag_wavemaps} with $\lambda:J^1(M\times N) \to \Lambda^m(M)$, we have
\begin{align*}
	d^2\mathcal{L}_{f,\lambda,\varphi}[0](\vec{X}_1,\vec{X}_2)&=\int_{M}f(x)\bigg\{ \frac{\partial^2 \lambda}{\partial y^i \partial y^j} \vec{X}^i_1 \vec{X}^j_2 + \frac{\partial^2 \lambda}{\partial y^i_{\mu} \partial y^j} d_{\mu}\big(\vec{X}^i_1\big) \vec{X}^j_2\\ 
	& \quad+ \frac{\partial^2 \lambda}{\partial y^i_{\mu} \partial y^j} \vec{X}^i_1 d_{\mu}\big(\vec{X}^j_2\big)+ \frac{\partial^2 \lambda}{\partial y^i_{\mu} \partial y^j_{\nu}} d_{\mu}\big(\vec{X}^i_1 \big) d_{\nu}\big(\vec{X}^j_2\big)  \bigg\}(x){d}\mu_g(x)\ ,
\end{align*}
where $d_{\mu}$ is the horizontal differential on jet bundles. The key ingredient for the principal symbol is the quantity $m^{\mu \nu}_{ij}\doteq \frac{\partial^2 \lambda}{\partial y^i_{\mu} \partial y^j_{\nu}}$. Applying the transformations to get the differential operator of linearized field equations, as in \eqref{eq_1_diffop_1}, to the above quantity yields the principal symbol
\begin{equation}\label{eq_1_princ_symb}
		\sigma_2(D_{\varphi})=h^{ij}m^{\mu \nu}_{jk} \otimes \partial_{\mu}\vee \partial_{\nu} \otimes e_i\otimes e^j\ .
\end{equation}
In case this quantity satisfies the condition of \eqref{eq_1_normal_hyp} we can conclude that the operator is normally hyperbolic. There are also other notions of hyperbolicity, for instance see \cite{christodoulou}, where the hyperbolicity condition is strictly weaker than the one employed here.

From now on we shall assume that our microlocal Lagrangian produces always normally hyperbolic linearized equations\footnote{We shall show in Section \ref{section_wave_maps} that a physically relevant example is the Lagrangian functional of wave maps}. We can invoke the results of Theorem \ref{thm_1_properties_of_Green_functions}: to each $\varphi$ in the domain of $\mathcal{L}$ we associate operators
\begin{equation}\label{eq_1_green_1}
	G_{\varphi}^{\pm} : \Gamma^{\infty}_c(M\leftarrow \varphi^{*}VB) \rightarrow \Gamma^{\infty}(M\leftarrow \varphi^{*}VB)\ .
\end{equation}
satisfying properties $(i) - (v)$ in the above theorem. Given $\vec X \in \Gamma^{\infty}_c(M\leftarrow \varphi^*VB)$ we can view $G^{\pm}(\vec X)$ as a mapping
$$
	\mathcal{U}\ni \varphi \mapsto G^{\pm}_{\varphi}(\vec X) \in \Gamma^{\infty}(M\leftarrow \varphi^{*}VB)\ ,
$$
where the target space is endowed with the Fréchet topology. Then we can show the following result:

\begin{lemma}\label{lemma_1_deriv_green}
Let $\gamma:\mathbb{R}\to \mathcal{U}\subset \Gamma^{\infty}(M\leftarrow B)$ be a smooth curve, then for each fixed $\vec{X}\in \Gamma_c^{\infty}(M \leftarrow \varphi^{*}VB)$ the mapping $\mathbb R \ni t\mapsto G^{\pm}_{\gamma(t)}(\vec{X})\in \Gamma^{\infty}(M\leftarrow \varphi^{*}VB)$ is Bastiani smooth. In particular, we have
\begin{equation}\label{eq_1_der_G+-}
	dG^{\pm }_{\varphi}(\vec{X})=\lim_{t\rightarrow 0}\frac{1}{t}\left( G^{\pm}_{u_{\varphi}^{-1}(t\vec{X})}-G^{\pm}_{\varphi} \right) = -G^{\pm}_{\varphi} \circ D_{\varphi}^{(1)}(\vec{X})\circ G^{\pm}_{\varphi}\ ,
\end{equation}
where $\mathcal{U}\ni\varphi\mapsto D_{\varphi}(\vec{X})\in \Gamma^{\infty}(M\leftarrow \varphi^{*}VB)$ is the mapping induced by \eqref{eq_1_diffop_1}.
\end{lemma}
\begin{proof}
We just show the claim for the retarded propagator since for the advanced one the result follows in complete analogy. We have to evaluate
$$
	\lim_{t\to 0} \frac{1}{t}\Big({G}^+_{\gamma(t)}(\vec{X})
 -{G}_{\gamma(0)}^+(\vec{X})\Big)\ .
$$
The mapping $D(\vec{X}) \circ \gamma :  \mathbb R \ni t \mapsto D_{\gamma(t)}(\vec{X}) \in \Gamma^{\infty}(M\leftarrow \varphi^{*}VB)$ is smooth, therefore 
$$
\begin{aligned}
    \lim_{t \to 0}\frac{1}{t}D_{\gamma(t)} \Big({G}^+_{\gamma(t)}-{G}_{\gamma(0)}^+\Big)(\vec{X}) & = \lim_{t\to 0} \frac{1}{t}\Big(D_{\gamma(t)}{G}^+_{\gamma(t)}-D_{\gamma(t)}{G}_{\gamma(0)}^+\Big)(\vec{X})\\
	&=  \lim_{t\to 0} \frac{1}{t}\Big(\mathrm{id}_{\Gamma^{\infty}(M\leftarrow \gamma(0)^{*}VB)}-D_{\gamma(t)}{G}_{\gamma(0)}^+\Big) (\vec{X})\\ 
	&=\lim_{t\to 0} \frac{1}{t}\Big(D_{\gamma(0)}{G}^+_{\gamma(0)}-D_{\gamma(t)}{G}_{\gamma(0)}^+\Big)(\vec{X})\\
	&=\lim_{t\to 0} \frac{1}{t}\Big(D_{\gamma(0)}-D_{\gamma(t)}\Big)\big({G}_{\gamma(0)}^+(\vec{X})\big)\\
	&=- D_{\gamma(0)}^{(1)}(\dot{\gamma}(0))\circ G^{+}_{\gamma(0)}\big(\vec{X}\big)\ .
\end{aligned}
$$
Since $\gamma$ is a smooth curve in $\Gamma^{\infty}(M\leftarrow B)$, given any interval $[-\epsilon,\epsilon]$ with $\epsilon>0$, there is a compact subset $K_{\epsilon}$ of $M$ for which $\gamma(t)(x)$ is constant in $t$ on $M\backslash K_{\epsilon}$, then differential operator $D_{\gamma(0)}^{(1)}\neq 0$ only inside $K_{\epsilon}$, therefore the quantity $ \big(D_{\gamma(0)}^{(1)}(\dot{\gamma}(0)) \circ G^{\pm}_{\gamma(0)}\big)(\vec{X}) $ has compact support for any $\vec{X}\in \Gamma^{\infty}(M\leftarrow \gamma(0)^{*}VB)$. Finally, using continuity of $G^+_{\varphi}$ and $D_{\varphi}$, we can write
$$
	\lim_{t\to 0} \frac{1}{t}\Big({G}^+_{\gamma(t)}-{G}_{\gamma(0)}^+\Big)=- G^{+}_{\gamma(0)} \circ D_{\gamma(0)}^{(1)}(\dot{\gamma}(0))\circ G^{+}_{\gamma(0)}\ .
$$
Iterating the above expression one can show that all iterated derivatives of ${G}_{\varphi}^+$ exists and are continuous, thus showing smoothness.
\end{proof}

\begin{remark}\label{remark_1_bastiani_smoothness}
    Notice that Lemma \ref{lemma_1_deriv_green} establishes that the operators $G^{\pm}_{\cdot}(\Vec{X}): \Gamma^{\infty}(M\leftarrow B) \to \Gamma^{\infty}(M\leftarrow \varphi^{*}VB) $ are conveniently smooth, however since the target space is Fréchet and the fact that both the convenient and Bastiani smooth structure of $\Gamma^{\infty}(M\leftarrow B) $ possesses the same smooth curves (\textit{c.f.} \cite[Remark 42.2]{kriegl1997convenient}) implies that $G^{\pm}_{\cdot}(\Vec{X})$ is Bastiani smooth as well. 
\end{remark}

Similarly, for the causal propagator we find 
\begin{align}\label{eq_1_der_G}
	dG_{\varphi}(\vec{X})\doteq \lim_{t\rightarrow 0}\frac{1}{t}\left( G_{u_{\varphi}^{-1}(t\vec{X})}-G_{\varphi} \right) = -G_{\varphi} \circ D_{\varphi}^{(1)}(\vec{X})\circ G^{+}_{\varphi}-G^{-}_{\varphi} \circ D_{\varphi}^{(1)}(\vec{X})\circ G_{\varphi}\ .
\end{align}
Given the Green's functions $G_{\varphi}^{\pm}$, set 
\begin{equation}\label{eq_1_Gvar_+-}
	\mathcal{G}^{\pm}_{\varphi} \doteq G^{\pm}_{\varphi} \circ (\varphi^{*}h)^{\sharp} \circ (id_{\left(\varphi^{*}VB\right)'} \otimes *_g) : \Gamma^{\infty}_c(M\leftarrow \left(\varphi^{*}VB\right)' \otimes  \Lambda_m(M)) \rightarrow \Gamma^{\infty}(M\leftarrow \varphi^{*}VB)\ ,
\end{equation}
\begin{equation}\label{eq_1_Gvar}
	\mathcal{G}_{\varphi} \doteq G_{\varphi} \circ (\varphi^{*}h)^{\sharp} \circ (id_{\left(\varphi^{*}VB\right)'} \otimes *_g) : \Gamma^{\infty}_c(M\leftarrow \left(\varphi^{*}VB\right)' \otimes  \Lambda_m(M)) \rightarrow \Gamma^{\infty}(M\leftarrow \varphi^{*}VB)\ .
\end{equation}

\begin{remark}
Note how, up to this point, we used some auxiliary metric $h$ in \eqref{eq_1_diffop_1} in order to have a proper differential operator for the subsequent steps. As a consequence the resulting operator $D_{\varphi}(h)$ does depend on the metric chosen and so do its retarded and advanced Green's operators $G_{\varphi}^{\pm}(h) $. What about their counterparts $\mathcal{G}_{\varphi}^{\pm}(h) $?\\

From the definition of Green's operators we have 
$$
	\left\lbrace \hspace{-0.4cm} \begin{array}{ll}
		&D_{\varphi}(h) \circ G_{\varphi}^{\pm}(h) = \mathrm{id}_{\Gamma^{\infty}_c(M\leftarrow \varphi^{*}VB)}\ ,\\
		& G_{\varphi}^{\pm}(h) \circ \left. D_{\varphi}(h) \right\vert_{\Gamma^{\infty}_c(M\leftarrow \varphi^{*}VB)}=\mathrm{id}_{\Gamma^{\infty}_c(M\leftarrow \varphi^{*}VB)}\ .
	\end{array}\ \right. 
$$
The latter is equivalent to 
$$
	\left\lbrace \hspace{-0.4cm} \begin{array}{ll}
		& \delta^{(1)}E(\mathcal{L})_{\varphi}[0]  \circ \mathcal{G}_{\varphi}^{\pm}(h)  = \mathrm{id}_{\Gamma^{\infty}_c(M\leftarrow \varphi^{*}VB' \otimes \Lambda_m(M))}\ ,\\
		& \mathcal{G}_{\varphi}^{\pm}(h) \circ  \left. \delta^{(1)}E(\mathcal{L})_{\varphi} )[0] \right\vert_{\Gamma^{\infty}_c(M\leftarrow \varphi^{*}VB)}=\mathrm{id}_{\Gamma^{\infty}_c(M\leftarrow \varphi^{*}VB)}\ .
	\end{array}\  \right. 
$$
Notice that the family of propagators $\lbrace \mathcal{G}^{\pm}_{\varphi}(h) \rbrace$, together with the differential operator $\delta^{(1)}E(\mathcal{L})_{\varphi}[0]$ of the linearized equations at $\varphi$, defines a family of Green-hyperbolic type operators (see \cite[Definition 3.2]{bar}). Then \cite[Theorem 3.8]{bar} ensures uniqueness for the advanced and retarded propagators, which in turn results in the independence of the Riemannian metric $h$ used before.% The idea behind the proof is as follows: one would like to both extend the domain of $\mathcal{G}^{\pm}_{\varphi}(h)$ and reduce the target space to the same suitable space, once this is done, each propagator becomes the inverse of the linearized equations, then using uniqueness of the inverse we conclude. It turns out that the extension to the spaces $\Gamma^{\infty}_{\pm}(M\leftarrow \varphi^{*}VB'\otimes \Lambda_m(M))$ of future/past compact smooth sections does the job.     
\end{remark}

%\textcolor{red}{proprietà degli operatori green hyp.}

\begin{lemma}\label{lemma_1_duality_green}
Let $g$ a Lorentzian metric on $M$ and $D: \Gamma^{\infty}(M\leftarrow E) \rightarrow \Gamma^{\infty}(M\leftarrow E)$ a linear partial differential operator. Then $D$ is formally self adjoint with respect to the pairing $\langle \ , \ \rangle :  \Gamma^{\infty}_c(M\leftarrow E) \otimes  \Gamma^{\infty}_c(M\leftarrow E)\to \mathbb R$ given by
$$ 
	\left\langle s,\vec{t}\right.  \left. \right\rangle = \int_M (id_{E'} \otimes *_g \circ h^{\flat} (\vec{t} \space\ )) s = \int_M h^{\flat} (\vec{t}\space\ ) s \space\ d\mu_g 
$$
if and only if its integral kernel ${D}(x,y)$ is symmetric. Moreover if $D$ is normally hyperbolic, $G_M^{+}$ and $G_M^{-}$ are each the adjoint of the other in the common domain.
\end{lemma}
\begin{proof}
The equivalent condition follows from
$$
\begin{aligned}
	\left\langle Ds,\vec{t}\space\ \right\rangle &=\int_M h^{\flat} (\vec{t})(x) Ds(x) \space\ d\mu_g(x)= \int_M h^{\flat} (\vec{t})(x)  h^{\sharp} \circ (id_{E'} \otimes *_g) \mathcal{D}(x,s)  \space\ d\mu_g(x) \\ &= \int_{M^2}  \mathcal{D}(x,y)\vec{t}(x)s(y) d\mu_g(x) d\mu_g(y),
\end{aligned} 
$$
where $\mathcal{D}$ is the integral kernel of $D$. If $D$ is self adjoint, then 
$$
%\begin{aligned}
	\left\langle s, G_M^{-} \vec{t} \space\ \right\rangle =\left\langle D G_M^{+} s, G_M^{-} \vec{t}\space\ \right\rangle =\left\langle  G_M^{+}s,D G_M^{-}\vec{t}\space\ \right\rangle  =\left\langle G_M^{+}s,\vec{t} \space\ \right\rangle
%\end{aligned} 
$$
whence the desired adjoint properties of $G_M^{+}$ and $G_M^{-}$.
\end{proof}

For future convenience, we calculate the functional derivatives of $\mathcal{G}^{\pm}_{\varphi}$ and $\mathcal{G}_{\varphi}$, which are clearly smooth by combining Lemma \ref{lemma_1_deriv_green} with \eqref{eq_1_Gvar_+-} and \eqref{eq_1_Gvar}, whence 
\begin{align}\label{eq_1_der_Gvar+-}
\begin{aligned}
	d^{k}\mathcal{G}_{\varphi}^{\pm} (\vec{X}_1,\ldots, \vec{X}_k)&= \sum_{l=1}^{k}(-1)^l \sum_{\substack{(I_1,\ldots,I_l) \\ \in \mathcal{P}(1,\ldots, k)} } \bigg( \bigcirc_{i=1}^{l} \mathcal{G}_{\varphi}^{\pm} \circ \delta^{(|I_{\sigma(i)}| +1)}E(L)_{\varphi}[0]\big( \vec{X}_{I_{i}}\big) \bigg) \circ \mathcal{G}_{\varphi}^{\pm},
\end{aligned}
\end{align}
\begin{equation}\label{eq_1_der_Gvar}
\begin{aligned}
	d^{k}\mathcal{G}_{\varphi} (\vec{X}_1,\ldots, \vec{X}_k) =  \sum_{l=1}^{k}(-1)^l & \sum_{\substack{(I_1,\ldots,I_l) \in \mathcal{P}(1,\ldots, k)} } \sum_{m=0}^{l}\bigg( \bigcirc_{i=1}^{m} \mathcal{G}^-_{\varphi}\circ \delta^{(|I_i| +1)}E(L)\big(\vec{X}_{I_i}\big) \bigg) \\
 &\circ \mathcal{G}_{\varphi}\circ \bigg( \bigcirc_{i=m+1}^{l} \delta^{(|I_i| +1)}E(L)_{\varphi}\big(\vec{X}_{I_i}\big)\circ \mathcal{G}^+_{\varphi}  \bigg),
\end{aligned}
\end{equation}
where $(I_1,\ldots,I_l)$ is partition of the set $\lbrace 1,\ldots,k \rbrace $, and $\vec{X}_I=\otimes_{i\in I}\vec{X}_i$. We stress that in \eqref{eq_1_der_Gvar} the pattern of the compositions of propagators and derivatives of $E(L)$ is as follows: first the $\mathcal{G}_{\varphi}^-$'s, then a single $\mathcal{G}$ and at the end some $\mathcal{G}_{\varphi}^+$'s each intertwined by derivatives of $E(\mathcal{L})$. These will be key to some later proofs. We are now in a position to introduce the Peierls bracket:

\begin{definition}\label{def_1_Peierls}
Let $\mathcal{U}\subset\Gamma^{\infty}(M\leftarrow B)$ be $CO$-open, and $F$, $H \in \mathcal{F}_{\mu loc}(B,\mathcal{U})$. Fix a generalized microlocal Lagrangian $\mathcal{L}$ whose linearized equations induce a normally hyperbolic operator. The \textit{retarded and advanced products} $\textsf{R}_{\mathcal{L}}(F,H)$, $\textsf{A}_{\mathcal{L}}(F,H)$ are functionals defined by 
\begin{equation}
	\textsf{R}_{\mathcal{L}}(F,H)(\varphi) \doteq \left\langle dF_{\varphi}[0], \mathcal{G}_{\varphi}^{+} dH_{\varphi}[0]\right\rangle,
\end{equation}
\begin{equation}
	\textsf{A}_{\mathcal{L}}(F,H)(\varphi) \doteq \left\langle dF_{\varphi}[0], \mathcal{G}_{\varphi}^{-} dH_{\varphi}[0]\right\rangle,
\end{equation}
while the \textit{Peierls bracket} of $F$ and $H$ is 
\begin{align}\label{eq_1_Peierls}
\left\lbrace F, H \right\rbrace_{\mathcal{L}} \doteq \textsf{R}_{\mathcal{L}}(F,H)-\textsf{A}_{\mathcal{L}}(F,H).
\end{align}
\end{definition}
We recall that for a microlocal functional $F$, by \eqref{eq_1_kernel_notation}
$$
	dF_{\varphi}[0](\vec{X})=\int_M f^{(1)}_{\varphi}[0]_i(x)X^i(x)d\mu_g(x),
$$
therefore we can write $\left\lbrace F, H \right\rbrace_{\mathcal{L}}(\varphi)$ as
\begin{equation}\label{eq_1_Peierls_kernel}
	\int_{M^2} f^{(1)}_{\varphi}[0]_i(x) \mathcal{G}_{\varphi}^{ij}(x,y) h^{(1)}_{\varphi }[0]_j(y) d\mu_g(x,y)
\end{equation}
where repeated indices as usual follows the Einstein notation. This implies clearly that Definition \ref{def_1_Peierls} is well posed. Moreover as a consequence of Lemma \ref{lemma_1_duality_green} we see that the Peierls bracket of $F$ and $H$ can also equivalently viewed as $\textsf{R}_{\mathcal{L}}(F,H)-\textsf{R}_{\mathcal{L}}(H,F)=\textsf{A}_{\mathcal{L}}(H,F)-\textsf{A}_{\mathcal{L}}(F,H)$. 

We begin our analysis of the Peierls bracket by listing the support properties of the functionals defined in Definition \ref{def_1_Peierls}. 

\begin{proposition}\label{prop_1_Peierls_1}
Let $\mathcal{U}$, $F$, $H$ be as in the above definition, then the retarded and advanced products and the Peierls bracket are Bastiani smooth with the following support properties:
\begin{align}
	\mathrm{supp} \left( \textsf{R}_{\mathcal{L}}(F,H) \right) \subset & J^+(\mathrm{supp}(F)) \cap  J^-(\mathrm{supp}(H)),\\
	\mathrm{supp} \left( \textsf{A}_{\mathcal{L}}(F,H) \right) \subset & J^+(\mathrm{supp}(H)) \cap  J^-(\mathrm{supp}(F)),	 
\end{align}
which combined yields
\begin{equation}\label{eq_1_supp_peierls}
	\mathrm{supp} \left( \left\lbrace F,G \right\rbrace_{\mathcal{L}} \right) \subset   \left( J^+(\mathrm{supp}(F)) \cup J^-(\mathrm{supp}(F)) \right) \cap \left( J^+(\mathrm{supp}(H)) \cup J^-(\mathrm{supp}(H)) \right).
\end{equation}
\end{proposition}

\begin{proof}
By definition the support properties of $\mathcal{G}_{\varphi}^{\pm}$ and $G_{\varphi}^{\pm}$ are analogue, so combining these properties with  $\textsf{R}_{\mathcal{L}}(F,H)= \frac{1}{2} \textsf{R}_{\mathcal{L}}(F,H) + \frac{1}{2}  \textsf{A}_{\mathcal{L}}(H,F)$ yields the desired result. We now turn to the smoothness. We calculate the $k$-th derivative of $ \textsf{R}_{\mathcal{L}}$. By the chain rule, taking $\mathcal{P}(1,\ldots,k)$ the set of permutations of $\lbrace 1,\ldots , k \rbrace$, we can write 
\begin{align}\label{eq_1_der_retarded}
\begin{aligned}
	 &d^k\textsf{R}_{\mathcal{L}} (F,H)_{\varphi}[0](\vec{X}_{1},\ldots, \vec{X}_k)\\
	 &\quad = \sum_{(J_1,J_2,J_3)\subset P_k} \left\langle F^{(|J_1|+1)}_{\varphi}[0](\otimes_{j_1 \in J_1}\vec{X}_{j_1}) \right., 
	 \left.d^{(|J_2|)}\mathcal{G}_{\varphi}^{+}(\otimes_{j_2 \in J_2}\vec{X}_{j_2}) H^{(|J_3|+1)}_{\varphi}[0] (\otimes_{j_3 \in J_3}\vec{X}_{j_3}) \right\rangle,
\end{aligned}
\end{align}
and similarly
\begin{align}\label{eq_1_der_advanced}
\begin{aligned}
	&d^k\textsf{A}_{\mathcal{L}} (F,H)_{\varphi}(\vec{X}_{1},\ldots, \vec{X}_k)\\ 
	&\quad = \sum_{(J_1,J_2,J_3)\subset P_k} \left\langle F^{(|J_1|+1)}[\varphi](\otimes_{j_1 \in J_1}\vec{X}_{j_1})  \right., \left. d^{(|J_2|)}\mathcal{G}_{\varphi}^{-}(\otimes_{j_2 \in J_2}\vec{X}_{j_2}) H^{(|J_3|+1)}_{\varphi}[0] (\otimes_{j_3 \in J_3}\vec{X}_{j_3}) \right\rangle.
\end{aligned}
\end{align}
To see that the pairing in the derivatives of the  advanced and retarded products are well defined, we use the kernel notation \eqref{eq_1_kernel_notation}, therefore we write the integral kernel of $\textsf{R}_{\mathcal{L}}(F,H)$, which by a little abuse of notation we call $\textsf{R}_{\mathcal{L}}(F,H)(x,y)$ for $x$,$y \in M$. It is
$$
	\textsf{R}_{\mathcal{L}}(F,H)(x,y)= f^{(1)}_{\varphi}[0]_i(x) \left(\mathcal{G}^{+}_{\varphi}\right)^{ij}(x,y) h^{(1)}_{\varphi }[0]_j(y).
$$
Using this notation, we can write the integral kernel $d^k\textsf{R}_{\mathcal{L}} (F,H)_{\varphi}[0](\vec{X}_{1},\ldots, \vec{X}_k)(x,y)$ in \eqref{eq_1_der_retarded} as a sum of terms with two possible contributions:\\
$1)$ [$J_2=\emptyset$]
$$
	 f^{(p+1)}_{\varphi}[0]_i(x,\vec{X}_1, \ldots, \vec{X}_p) (\mathcal{G}^{+}_{\varphi})^{ij}(x,y) h^{(q+1)}_{\varphi}[0]_j(y,\vec{X}_{q+1} \ldots \vec{X}_{p+q}),
$$
where $p+q=k$. Due to smoothness of the functionals $F, \ H$ and, by Remark \ref{remark_1_bastiani_smoothness}, of $\mathcal{G}_{\varphi}^{\pm}$, this is well defined and is Bastiani smooth.\\
$2)$ [$J_2 \neq \emptyset$] 
$$
\begin{aligned}
	&\quad\int_{M^{k-2}} f^{(|J_1|+1)}_{\varphi}[0]_i\big(x,\vec{X}_{J_1}\big) \left(\mathcal{G}^{+}_{\varphi}\right)^{ij_1}(x,z_1)  \delta^{(|I_1|+1)}  E(\mathcal{L})_{\varphi}[0]_{j_1 j_2} \big(z_1,z_2,\vec{X}_{I_1}\big) \\
	&\qquad\times\left(\mathcal{G}^{+}_{\varphi}\right)^{j_2j_3}(z_2,z_3) \delta^{(|I_2|+1)} E( \mathcal{L})_{\varphi}[0]_{j_3 j_4}\big(z_3,z_4,\vec{X}_{I_2}\big) \cdots \delta^{(|I_l|+2)} E(\mathcal{L})_{\varphi}[0]_{j_{2l-1} j_{2l}}\big(z_{2l-1},z_{2l},\vec{X}_{I_l} \big) \\
	&\qquad\times\left(\mathcal{G}^{+}_{\varphi}\right)^{j_{2l}j}(z_{2l},y)   h^{(|J_3|+1)}_{\varphi}[0]_j(y,\vec{X}_{p+k_1+\ldots+ k_l+1}, \ldots, \vec{X}_{J_3})d\mu_g(z_1,\ldots ,z_{2l}) ,
\end{aligned}
$$
where $I_1\cup\ldots\cup I_l=J_2$. Then again, due to the Bastiani smoothness of $F, \ H$, $\mathcal{G}_{\varphi}^{\pm}$ and $\mathcal{L}$, we conclude that this piece exists and is Bastiani smooth. Repeating the above calculations for $\textsf{A}_{\mathcal{L}} $ amounts to substituting  each $+$ with $-$, resulting in Bastiani smoothness for the advanced product. Finally since $\left\lbrace F, H \right\rbrace_{\mathcal{L}} = \textsf{R}_{\mathcal{L}}(F,H)-\textsf{A}_{\mathcal{L}}(F,H)$ we conclude that it is smooth as well.
\end{proof}

We have seen that the Peierls bracket is well defined for microlocal functionals, we stress however that the image under the Peierls bracket of microlocal functionals fails to be microlocal, so it is necessary to broaden the domain of this bracket. An idea is to use the full potential of microlocal analysis, and use wave front sets to define pairings. First though we make explicit the ``good" subsets of $T^{*}M$, that is, those subsets in which the wavefront can be localized.

\begin{definition}\label{def_1_WF_mucaus}
Let $(M,g)$ be a Lorentzian spacetime, define $\Upsilon_k(g) \subset T^{*}M^k$ as follows:
\begin{align}\label{eq_1_WF_mu_caus}
\begin{aligned}
	\Upsilon_k(g)\doteq \Big\lbrace (x_1,\ldots,x_k,\xi_1,\ldots,\xi_k)& \in T^{*}M^k \backslash 0 :\\
	&(x_1,\ldots,x_k,\xi_1,\ldots,\xi_k) \notin \overline{V}^+_{k}(x_1,\ldots,x_k) \cup \overline{V}^-_{k}(x_1,\ldots,x_k) \Big\rbrace
\end{aligned}
\end{align}
where 
$$
	\overline{V}^{\pm}_k(x_1,\ldots,x_k) = \prod_{j=1}^k \overline{V}^{\pm}(x_j). 
$$
If $\mathcal{U}\subset \Gamma^{\infty}(M\leftarrow B)$ is open we say that a functional $F:\mathcal{U} \rightarrow \mathbb{R}$ with compact support is \textit{microcausal} with respect to the Lorentz metric $g$ in $\varphi$ if $\mathrm{WF}(d^kF_{\varphi}[0]) \cap \Upsilon_k(g)=\emptyset$ for all $k \in \mathbb{N}$. We say that $F$ is microcausal with respect to $g$ in $\mathcal{U}$ if $F$ is microcausal for all $\varphi \in \mathcal{U}$. We denote the set of microcausal functionals in $\mathcal{U}$ by $\mathcal{F}_{\mu c}(B,\mathcal{U},g)$.
\end{definition}
% We remark how Definition \eqref{def_1_WF_mucaus} can equivalently be given by requiring 
% $$
% 	\mathrm{WF} \left(\nabla^{(k)}F_{\varphi}[0]\right) \subset \Upsilon_k(g).
% $$
One can show by induction, using \eqref{eq_1_cov_der}, that testing microcausality by calculating $\mathrm{WF}(\nabla^kF_{\varphi}[0])$ or $\mathrm{WF}(d^kF_{\varphi}[0])$ is equivalent. The case $k=1$ is trivial, while the case with arbitrary $k$ follows from:

\begin{lemma}\label{lemma_1_cov_WF_equivalence}
    Let $F$ be a functional and $\Phi$ a symmetric linear connection with covariant derivative $\nabla$. Then $F $ is microcausal if and only if $\nabla^kF_{\varphi}[0]$ does not have wave front set contained in $\Upsilon_k(g)$ for all $k \in \mathbb{N}$.
\end{lemma}
\begin{proof}
We show this by induction on the derivative order. For $k=1$ we have $\nabla F_{\varphi}[0] = dF_{\varphi}[0] $. More generally, suppose $\mathrm{WF}\left(\nabla^{k-1}F_{\varphi}[0]\right) \cap \Upsilon_{k-1}(g)=\emptyset$, we claim
$$
	\mathrm{WF}\left(\nabla^{k}F_{\varphi}[0]\right) \cap \Upsilon_k(g) =\emptyset.
$$
From \eqref{eq_1_cov_der} we have 
\begin{align*}
	&\nabla^{k}F_{\varphi}[0](\vec{X}_1,\ldots, \vec{X}_k) \\ &\doteq d^kF_{\varphi}[0](\vec{X}_1,\ldots, \vec{X}_k) 
	 + \sum_{j=1}^k \frac{1}{k!} \sum_{\sigma \in \mathcal{P}(k)} \nabla^{k-1}F_{\varphi} [0](\Gamma_{\varphi} (\vec{X}_{\sigma(j)}, \vec{X}_{\sigma(k)}), \vec{X}_{\sigma(1)}, \ldots, \widehat{\vec{X}_{\sigma(j)}}, \ldots \vec{X}_{\sigma(k-1)})\ .
\end{align*}
Assume that $\nabla^{k-1}F_{\varphi}[0]$ is microcausal. Since $F$ is microcausal as well, it is sufficient to show microcausality holds for the other terms in the sum. Due to symmetry of the connection, we can simply study the wave front set of a single term such as 
\begin{equation}\label{eq_1_kernel_comp}
	\nabla^{k-1}F_{\varphi} (\Gamma_{\varphi} (\vec{X}_{j}, \vec{X}_{k}), \vec{X}_{1}, \ldots, \widehat{\vec{X}_{j}}, \ldots \vec{X}_{k-1}).
\end{equation}
The idea is to apply Theorem 8.2.14 in \cite{hormanderI}. Recall that a connection $\Gamma_{\varphi}$ can be seen as a mapping $\Gamma^{\infty}_c(M\leftarrow \varphi^*VB)\times \Gamma^{\infty}_c(M\leftarrow \varphi^*VB)\rightarrow\Gamma^{\infty}_c(M\leftarrow \varphi^*VB)$ with associated integral kernel $\Gamma[\varphi](x,y,z)$ defined by
$$
	\otimes^3 \Gamma^{\infty}_c(M\leftarrow \varphi^*VB) \rightarrow \mathbb{R}: (\vec{X},\vec{Y},\vec{Z}) \mapsto \int_{M^3} h_{kl}(\varphi(x))\Gamma[\varphi]^l_{ij}(x,y,z)\vec{X}^i(x)\vec{Y}^j(y) \vec{Z}^k(z) d\mu_g(x,y,z)
$$
where $h$ is an auxiliary Riemannian metric on the fiber of the bundle $B$ which is to be regarded as a tool for calculations\footnote{We can, without loss of generality assume that $\Gamma$ are the Christoffel symbols with respect to the metric $h$.}. Using the support properties of the connection coefficients $\Gamma_{\varphi}$ we obtain $\Gamma[\varphi]^l_{ij}(x,y,z)=\Gamma^l_{ij}(\varphi(x))\delta(x,y,z)$, where $\Gamma^l_{ij}(\varphi(x))$ are the Christoffel coefficients of the connection on the typical fiber of $B$ as in \eqref{eq_1_connection_induced}. Then
$$
	\mathrm{WF}\left(\Gamma[\varphi]\right)= \lbrace (x,y,z,\xi,\eta,\zeta) \in T^*M^3\backslash 0: x=y=z, \ \xi+\eta+\zeta=0 \rbrace.
$$
Composition of the two integral kernels in \eqref{eq_1_kernel_comp} is well defined provided $\mathrm{WF}'(\nabla^{n-1}F_{\varphi}[0])_{M}\cap \mathrm{WF}\left(\Gamma[\varphi]\right)_{M}= \emptyset$ and that the projection map $: \triangle_{3}M \rightarrow M $ is proper. The former is a consequence of $\mathrm{WF}\left(\Gamma[\varphi]\right)_{M}= \emptyset$, the latter is a trivial statement for the diagonal embedding. Then we can apply Theorem 8.2.14, and estimate
\begin{align*}
	\mathrm{WF}\left( \nabla^{k-1}F_{\varphi} \circ \Gamma_{\varphi} \right) \subset & \Big\{ (x_1,\ldots,x_k,\xi_1,\ldots,\xi_k) \in T^*M^k : \exists (y,\eta) : \  (x_j,x_k,y,\xi_j,\xi_k,-\eta) \in  \mathrm{WF}\left(\Gamma[\varphi]\right)\ ,\\ 
	& \quad (y,x_1,\ldots, \widehat{x_j},\ldots ,x_{k-1},\eta,\xi_1,\ldots,\widehat{\xi}_j,\ldots,\xi_{k-1})\in \mathrm{WF} \left(\nabla^{k-1}F_{\varphi}[0] \right) \Big\} \\
	& \bigcup
\Big\{ (x_1,\ldots,x_k,\xi_1,\ldots,\xi_k) \in T^*M^k : x_j=x_k , \space\ \xi_j = \xi_k=0\ , \\ 
	&\qquad  (y,x_1,\ldots, \widehat{x_j},\ldots ,x_{k-1},0,\xi_1,\ldots,\widehat{\xi}_j,\ldots,\xi_{k-1})\in \mathrm{WF} \left(\nabla^{k-1}F_{\varphi}[0] \right) \Big\} \\ 
	& \bigcup \Big\{ (x_1,\ldots,x_k,0,\ldots,0,\xi_j,0,\ldots,0,\xi_k) \in T^*M^k :  (x_j,x_k,y,\xi_j,\xi_k,0) \in  \mathrm{WF}\left(\Gamma_{\varphi}\right)\ , \\ 
	&\qquad (y,x_1,\ldots, \widehat{x_j},\ldots ,x_{k-1},\eta,0,\ldots,0)\in \mathrm{WF} \left(\nabla^{k-1}F_{\varphi}[0] \right) \Big\}\\
	& = \Pi_1\cup \Pi_2\cup \Pi_3.
\end{align*}
If by contradiction, we had that $\nabla^{k-1}F_{\varphi} \circ \Gamma_{\varphi}$ was not microcausal, there would be elements of its wavefront set for which all $\xi_1,\ldots,\xi_k$ are, say, future pointing. In this case those must belong to $\Pi_1$, but then $\eta$ is future-pointing as well by the form of $\mathrm{WF}(\Gamma_{\varphi})$, contradicting our initial assumption.
\end{proof}

One can also show that microcausality does not depend upon the connection chosen by computing
\begin{align*}
	\nabla^{n}F_{\varphi}[0](\vec{X}_1,\ldots, \vec{X}_n) & - \widetilde{\nabla}^{n}F_{\varphi}[0]\big(\vec{X}_1,\ldots, \vec{X}_n\big)  \\
	&=  \sum_{j=1}^n \frac{1}{n!} \sum_{\sigma \in \mathcal{P}(n)} \nabla^{n-1}F_{\varphi} [0]\Big(\Gamma_{\varphi} \big(\vec{X}_{\sigma(j)}, \vec{X}_{\sigma(n)}\big), \vec{X}_{\sigma(1)}, \ldots, \widehat{\vec{X}_{\sigma(j)}}, \ldots \vec{X}_{\sigma(n-1)}\Big) \\ 
	& \quad- \sum_{j=1}^n \frac{1}{n!} \sum_{\sigma \in \mathcal{P}(n)} \widetilde{\nabla}^{n-1}F_{\varphi} [0]\Big(\widetilde{\Gamma}_{\varphi} \big(\vec{X}_{\sigma(j)}, \vec{X}_{\sigma(n)}\big), \vec{X}_{\sigma(1)}, \ldots, \widehat{\vec{X}_{\sigma(j)}}, \ldots \vec{X}_{\sigma(n-1)}\Big);
\end{align*}
and then combining induction with Lemma \ref{lemma_1_cov_WF_equivalence} to get an empty wave front set for the terms on right hand side of the above equation. Another consequence of Lemma \ref{lemma_1_cov_WF_equivalence} is that microcausality of a functional does not depend on the ultralocal charts used to perform the derivatives. We immediately have the inclusion $\mathcal{F}_{reg}(B,\mathcal{U}) \subset \mathcal{F}_{\mu c}(B,\mathcal{U},g)$.

\begin{proposition}\label{prop_1_muloc_into_mucaus}
Let $\mathcal{U}\subset\Gamma^{\infty}(M\leftarrow B)$ be $CO$-open, then if $F\in \mathcal{F}_{\mu loc}(B,\mathcal{U})$, $\mathrm{WF}\big(F^{(k)}_{\varphi}[0]\big)$ is conormal to $ \triangle_k (M)$ \textit{i.e.} $\mathrm{WF}\left(F^{(k)}_{\varphi}[0]\right)\subset \lbrace(x,\ldots,x,\xi_1,\ldots,\xi_k): \xi_1+\ldots+\xi_k=0 \rbrace$ for all $k \geq 2$ and $\varphi \in \mathcal{U}$. Therefore $\mathcal{F}_{\mu loc}(B,\mathcal{U}) \subset \mathcal{F}_{\mu c}(B,\mathcal{U}, g)$. 
\end{proposition}

\begin{proof}
By Lemma \ref{lemma_1_cov_WF_equivalence} microcausality is an intrinsic property of functionals in $\Gamma^{\infty}(M\leftarrow B)$, therefore it suffice to verify the claim in a generic chart $\mathcal{U}_{\varphi}$. Note that by Definition \ref{def_1_func_classes}, $\mathrm{WF}(F^{(1)}_{\varphi}[0])=\emptyset$. Consider therefore $d^kF_{\varphi}[0](\vec{X}_1,\ldots,\vec{X}_k)$ with $k \geq 2$, the associated integral kernel has the form 
$$
	\int_{M^k} f^{(k)}_{\varphi}[0]_{i_1 \cdots i_k}(x_1)\delta(x_1,\ldots,x_k)\vec{X}_1^{i_1}(x_1)\cdots\vec{X}_k^{i_k}(x_k)d\mu_g(x_1,\ldots,x_k)\ .
$$
where $f^{(k)}_{\varphi}[0]_{i_1 \cdots i_k}$ is some smooth function for each indices $i_1,\ldots, i_k$. Therefore,
$$
	\mathrm{WF}(F^{(k)}_{\varphi}[0])= N^{*} \triangle_k (M)=\bigg\{(x_1,\ldots,x_k,\xi_1,\ldots,\xi_k)\in T^{*}M^k\backslash 0 : \ x_1=\cdots=x_k; \  \sum_{j=1}^k \xi_j=0\bigg\}	\ .
$$
In addition, if $(x,\ldots,x,\xi_1,\ldots,\xi_k)$ is in $\mathrm{WF}\big(F^{(k)}_{\varphi}[0]\big)$ and has, say, the first $k-1$ covectors in $\overline{V}_{k-1}^{+}(x,\ldots,x)$, then $\xi_k=-(\xi_1+\ldots+\xi_{k-1})$ and we see that $\xi_k \in \overline{V}^{-}(x)$; whence microlocality implies microcausality.
\end{proof}

\begin{theorem}\label{thm_1_mucaus_1}
Let $\mathcal{U}\subset\Gamma^{\infty}(M\leftarrow B)$ be $CO$-open and $\mathcal{L}$ a generalized microlocal Lagrangian with normally hyperbolic linearized equations. Then the Peierls bracket associated to $\mathcal{L}$ extends to $\mathcal{F}_{\mu c}(B,\mathcal{U}, g)$, has the same support property of Proposition \ref{prop_1_Peierls_1} and depends only locally on $\mathcal{L}$, that is, for all $F$, $H\in \mathcal{F}_{\mu c}(B,\mathcal{U}, g) $, $\lbrace F,H\rbrace_{\mathcal{L}}$ is unaffected by perturbations of $\mathcal{L}$ outside the right hand side of \eqref{eq_1_supp_peierls}. The same locality property holds for the retarded and advanced products.
\end{theorem}

\begin{proof}
Clearly $\lbrace F,H \rbrace_{\mathcal{L}}$ is well defined, in fact since $\mathrm{WF}(H^{(1)}_{\varphi}[0])$ is spacelike, and $\mathcal{G}_{\varphi} $, according to Theorem \ref{thm_1_properties_of_Green_functions}, propagates only lightlike singularities along lightlike geodesics, then $\mathcal{G}_{\varphi} dH_{\varphi}[0]$ must be smooth, giving a well defined pairing. As for support properties the proof can be carried on analogously to the proof of Proposition \ref{prop_1_Peierls_1}.

For the support behavior of the bracket, suppose $\mathcal{L}_1$ and $\mathcal{L}_2$ are generalized Lagrangians, such that for some fixed $\varphi \in \mathcal{U}$, $\delta^{(1)}E(\mathcal{L}_1)_{\varphi}[0]$ and $\delta^{(1)}E(\mathcal{L}_2){\varphi}[0]$ differ only in a region outside 
\begin{equation}\label{eq_1_supp_peierls_2}
\mathcal{O}\doteq \left( J^+(\mathrm{supp}(F)) \cup J^-(\mathrm{supp}(F)) \right) \cap \left( J^+(\mathrm{supp}(H)) \cup J^-(\mathrm{supp}(H)) \right).
\end{equation}
By the support properties of retarded and advanced propagators of Proposition \ref{prop_1_Peierls_1} we have 
$$
\left\langle dF_{\varphi}[0], (\mathcal{G}^+_{\varphi,\mathcal{L}_1}-\mathcal{G}^+_{\varphi,\mathcal{L}_2})dH_{\varphi}[0]\right\rangle=0,
$$
as well as
$$
\left\langle dF_{\varphi}[0], (-\mathcal{G}^-_{\varphi,\mathcal{L}_1}+\mathcal{G}^-_{\varphi,\mathcal{L}_2})dH_{\varphi}[0]\right\rangle=0.
$$
Taking the sum of the two we find 
$$
	\lbrace F,H \rbrace_{\mathcal{L}_1}-\lbrace F,H \rbrace_{\mathcal{L}_2}=0.
$$
\end{proof}

\begin{theorem}\label{thm_1_peierls_closedness}
Let $\mathcal{U} \subset \Gamma^{\infty}(M\leftarrow B)$ $CO$-open and $\mathcal{L}$ a generalized Lagrangian. If $F$, $H\in \mathcal{F}_{\mu c}(B,\mathcal{U},g) $ we have that $\lbrace F,H \rbrace_{\mathcal{L}}\in \mathcal{F}_{\mu c}(B, \mathcal{U},g) $ as well.
\end{theorem}

\begin{proof}
%Suppose that we can take $\delta^{k}\mathcal{L}(1)[\varphi]$ with compact support, if not by the above theorem we restrict the operator by subtracting a suitable piece, so as to make the operator supported in the set $\mathcal{O}$ defined in $2.26$. This is achieved by replacing $E(\mathcal{L})[\varphi]$ with a cutoff version $$E'(\mathcal{L})[\varphi_0](\varphi)+f\left(E(\mathcal{L})[\varphi]-E'(\mathcal{L})[\varphi_0](\varphi)\right)$$ where $f\equiv 1$ in a neighborhood of $\mathcal{O}$. Then no modifications of $\lbrace F,G \rbrace_{\mathcal{L}}$ occur for all $\varphi_0$, $\varphi \in \mathcal{U}_{\varphi_0}$. 
By Faà di Bruno's formula,
\begin{equation}\label{eq_1_k-th_derivative_Peierls}
\begin{aligned}
	&d^k\left\lbrace F,G \right\rbrace_{\mathcal{L} ,\varphi}[0](\vec{X}_{1},\ldots, \vec{X}_k)\\
	&\qquad=\sum_{(J_1,J_2,J_3)\subset P_k} \left\langle F^{(|J_1|+1)}_{\varphi }[0](\otimes_{j_1 \in J_1}\vec{X}_{j_1}) \right.,  \left.d^{(|J_2|)}\mathcal{G}_{\varphi}(\otimes_{j_2 \in J_2}\vec{X}_{j_2}) G^{(|J_3|+1)}_{\varphi}[0] (\otimes_{j_3 \in J_3}\vec{X}_{j_3}) \right\rangle.
\end{aligned}
\end{equation}
while by \eqref{eq_1_der_Gvar},
\begin{align}\label{eq_1_k-th_derivative_G_var}
\begin{aligned}
	d^{|J_2|}\mathcal{G}_{\varphi} (\vec{X}_1,\ldots, \vec{X}_k) =\sum_{l=1}^{k}(-1)^l & \sum_{\substack{(I_1,\ldots,I_l) \\ \in \mathcal{P}(J_2)} } \sum_{p=0}^{l}\bigg( \bigcirc_{i=1}^{p} \mathcal{G}^-_{\varphi}\circ \delta^{(|I_i| +1)}E(\mathcal{L})\big(\vec{X}_{I_i}\big) \bigg) \\ 
    &\circ \mathcal{G}_{\varphi}\circ \bigg( \bigcirc_{i=p+1}^{l} \delta^{(|I_i| +1)}E(\mathcal{L})_{\varphi}\big(\vec{X}_{I_i}\big)\circ \mathcal{G}^+_{\varphi}  \bigg),
\end{aligned}
\end{align}
where $\bigcirc_{i=1}^p$ stands for composition of mappings indexed by $i$ from $1$ to $p$. For the rest of the proof, we will use the integral notation we used in \eqref{eq_1_kernel_notation} and in the proof of Proposition \ref{prop_1_Peierls_1}. Recall that the mapping $\delta^{(n)}E(\mathcal{L})_{\varphi}[0]$ has associated a compactly supported integral kernel $\mathcal{L}(1)^{(n+1)}_{\varphi}[0](x,z_{1},\ldots, z_{n})$ and its wave front is in $N^{*}\triangle_{n+1}(M)$ by Proposition \ref{prop_1_muloc_into_mucaus}. Then again, we have two general cases:\\
$1)$ $J_2=\emptyset$.\\
Let $|J_1|=p$, $|J_3|=q=k-p$, the typical term has the form
\begin{equation}\label{eq_1_kernel_k-th_derivative_Paierl_1}
\begin{aligned}
	d^k\left\lbrace F,H \right\rbrace_{\varphi}[0](z_1,\ldots,z_k)=\int_{M^{2}} f^{(p+1)}_{\varphi}[0]_i(x,z_1,\ldots,z_{p}) , \mathcal{G}_{\varphi}^{ij}(x,y) h^{(q+1)}_{\varphi}[0]_j(y,z_{p+1},\ldots,z_k)d\mu_g(x,y).
\end{aligned}
\end{equation}
Suppose by contradiction that there is some $(x_1,\ldots,x_k,\xi_1,\ldots,\xi_k)\in \mathrm{WF}( \left\lbrace F,H \right\rbrace_{\varphi}[0])$ having $(\xi_1,\ldots,\xi_{k})\in \overline{V}^+_{k}(x_1,\ldots,x_k)$ (the argument works similarly for $(\xi_1,\ldots,\xi_{k})\in \overline{V}^-_{k}(x_1,\ldots,x_k)$). Using twice Theorem 8.2.14 in \cite{hormanderI} in the above pairing yields
$$
\begin{aligned}
	&\mathrm{WF}\Big(\big\lbrace F,H \big\rbrace^{(k)}_{\varphi}[0]\Big)\\ &\qquad\subseteq  \Big\lbrace (z_1,\ldots,z_k,\xi_1,\ldots,\xi_k)  :\  \exists (y,\eta)\in T^{*}M 	(x,z_1,\ldots,z_p,-\eta,\xi_1,\ldots,\xi_p) \in \mathrm{WF}(F^{(p+1)}_{\varphi}[0]_i), \\
	&\qquad\qquad (x,z_{p+1},\ldots,z_{k},\eta,\xi_{p+1},\ldots,\xi_k) \in \mathrm{WF}\big(\mathcal{G}_{\varphi}^{ij}H^{(q+1)}_{\varphi}[0]_j\big) \Big\rbrace \\ 
	 &\qquad\subset \Big\lbrace (z_1,\ldots,z_k,\xi_1,\ldots,\xi_k):  \ \exists (x,\eta),(y,\zeta)\in T^*M : (x,z_1,\ldots,z_p,-\eta,\xi_1,\ldots, \xi_p) \in \mathrm{WF}\big(F^{(p+1)}_{\varphi}[0]_i\big) \\ 
	&\qquad\qquad (x,y,\eta,-\zeta)\in \mathrm{WF}(\mathcal{G}_{\varphi}^{ij}), \ (y,z_{p+1},\ldots,z_{k},\zeta,\xi_{p+1},\ldots, \xi_k) \in \mathrm{WF}(H^{(q+1)}_{\varphi}[0]_j) \Big\rbrace.
\end{aligned}
$$
So if $(z_1,\ldots,z_k,\xi_1,\ldots,\xi_k)\in \mathrm{WF}\big(\left\lbrace F,H \right\rbrace^{(k)}_{\varphi}\big)$, then $\exists$ $(x,\eta), (y,\zeta) \in T^{*}M$ such that 
$$
	\left\lbrace \begin{array}{ll}
		(x,z_1,\ldots,z_p,-\eta,\xi_1,\ldots,\xi_p)&  \in \mathrm{WF}(F^{(p+1)}_{\varphi}[0]_i))\\
		(x,y,\eta,-\zeta)  &\in \mathrm{WF}(\mathcal{G}_{\varphi}^{ij})\\
		(y,z_{p+1},\ldots,z_k,\zeta,\xi_{p+1},\ldots,\xi_k) &\in \mathrm{WF}(H^{(q+1)}_{\varphi}[0]_j)\ .
	\end{array} \right. 
$$
By $(iv)$ Theorem \ref{thm_1_properties_of_Green_functions}, $\mathrm{WF}(\mathcal{G}_{\varphi}^{ij})$ contains pairs of lightlike covectors with opposite time orientation, so in case $\eta\in \overline{V}^{+}(x)$ (resp. $\eta\in \overline{V}^{-}(x)$), $\zeta\in \overline{V}^{+}(y)$ (resp. $\zeta\in \overline{V}^{-}(y)$) in which case $	\mathrm{WF}\big(H^{(q+1)}_{\varphi}[0]_j\big)$ (resp. $\mathrm{WF}\big(F^{(p+1)}_{\varphi}[0]_i\big)$) does violate the microcausality condition of Definition \ref{def_1_WF_mucaus}.\\
$2)$ $J_2 \neq \emptyset$. \\
Then again, call $|J_1|=p$, $|J_3|=k-q$, set also, referring to \eqref{eq_1_k-th_derivative_G_var}, $|I_j|=k_j$ for $j=1,\ldots,l$ so that $|J_2|=k_1+\cdots+k_l$. Combining \eqref{eq_1_k-th_derivative_Peierls} with \eqref{eq_1_k-th_derivative_G_var} with the integral kernel notation we get
\begin{equation}\label{eq_1_kernel_k-th_derivative_Paierl_2}
\begin{aligned}
	 \{ F,G \}^{(k)}_{\varphi}[0](z_1,\ldots,z_k)&=\int_{M^{k}} f^{(p+1)}_{\varphi}[0]_i(x,z_1,\ldots,z_p) \mathcal{G}^{- \space\ i j_1}_{\varphi}(x,x_1) d^{(k_1+2)} \mathcal{L}_{\varphi}[0]_{j_1 i_1} (x_1,y_1,z_{I_1})\\ & \quad
	 \mathcal{G}^{- \space\ i_1 j_2}_{\varphi}(y_1,x_2)   \cdots  \mathcal{G}^{- \space\ i_{m-1} j_{m}}_{\varphi}(y_{m-1},x_{m}) d^{(k_m+2)}\mathcal{L}_{\varphi}[0]_{j_{m} i_{m}} (x_{m},y_{m},z_{I_m})\\ &  \quad
	 \mathcal{G}_{\varphi}^{i_{m} j_{m+1}}(y_{m},x_{m+1}) 
	 d^{(k_{m+1}+2)} \mathcal{L}_{\varphi}[0]_{j_{m+1} i_{m+1}} (x_{m+1},y_{m+1},z_{I_{m+1}})  \\ & \quad
	 \mathcal{G}^{+ \ i_{m+1} j_{m+2}}_{\varphi}(y_{m+1},z_{m+2})  \ldots  d^{(k_l+2)} \mathcal{L}_{\varphi}[0]_{j_{l} i_{l}} (x_{l},y_{l},z_{I_l},)\mathcal{G}^{+\ i_{l} j}(y_{l} ,y) \\
	 & h^{(k-q+1)}_{\varphi}[0]_j(y,z_{q+1},\ldots,z_k)d\mu_g(x,x_1,y_1,\ldots,x_{l},y_{l},y).
\end{aligned}
\end{equation}
Using in \cite[Theorem 8.2.14]{hormanderI}, Theorem \ref{thm_1_properties_of_Green_functions} and Proposition \ref{prop_1_muloc_into_mucaus}, we can estimate the wave front set of the integral kernel of $\{ F,H \}^{(k)}_{\varphi}[0]$ as all elements $(z_1,\ldots,z_k,\xi_1,\ldots,\xi_k)\in T^{*}M^k$ for which there are $(x,\eta)$, $(x_1,\eta_1)$, $\ldots,(x_l\eta_l)$, $(y_1,\zeta_1)$, $\ldots , (y_l,\zeta_l)$ $(y,\zeta) \in T^{*}M $ having 
$$
	\left\lbrace \begin{array}{ll}
		(x,z_1,\ldots,z_p,-\eta,\xi_1,\ldots,\xi_p) & \in \mathrm{WF}\big(f^{(p+1)}_{\varphi}[0]_i\big),\\
		(x,x_{1},\eta,-\eta_{1}) &\in \mathrm{WF}(\mathcal{G}^{-\ ij_1}_{\varphi}),\\
		(x_{1},y_{1},z_{|I_1|},\eta_{1},-\zeta_{1},\xi_{I_1}) &\in \mathrm{WF}\big(d^{(k_{1}+2)}\mathcal{L}_{\varphi}[0]_{i_1j_1}\big),\\
		\vdots & \vdots\\
		(y_{m-1},x_{m},\zeta_{m-1},-\eta_{m}) &\in \mathrm{WF}(\mathcal{G}^{-\ i_{m-1}j_m}_{\varphi}),\\
		(x_{m},y_{m},z_{I_m},\eta_{m},-\zeta_{m},\xi_{I_m}) &\in \mathrm{WF}\big(d^{(k_{m}+2)}\mathcal{L}_{\varphi}[0]_{i_mj_m}\big),\\
		(y_{m},x_{m+1},\zeta_{m},-\eta_{m+1})  &\in \mathrm{WF}(\mathcal{G}_{\varphi}^{i_mj_{m+1}}),\\
		(x_{m+1},y_{m+1},z_{I_{m+1}},\eta_{m+1},-\zeta_{m+1},\xi_{I_{m+1}}) &\in \mathrm{WF}\big(d^{(k_{m+1}+2)}\mathcal{L}_{\varphi}[0]_{i_{m+1}j_{m+1}}\big),\\
		(y_{m+1},x_{m+2},\zeta_{m+1},-\eta_{m+2}) &\in \mathrm{WF}(\mathcal{G}^{+ \ i_{m+1}j_{m+2}}_{\varphi}),\\
		\vdots & \vdots\\
		(y_{l-1},x_{l},\zeta_{l-1},-\eta_{l}) &\in \mathrm{WF}(\mathcal{G}^{+\ i_{l-1}j_l}_{\varphi}),\\
		(x_{l},y_{l},z_{I_l},\eta_{l},-\zeta_{l},\xi_{I_l}) &\in \mathrm{WF}\big(d^{(k_{l}+2)}\mathcal{L}_{\varphi}[0]_{i_lj_l}\big),\\
		(y_{l},y,\zeta_{l},-\zeta) &\in \mathrm{WF}(\mathcal{G}^{+\ i_lj}_{\varphi}),\\
		(y,z_{k-q+1},\ldots,z_k,\zeta,\xi_{k-q+1},\ldots,\xi_k) & \in \mathrm{WF}\big(h^{(k-q+1)}_{\varphi}[0]_j\big).
	\end{array} \right. 
$$
Suppose by contradiction that $(z_1,\ldots,z_{k},\xi_1,\ldots,\xi_k)\in \mathrm{WF}\big( \lbrace F,H \rbrace_{\varphi}^{(k)}\big)$ has $(\xi_1,\ldots,\xi_k)\in \overline{V}^+_k(z_1,\ldots,z_k)$ $\big($resp. $(\xi_1,\ldots,\xi_k)\in \overline{V}^-_{k}(z_1,\ldots,z_k)\big)$. Then $\zeta_m$ and $\eta_{m+1}$ are both either lightlike future directed, or lightlike past directed. In the first case, propagation of singularities implies that $\zeta$ is lightlike future directed, contradicting microcausality of $H^{(k-q+1)}_{\varphi}[0]$ (resp. $\zeta$ is lightlike past directed, contradicting the microlocality of $H^{(k-q+1)}_{\varphi}[0]$); in the second case, propagation of singularities implies that $\eta$ is lightlike past directed, contradicting microcausality of $F^{(p+1)}_{\varphi}[0]$ (resp. $\eta$ is lightlike future directed, contradicting the microlocality of $F^{(p+1)}_{\varphi}[0]$). We remark that in the wave front set of $\{F,H\}_{\varphi}^{(k)}[0]$ is the (finite) union under all possible choices of indices for all wave front sets of the form \eqref{eq_1_kernel_k-th_derivative_Paierl_1} or \eqref{eq_1_kernel_k-th_derivative_Paierl_2}, each of which is however microcausal, implying that their finite union will be microcausal as well.
\end{proof}

%An important feature of this proof is the presence of the causal propagator $\mathcal{G}$ at some point in the composition pattern of the integral kernels. This implies that its wavefront covectors lacks the diagonal part of $\mathrm{WF}(\mathcal{G}^{\pm})$ which in some sense is responsible for a "stationary" propagation of singularities, therefore resulting in the appearance of spacelike covectors which will generate no final contradiction in the argument. Also note that spacetime compactness of derivatives of $E(\mathcal{L})$ enables us to compose operators when the wave front set analysis is carried out.

\begin{theorem}\label{thm_1_jacobi}
The mapping $(F,H) \mapsto \lbrace F,H \rbrace_{\mathcal{L}}$ defines a Lie bracket on $\mathcal{F}_{\mu c}(B,\mathcal{U},g) $, for all $\mathcal{U}$ $CO$-open. %In particular $\left(\mathcal{F}_{\mu caus}(B,\mathcal{U},g), \lbrace \cdot,\cdot \rbrace_{\mathcal{L}}\right)$ is a Poisson $*$-algebra.
\end{theorem}

\begin{proof}
Bilinearity and antisymmetry are clear from Definition \ref{def_1_Peierls}, while Theorem \ref{thm_1_peierls_closedness} ensures the closure of the bracket operation. We are thus left with showing Jacobi identity
$$
\lbrace F,\lbrace G,H \rbrace_{\mathcal{L}} \rbrace_{\mathcal{L}}+\lbrace G,\lbrace H,F \rbrace_{\mathcal{L}} \rbrace_{\mathcal{L}}+\lbrace H,\lbrace F,G \rbrace_{\mathcal{L}} \rbrace_{\mathcal{L}}=0. \quad \forall F, \ G, \ H \in \mathcal{F}_{\mu c}(B,\mathcal{U},g).
$$
Using the integral kernel notation as in the above proof, we have 
$$
\begin{aligned}
	& \lbrace F,\lbrace G,H \rbrace_{\mathcal{L}} \rbrace_{\mathcal{L}}(\varphi) = \int_{M^{2}}  f^{(1)}_{\varphi}[0]_i(x)\mathcal{G}_{\varphi}^{ij}(x,y)\lbrace G,H \rbrace^{(1)}_{\mathcal{L} \space\ \varphi}[0]_j(y)d\mu_g(x,y) \\ &=
	\int_{M^{4}}f^{(1)}_{\varphi}[0]_i(x)\mathcal{G}_{\varphi}^{ij}(x,y)
	\Big( g^{(2)}_{\varphi}[0]_{jk}(y,z) \mathcal{G}_{\varphi}^{kl}(z,w) h^{(1)}_{\varphi}[0]_l(w) \\
 &\qquad + g^{(1)}_{\varphi}[0]_k(z) \mathcal{G}^{kl}_{\varphi}(z,w) h^{(2)}_{\varphi}[0]_{jl}(y,w) \Big) d\mu_g(x,y,z,w)   \\ & 
	-\int_{M^{6}} f^{(1)}_{\varphi}[0]_i(x)\mathcal{G}_{\varphi}^{ij}(x,y) \Big(d^{(3)}\mathcal{L}_{\varphi}[0]_{jj_1i_1}(y,y_1,x_1,) \mathcal{G}_{\varphi}^{- \space\ k j_1}(z,y_{1}) g^{(1)}_{\varphi}[0]_{k}(z) \mathcal{G}_{\varphi}^{i_1 l}(x_{1},w)h^{(1)}_{\varphi}[0]_{l}(w) \\ & \ \ \ \ \ \ \ \ \  
	+d^{(3)}\mathcal{L}_{\varphi}[0]_{j j_{1} i_1}(y_1,x_1,y) \mathcal{G}_{\varphi}^{k j_1}(z,y_1) g^{(1)}_{\varphi}[0]_{k}(z) \mathcal{G}^{+ \space\ i_1 l}_{\varphi}(x_1,w)h^{(1)}_{\varphi}[0]_l(w) \Big)d\mu_g(x,y,z,w,x_1,y_1).
\end{aligned}
$$
Summing over cyclic permutations of the first two terms yields 
$$
\begin{aligned}
	\int_{M^{4}} & \Big( \textcolor{blue}{ f^{(1)}_{\varphi}[0]_i(x)\mathcal{G}_{\varphi}^{ij}(x,y) g^{(2)}_{\varphi}[0]_{jk}(y,z) \mathcal{G}_{\varphi}^{kl}(z,w) h^{(1)}_{\varphi}[0]_l(w)}   \\ & 
	 +\textcolor{red}{f^{(1)}_{\varphi}[0]_i(x)\mathcal{G}_{\varphi}^{ij}(x,y)g^{(1)}_{\varphi}[0]_k(z) \mathcal{G}^{kl}_{\varphi}(z,w) h^{(2)}_{\varphi}[0]_{jl}(y,w)}  \\ & 
	 +\textcolor{red}{g^{(1)}_{\varphi}[0]_i(x)\mathcal{G}_{\varphi}^{ij}(x,y) h^{(2)}_{\varphi}[0]_{jk}(y,z) \mathcal{G}_{\varphi}^{kl}(z,w) f^{(1)}_{\varphi}[0]_l(w) } \\ & 
	 +\textcolor{violet}{g^{(1)}_{\varphi}[0]_i(x)\mathcal{G}_{\varphi}^{ij}(x,y)h^{(1)}_{\varphi}[0]_k(z) \mathcal{G}^{kl}_{\varphi}(z,w) f^{(2)}{\varphi}[0]_{jl}(y,w)}  \\ &
	 +\textcolor{violet}{h^{(1)}_{\varphi}[0]_i(x)\mathcal{G}_{\varphi}^{ij}(x,y) f^{(2)}{\varphi}[0]_{jk}(y,z) \mathcal{G}_{\varphi}^{kl}(z,w) g^{(1)}_{\varphi}[0]_l(w) } \\ & 
	 +\textcolor{blue}{h^{(1)}_{\varphi}[0]_i(x)\mathcal{G}_{\varphi}^{ij}(x,y)f^{(1)}_{\varphi}[0]_k(z) \mathcal{G}^{kl}_{\varphi}(z,w) g^{(2)}_{\varphi}[0]_{jl}(y,w)}\Big)d\mu_g(x,y,z,w)\\
	 &=0,
\end{aligned}
$$
while for the other two,
$$
\begin{aligned}
& \int_{M^{6}}d\mu_g(x,y,z,w,x_1,y_1)\\
	&\Big(\textcolor{blue}{f^{(1)}_{\varphi}[0]_i(x)\mathcal{G}_{\varphi}^{ij}(x,y) d^{(3)}\mathcal{L}_{\varphi}[0]_{jj_1x_1}(y,y_1,i_1,) \mathcal{G}_{\varphi}^{- \ k j_1}(z,y_{1}) g^{(1)}_{\varphi}[0]_{k}(z) \mathcal{G}_{\varphi}^{i_1 l}(x_{1},w)h^{(1)}_{\varphi}[0]_{l}(w)} \\ & 
	+\textcolor{red}{f^{(1)}_{\varphi}[0]_i(x)\mathcal{G}_{\varphi}^{ij}(x,y) d^{(3)}\mathcal{L}_{\varphi}[0]_{j j_{1} i_1}(y,y_1,x_1) \mathcal{G}_{\varphi}^{k j_1}(z,y_1) g^{(1)}_{\varphi}[0]_{k}(z) \mathcal{G}^{+ \space\ i_1 l}_{\varphi}(x_1,w)h^{(1)}_{\varphi}[0]_l(w)} \\ & 
	+\textcolor{red}{g^{(1)}_{\varphi}[0]_i(x)\mathcal{G}_{\varphi}^{ij}(x,y) d^{(3)}\mathcal{L}_{\varphi}[0]_{jj_1i_1}(y,y_1,x_1,) \mathcal{G}_{\varphi}^{- \ k j_1}(z,y_{1}) h^{(1)}_{\varphi}[0]_{k}(z) \mathcal{G}_{\varphi}^{i_1 l}(x_{1},w)f^{(1)}_{\varphi}[0]_{l}(w)} \\ & 
	+\textcolor{violet}{g^{(1)}_{\varphi}[0]_i(x)\mathcal{G}_{\varphi}^{ij}(x,y) d^{(3)}\mathcal{L}_{\varphi}[0]_{j j_{1} i_1}(y,y_1,x_1) \mathcal{G}_{\varphi}^{k j_1}(z,y_1) h^{(1)}_{\varphi}[0]_{k}(z) \mathcal{G}^{+ \ i_1 l}_{\varphi}(x_1,w)f^{(1)}_{\varphi}[0]_l(w) }\\ & 
	+\textcolor{violet}{h^{(1)}_{\varphi}[0]_i(x)\mathcal{G}_{\varphi}^{ij}(x,y) d^{(3)}\mathcal{L}_{\varphi}[0]_{jj_1i_1}(y,y_1,x_1,) \mathcal{G}_{\varphi}^{- \ k j_1}(z,y_{1}) f^{(1)}_{\varphi}[0]_{k}(z) \mathcal{G}_{\varphi}^{i_1 l}(x_{1},w)g^{(1)}_{\varphi}[0]_{l}(w)} \\ & 
	+\textcolor{blue}{h^{(1)}_{\varphi}[0]_i(x)\mathcal{G}_{\varphi}^{ij}(x,y) d^{(3)}\mathcal{L}_{\varphi}[0]_{j j_{1} i_1}(y,y_1,x_1) \mathcal{G}_{\varphi}^{k j_1}(z,y_1) f^{(1)}_{\varphi}[0]_{k}(z) \mathcal{G}^{+ \ i_1 l}_{\varphi}(x_1,w)g^{(1)}_{\varphi}[0]_l(w)} \Big)\\
	&=0.
\end{aligned}
$$
To make the simplifications we used the antisymmetry of the integral kernel $\mathcal{G}_{\varphi}(x,y)$, the adjoint relation between the propagators $\mathcal{G}_{\varphi}^+(x,y)=\mathcal{G}_{\varphi}^-(y,x)$ (see Lemma \ref{lemma_1_duality_green}) and $\mathcal{G}_{\varphi}^{ij}=\mathcal{G}_{\varphi}^{ji}$.
\end{proof}

%%%%%%%%%%%%%%%%%%%%%%% STRUCTURE %%%%%%%%%%%%%%%%%%%%%%%%%%%%%

\section{Structure of the space of microcausal functionals}\label{section_properties_of_muc_functionals}

The first point of emphasis is to give a topology to $\mathcal{F}_{\mu c}(B,\mathcal{U},g)$. % We shall proceed step by step refining our starting definitions to better grasp the reasoning behind the choice of topology we will be giving $\mathcal{F}_{\mu c}(B,\mathcal{U},g)$.
The simplest guess, as well as the weakest, on $\mathcal{F}_{\mu c}(B,\mathcal{U},g)$ is the initial topology induced by the mappings
$$
 F \rightarrow F(\varphi) \in \mathbb{R}.
$$
Taking into account smoothness of functionals, we could refine the above topology by requiring continuity of
\begin{align*}
	&F \rightarrow F(\varphi) \in \mathbb{R},\\
	&F \rightarrow \nabla^{k}F_{\varphi}[0] \in \Gamma^{-\infty}_c\left(M^k\leftarrow\boxtimes^k \left(\varphi^{*}VB\right)\right).
\end{align*}
This time we are leaving out all information on the wave front set which plays a role in defining microcausal functionals. To remedy this, we would like to set up the H\"ormander topology on the spaces $\Gamma^{-\infty}_{ \Upsilon_{k,g}}\left(M^k\leftarrow\boxtimes^k \left(\varphi^{*}VB\right)\right)$. This is, however, not immediately possible since the sets $\Upsilon_{k,g}$ are open cones, and the H\"ormander topology requires closed ones. To tackle this issue we need the following result:

\begin{lemma}\label{lemma_1_cones}
Given the open cone $\Upsilon_k(g)$ it is possible to find a sequence of closed cones $\lbrace \mathcal{V}_{m}(k) \subset T^{*}M^k \rbrace_{m \in \mathbb{N}}$ such that $\mathcal{V}_{m}(k) \subset \mathrm{Int}(\mathcal{V}_{m+1}(k) )$ and $\cup_{m \in \mathbb{N}}\mathcal{V}_{m}(k)= \Upsilon_k(g)$ for all $k \geq 1$. 
\end{lemma}
We refer to \cite[Lemma 4.1]{acftstructure} for the proof of the above result. We are then able to topologize the distributional space $\Gamma^{-\infty}_{c\ \Upsilon_{k,g}}\left(M^k\leftarrow\boxtimes^k \left(\varphi^{*}VB\right)\right)$.

\begin{lemma}\label{lemma_1_mucaus_top_1}
The topology on $\Gamma^{-\infty}_{c\ \Upsilon_{k,g}}\left(M^k\leftarrow\boxtimes^k \left(\varphi^{*}VB\right)\right)$ induced as a direct limit topology of the spaces $\Gamma^{-\infty}_{c\ \mathcal{V}_{m}(k)}\left(M^k\leftarrow\boxtimes^k \left(\varphi^{*}VB\right)\right)$ each possessing the H\"ormander topology is a Hausdorff nuclear locally convex space.
\end{lemma}

\begin{proof}
By Lemma \ref{lemma_1_cones}, we have
\begin{align}
	\Gamma^{-\infty}_{c\ \Upsilon_k(g)}\left(M^k\leftarrow\boxtimes_k \varphi^{*}VB\right)  = \underrightarrow{\lim_{m\in \mathbb{N}}} \Gamma^{-\infty}_{c\ \mathcal{V}_{m}(k)}\left(M^k\leftarrow\boxtimes_k \varphi^{*}VB\right).
\end{align}
By construction of the direct limit we have mappings 
$$
	\Gamma^{-\infty}_{c\ \mathcal{V}_{m}(k)}\left(M^k\leftarrow\boxtimes^k \left(\varphi^{*}VB\right)\right)  \rightarrow \Gamma^{-\infty}_{c\ \Upsilon_k(g)}\left(M^k\leftarrow\boxtimes^k \left(\varphi^{*}VB\right)\right)
$$
where each source space can be given the H\"ormander topology. In particular, by the remark after \cite[Theorem 18.1.28]{hormanderIII}, when dealing with standard compactly supported distributions, the limit topology can be defined to be the initial topology with respect to the mappings
\begin{align*}
	&F \rightarrow F(\varphi) \in \mathbb{C},\\
	&F \rightarrow P F \in \Gamma^{\infty}_{c}\left(M^k\leftarrow \boxtimes^k \left(\varphi^{*}VB\right)\right)
\end{align*}  
where $\varphi$ is any smooth section of $B \to M$ and $P$ any properly supported pseudo-differential operator of order zero on the vector bundle $\boxtimes^k \left(\varphi^{*}VB\right) \rightarrow M^k$ such that $\mathrm{WF}(P) \cap \mathcal{V}_m(k)=\emptyset$. Generalizing to vector bundles \cite[Definition 18.1.32, Theorem 18.1.16]{hormanderIII} and \cite[Theorem 8.2.13]{hormanderI} using \eqref{eq_1_def_WF_sections} for the wave front set of vector valued distributions; we can argue as in \cite[Corollary 4.1 pp. 50]{acftstructure} that each $\Gamma^{-\infty}_{c\ \mathcal{V}_{m}(k)}\left(M^k\leftarrow\boxtimes^k\left(\varphi^{*}VB\right)\right)$ is a Hausdorff topological space.
By Theorem \ref{thm_1_Gamma_c_TVS}, since the base manifold $M$ is separable and the fibers are finite dimensional vector spaces, $\Gamma^{\infty}_{c}\left(M^k\leftarrow \boxtimes^k \left(\varphi^{*}VB\right)\right) $ is a nuclear Hausdorff LF space, thus each dual
$$
	\Gamma^{-\infty}_{c\ \mathcal{V}_{m}(k)}\left(M^k\leftarrow\boxtimes^k \left(\varphi^{*}VB\right)\right)
$$ 
is nuclear as well. Finally, by \cite[Proposition 50.1 pp. 514]{treves2016topological} the direct limit topology on 
$$
    \Gamma^{-\infty}_{c\ \Upsilon_{k,g}} \left(M^k\leftarrow\boxtimes^k(\varphi^{*}VB)\right)
$$ 
is nuclear for all $k$ (and also Hausdorff).
\end{proof}

% Finally we can induce on $\mathcal{F}_{\mu c}(B,\mathcal{U},g)$ a topology by 

\begin{theorem}\label{thm_1_mucaus_top}
Given the set $\mathcal{F}_{\mu c}(B,\mathcal{U},g)$, consider the mappings 
\begin{equation}\label{eq_1_point_seminorm}
	\mathcal{F}_{\mu c}(B,\mathcal{U},g) \ni F \mapsto F(\varphi) \in \mathbb{R},
\end{equation}
\begin{equation}\label{eq_1_hormander_seminorm}
	\mathcal{F}_{\mu c}(B,\mathcal{U},g) \ni F  \mapsto \nabla^{k}F_{\varphi}[0] \in \Gamma^{-\infty}_{c\ \Upsilon_{k,g}}\left(M^k\leftarrow\boxtimes^k \left(\varphi^{*}VB\right)\right),
\end{equation}
and the related initial topology on $\mathcal{F}_{\mu c}(B,\mathcal{U},g)$. Then $\mathcal{F}_{\mu c}(B,\mathcal{U},g)$ is a nuclear locally convex topological space with a Poisson *-algebra with respect to the Peierls bracket of some microlocal generalized Lagrangian $\mathcal{L}$.
\end{theorem}

\begin{proof}
Nuclearity follows from the stability of nuclear spaces under projective limit topology (see \textit{e.g.} \cite[Proposition 50.1 pp. 514]{treves2016topological}) using the nuclearity of both $\Gamma^{-\infty}_{c\ \Upsilon_{k,g}}\left(M^k\leftarrow\boxtimes^k \left(\varphi^{*}VB\right)\right)$ (by Lemma \ref{lemma_1_mucaus_top_1}) and $\mathbb{R}$ (trivially). The Peierls bracket is well defined by Theorem \ref{thm_1_peierls_closedness} and satisfies the Jacobi identity due to Theorem \ref{thm_1_jacobi}, thus we are left with the Leibniz rule, that is 
$$
	\lbrace F,GH \rbrace_{\mathcal{L}}= G\lbrace F,H \rbrace_{\mathcal{L}} + \lbrace F,G \rbrace_{\mathcal{L}}H.
$$
Using $d(F\cdot G)_{\varphi}[0]= dF_{\varphi}[0]G(\varphi)+ F(\varphi)dG_{\varphi}[0]$, Leibniz rule follows if we show that the product $F,G \mapsto F\cdot G$ is closed in $\mathcal{F}_{\mu c}(B,\mathcal{U},g)$. 
$$
	d^k(F\cdot G)_{\varphi}[0](\vec{X}_1, \ldots,\vec{X}_k)= \sum_{\sigma \in \mathcal{P}(1,\ldots ,k)}\sum_{l=0}^k d^lF_{\varphi}[0](\vec{X}_{\sigma(1)}, \ldots,\vec{X}_{\sigma(l)}) d^{k-l}G_{\varphi}[0](\vec{X}_{\sigma(k-l+1)}, \ldots,\vec{X}_{\sigma(k)}),
$$
where $\mathcal{P}(1,\ldots ,k)$ is the set of permutations of $\lbrace 1,\ldots,k \rbrace$. For each of those terms, using \cite[Theorem 8.2.9]{hormanderI}, we have 
$$
\begin{aligned}
	\mathrm{WF}(F^{(l)}_{\varphi}[0]G^{(k-l)}_{\varphi}[0]) & \subset  \mathrm{WF}(F^{(l)}_{\varphi}[0]) \times \mathrm{WF}(G^{(k-l)}_{\varphi}[0]) \\ 
	&\quad  \bigcup \mathrm{WF}(G^{(k-l)}_{\varphi}[0]) \times \left( \mathrm{supp}(G^{(k-l)}_{\varphi}[0])\times \lbrace {0} \rbrace \right) \\ 
	&\quad \bigcup\left( \mathrm{supp}(F^{(l)}_{\varphi}[0])\times \lbrace {0} \rbrace \right)  \times  \mathrm{WF}(G^{(k-l)}_{\varphi}[0]),
\end{aligned}
$$
implying microcausality of $F \cdot G$.
\end{proof}

Notice that closed linear subspaces of $\mathcal{F}_{\mu c}(B,\mathcal{U},g)$ are nuclear as well (see \textit{e.g.} \cite[Proposition 50.1 pp. 514]{treves2016topological}), thus $\mathcal{F}_{\mu loc}(B,\mathcal{U},g)$ is a Hausdorff nuclear space. The space $\mathcal{F}_{\mu c}(B,\mathcal{U},g)$ can be given a $C^{\infty}$-ring structure\footnote{An algebra $\mathcal{A}$ has the $C^{\infty}$-ring structure if, given any $a_1,\ldots,a_n\in \mathcal{A}$, $f\in C^{\infty}(\mathbb K^n,\mathbb K)$, there are mappings $\rho_f:\times^n\mathcal{A}\to \mathcal{A}$ such that if $g\in C^{\infty}(\mathbb K^m,\mathbb K)$, $f_i\in C^{\infty}(\mathbb K^n,\mathbb K)$ with $i=1,\ldots,m$ then $$\rho_h\big(\rho_{f_1}(a_1,\ldots,a_n),\ldots,\rho_{f_m}(a_1,\ldots,a_n)\big)=\rho_{h\circ (f_1,\ldots, f_m)}(a_1,\ldots,a_n).$$The field $\mathbb K$ can either be $\mathbb R$ or $\mathbb C$.}, more precisely

\begin{proposition}\label{prop_1_C-infty_ring}
If $F_1,\ldots,F_n$ $\in \mathcal{F}_{\mu c}(B,\mathcal{U},g)$ and $\psi \in V \subset \mathbb{R}^n \rightarrow \mathbb{R}$ is smooth, then $\psi (F_1,\ldots,F_n) \in \mathcal{F}_{\mu c}(B,\mathcal{U},g)$ and 
$$
	\mathrm{supp}(\psi (F_1,\ldots,F_n)) \subset \bigcup_{i=1}^n \mathrm{supp}(F_i). 
$$
\end{proposition}

\begin{proof}
First we check the support properties. Suppose that $x \notin \cup_{i=1}^n \mathrm{supp}(F_i)$, we can find an open neighborhood $V$ of $x$ for which given any $\varphi \in \mathcal{U}$ and any $\vec{X} \in \Gamma^{\infty}_c(M\leftarrow \varphi^*VB)$ having $\mathrm{supp}(\vec{X}) \subset V$ implies $(F_i\circ u_{\varphi})(t\vec{X} )=(F_i\circ u_{\varphi})(0)$ for all $t$ in a suitable neighborhood of $0\in \mathbb{R}$. Then $ \psi\big((F_1\circ u_{\varphi})(t\vec{X}),\ldots,(F_n\circ u_{\varphi})(t\vec{X})\big)= \psi\big((F_1\circ u_{\varphi})(0),\ldots,(F_n\circ u_{\varphi})(0)\big)$ as well, implying $x\notin \mathrm{supp}(\psi \circ (F_1,\ldots,F_n))$. Due to $(ii)$ Proposition \ref{prop_1_continuity_of_push_forward} and smoothness of $f$, we see that the composition $f \circ (F_1, \ldots, F_n)$ is Bastiani smooth. Its $k$th derivative is
$$\small{
\begin{aligned}
	& d^k\psi(F_1,\ldots,F_1)_{\varphi}[0](\vec{X}_1,\ldots,\vec{X}_k)\\ & =  \sum_{\substack{(J_1,\ldots,J_n)\\ \in \mathcal{P}(1,\ldots,k)}} 
	\frac{\partial^k \psi (F_1(\varphi),\ldots,F_n(\varphi))}{\partial z^{J_1+\ldots+J_n}}\left( F^{(\vert J_1 \vert)}_{1 \space\ \varphi}[0] (\vec{X}_{j_{1,1}} ,\ldots ,\vec{X}_{j_{\vert J_1 \vert,1}})\cdot  \ldots \cdot    F^{(\vert J_n \vert)}_{n \space\ \varphi}[0](\vec{X}_{j_{1,n}},\ldots,\vec{X}_{j_{\vert J_n \vert,n}}) \right)  \\ 
\end{aligned}}
$$
where $J_1, \ldots , J_n$ denotes a partition of $\{1,\ldots,k\}$ into $n$ subsets and $z\in \mathbb R^n$. Since $\psi$ is smooth, the only contribution to the wavefront set of the composition is the product of functional derivatives in the above sum for which, \cite[Theorem 8.2.9]{hormanderI}, gives 
$$
\begin{aligned}
	\mathrm{WF} \left( F_1^{(\vert J_1 \vert)},\ldots, F_n^{(\vert J_n \vert)}\right)
  \subset &  \quad \mathrm{WF}\left( F_{1 \space\ \varphi}^{(\vert J_1 \vert)}\right) \times \ldots \times \mathrm{WF}\left(F_{n \space\ \varphi}^{(\vert J_n \vert)}\right)   \\ 
	&\quad\bigcup \mathrm{supp}\left(F_{1 \space\ \varphi}^{(\vert J_1 \vert)}\right)\times \lbrace \vec{0} \rbrace^{\vert J_1\vert} \times  \mathrm{WF}\left(F_{2 \space\ \varphi}^{(\vert J_2 \vert)}\right) \times \ldots \times \mathrm{WF}\left(F_{n \space\ \varphi}^{(\vert J_n \vert)}\right) \\ 
	&\quad \ldots \\ 
	&\quad\bigcup \mathrm{WF}\left(F_{1 \space\ \varphi}^{(\vert J_1 \vert)}\right) \times \ldots \times \mathrm{WF}\left(F_{n-1 \space\ \varphi}^{(\vert J_{n-1} \vert)}\right) \times \mathrm{supp}\left(F_{n \space\ \varphi}^{(\vert J_n \vert)}\right)\times \lbrace \vec{0} \rbrace^{\vert J_n\vert} \\
	&\quad \ldots \\ 
	& \quad\bigcup\mathrm{supp}\left(F_{1 \space\ \varphi}^{(\vert J_1 \vert)}\right)\times \lbrace \vec{0} \rbrace^{\vert J_1\vert} \times \ldots \times \mathrm{supp}\left(F_{n-1 \space\ \varphi}^{(\vert J_{n-1} \vert)}\right)\times \lbrace \vec{0} \rbrace^{\vert J_{n-1}\vert} \times \mathrm{WF}\left(F_{n \space\ \varphi}^{(\vert J_n \vert)}\right).
\end{aligned}
$$
Therefore, if any element of $\mathrm{\mathrm{WF}} \left( F_1^{(\vert J_1 \vert)},\ldots, F_n^{(\vert J_n \vert)}\right)$ was contained in either $\overline{V}^k_{+,g}$ or $\overline{V}^k_{-,g}$ then at least one of the starting functional would not be microcausal.
\end{proof}

Going through the same calculation for the proof of Proposition \ref{prop_1_C-infty_ring} we get the expression for the Peierls bracket of this composition:
\begin{equation}
\lbrace \psi(F_1,\ldots,F_n),G \rbrace_{\mathcal{L}}=\sum_{j=1}^n\left( \frac{\partial \psi}{\partial z^j}(F_1,\ldots,F_n)\lbrace F_j,G \rbrace_{\mathcal{L}} \right).
\end{equation}

\begin{remark}
With the topology of Theorem \ref{thm_1_mucaus_top} the space of microcausal functionals lacks sequential completeness. Consider as an example the simpler case where $B=M\times \mathbb{R}$, then let $F:\varphi \in \mathcal{U} \mapsto \int_M \varphi(x)\omega$ for some smooth compactly supported $m$-form $\omega$ over $M$, also let $\lbrace f_n\rbrace$ be a sequence of smooth functions $f_n: \mathbb{R}\rightarrow [0,1]$ supported in $[-2,2]$ and converging pointwise to the characteristic function $\chi_{[-1,1]}$ of $[-1,1]$. Sequences of derivatives of $f_n$ all converge pointwise to the zero function on $\mathbb{R}$. If we define $F_n(\varphi)=f_n \circ F(\varphi)$, then $F_n(\varphi)\rightarrow \chi_{[-1,1]} \circ F(\varphi)$ and each
$$
	F^{(k)}_{n \space\ \varphi}[0](\psi_1,\ldots,\psi_k)=f^{(k)}_n(F(\varphi))\int_M \psi_1(x)\omega(x) \ldots \int_M\psi_k(x) \omega(x)
$$
converges pointwise to the zero functional in the topology of Theorem \ref{thm_1_mucaus_top}. However, $\chi_{[-1,1]} \circ F(\cdot)$ is not even continuous, let alone microcausal.
\end{remark}

\begin{proposition}
    We can endow $\mathcal{F}_{\mu c}(B,\mathcal{U},g)$ with a nuclear, locally convex topology $\tau_{sc}$, called the strong convenient topology, generated by strengthening \eqref{eq_1_point_seminorm}, \eqref{eq_1_hormander_seminorm} with the seminorms
\begin{equation}\label{eq_1_strong_conv_seminorm}
\begin{aligned}
    \| F \|_{k, I, \varphi, \mathcal{B}} \doteq \sup_{\substack{t\in I\subset \mathbb{R}\\ \gamma\in C^{\infty}(\mathbb{R},\mathcal{U}_0)\\(\Vec{X}_1,\ldots,\Vec{X}_k)\in \mathcal{B}\subset \boxtimes^k \Gamma^{\infty}(M\leftarrow \varphi^*VB)}} \big\vert \nabla^{(k)}F[\gamma(t)](\Vec{X}_1,\ldots,\Vec{X}_k)\big\vert,
\end{aligned}
\end{equation}
where $k\in \mathbb N$, $I\subset \mathbb R$ is a compact interval, $\varphi\in \mathcal{U}$ and $\mathcal{B}\subset \boxtimes^k \Gamma^{\infty}(M\leftarrow \varphi^*VB)$ is a closed, bounded subset.
\end{proposition}

\begin{proof}[Sketch of a proof]
Similarly to Lemma \ref{lemma_1_loc_bornology}, we rely on \cite[Remark 4.3 pp.53]{acftstructure} for the full discussion and highlight the key differences in our setting. We start with the following observations: 
\begin{itemize}
    \item in any $CO$-open subset $\mathcal{U}$ of $\Gamma^{\infty}(M\leftarrow B)$ the Bastiani smooth curves $\gamma:I\subseteq \mathbb{R} \to \mathcal{U}$ are precisely the conveniently smooth curves from $I$ to $\mathcal{U}$;
    \item each $\mathcal{U}$ can be decomposed as the disjoint union of CO-open subsets $\sqcup_{\varphi\in \mathcal{U}} \mathcal{U}\cap\mathcal{V}_{\varphi} $, each of which is topologically isomorphic to an open subset of $\Gamma^{\infty}_c(M\leftarrow \varphi^*VB)$ with the LF topology;
    \item by Corollary \ref{coro_1_WO^0-curves} we see that each smooth curve $\gamma:I\to \mathcal{U}$ is just valued in some $\mathcal{V}_{\varphi}$ for some $\varphi\in \mathcal{U}$. 
    \item given $\varphi_0\in \mathcal{U}$, we can consider the ultralocal chart representation $F_{\varphi_0}:\mathcal{U}_0\equiv u_{\varphi_0}(\mathcal{U}\cap \mathcal{U}_{\varphi_0})\to \mathbb{R}$ of any functional $F\in \mathcal{F}_c(\mathcal{U},B)$;
    \item  similarly to Remark \ref{rmk_1_extension_of_functionals}, using the fact that the functional has compact support we can extend it to $\widetilde{F}_{\varphi_0}=:\widetilde{\mathcal{U}}_{0}\subset \Gamma^{\infty}(M\leftarrow \varphi_0^*VB) \to \mathbb{R}$ where $\widetilde{\mathcal{U}}_{0}$ is $CO$-open;
    \item if we equip $\Gamma^{\infty}(M\leftarrow \varphi_0^*VB)$ with the $CO$-open topology, then it becomes a Fréchet space {(the compact-open topology is equivalent to the topology of uniform convergence on compact subsets)}, thus, on $\widetilde{\mathcal{U}}_{0}$, the notions of Bastiani smoothness and convenient smoothness do coincide by \cite[(1) Theorem 4.11 pp. 39]{kriegl1997convenient} and \cite[Theorem 1, pp. 71]{frolicher2006smooth}.
\end{itemize}

Next we claim that we can write
\begin{equation}
    C^{\infty}(\widetilde{\mathcal{U}}_{0},\mathbb{R})=\lim_{\substack{\longleftarrow \\ \widetilde{\gamma}\in C^{\infty}(\mathbb{R}, \widetilde{\mathcal{U}}_0)}}C^{\infty}(\mathbb{R},\mathbb{R})=\bigg\{ \{\widetilde{F}\}_{\widetilde{\gamma}} \in \prod_{\widetilde{\gamma}\in C^{\infty}(\mathbb{R},\widetilde{\mathcal{U}}_0)} C^{\infty}(\mathbb{R},\mathbb{R}) : \widetilde{F}_{\widetilde{\gamma}}\circ \kappa =\widetilde{F}_{\widetilde{\gamma}\circ \kappa}\bigg\}
\end{equation}
where the inverse limit is taken with respect to the pre-order $\gamma\leq \gamma'$ if and only if there is $\kappa \in C^{\infty}(\mathbb{R},\mathbb{R})$ with $\gamma=\gamma'\circ \kappa$. To wit, notice that any mapping $F_{\gamma}$ such that $F_{\gamma}\circ g =F_{\gamma\circ g}$ for all reparametrization $ g\in C^{\infty}(\mathbb{R},\mathbb{R})$ gives rise to a mapping $F: \widetilde{U}_0\to Y$ by setting $F(\varphi)=F_{\gamma_{\varphi}}(t)$ where $\gamma_{\varphi}:t\mapsto \varphi \in \widetilde{U}_0$ is the constant curve if $g$ is any reparametrization, then $F_{\gamma_{\varphi}}\circ g =F_{\gamma_{\varphi}\circ g}$ but $\gamma\circ g (t)\equiv \varphi$ thus $F$ is not altered; finally $F$ is smooth since $F\circ \gamma \equiv F_{\gamma}\in  C^{\infty}(\mathbb{R},\mathbb R)$. On the other hand any smooth function $F:\widetilde{U}_0\to \mathbb R$ gives rise to $\{F_{\gamma}\}_{\gamma}$ by setting $F_{\gamma}=F\circ \gamma$, then $F_{\gamma}\circ g= F\circ \gamma\circ g  =F_{\gamma\circ g}$ for any $ g\in C^{\infty}(\mathbb{R},\mathbb{R})$.

We induce on $C^{\infty}(\widetilde{\mathcal{U}}_0,\mathbb{R})$ the initial topology from the Fréchet space topology on $C^{\infty}(\mathbb{R},\mathbb{R})$ through the pullbacks $\gamma^{*}$. This is a nuclear and complete topology (see the discussion in remark 4.3 pp. 55 on \cite{acftstructure}); finally, since $\mathcal{F}_c(B,\mathcal{U}_0)$ is closed in $C^{\infty}(\widetilde{\mathcal{U}}_0,\mathbb{R})$, nuclearity and completeness are inherited in the quotient topology. This space is even a locally convex topological vector space with the seminorms
\begin{equation}\label{eq_1_strong_conv_seminorms_bundle}
    \sup_{\substack{t\in I\subset \mathbb{R}\\ \gamma\in C^{\infty}(\mathbb{R},\mathcal{U}_0)\\(\Vec{X}_1,\ldots,\Vec{X}_k)\in \mathcal{B}\subset \boxtimes^k \Gamma^{\infty}(M\leftarrow \varphi^*VB)}} \big\vert \nabla^{(k)}F[\gamma(t)](\Vec{X}_1,\ldots,\Vec{X}_k)\big\vert,
\end{equation}
where $I$, $\mathcal{B}$ are as in \eqref{eq_1_strong_conv_seminorm}.%, then we can evaluate $ \nabla^{(k)}F[\gamma(t)](\Vec{X}_1,\ldots,\Vec{X}_k)$ by taking a properly chosen cutoff $\chi\in C^{\infty}_c(M)$ and calculating $\nabla^{(k)}F[\gamma(t)](\chi\Vec{X}_1,\ldots,\chi\Vec{X}_k)$.

Finally, we can induce a topology $\tau_{\mathrm{sc}}$, which we call the \textit{strong convenient} topology, on $\mathcal{F}_{\mu c}(B,\mathcal{U},g)$ by substituting the seminorms \eqref{eq_1_strong_conv_seminorms_bundle} in place of \eqref{eq_1_point_seminorm} and \eqref{eq_1_hormander_seminorm}. $\tau_{\mathrm{sc}}$ will enjoy the following additional properties:

\begin{itemize}
    \item[(a)] it remains a nuclear locally convex space topology;
    \item[(b)] it will have a well controlled and nuclear\footnote{See Proposition 5.3.1 in \cite{nlcs}.} completion, \textit{i.e.} its completion is equivalent to the initial topology induced from the completion of the spaces $\Gamma^{-\infty}_{c\ \Upsilon_{k,g}}(M\leftarrow \varphi^*VB)$;
    \item[(c)] the Poisson *-algebra and $C^\infty$-ring operations are continuous and remain such when passing to the completion described in (b), thanks to the results in \cite{brouder2014continuity}.
\end{itemize}
\end{proof}

Prior to the next result, let us recall some notions from \cite[$\S$ 16]{kriegl1997convenient}. A topological space $(X,\tau)$ is Lindel\"of if given any open cover of $X$ there is a countable open subcover, it is separable if it admits a countable dense subset, and it is second countable if it admits a countable basis for the topology.\\
Let $X$ be a Hausdorff locally convex topological space, possibly infinite dimensional, and $S\subset C(X,\mathbb{R})$ a subalgebra. We say that $X$ is $S$-\textit{normal} if  for all closed disjoint subsets $A_0,\ A_1$ of $X$ there is some $f \in S$ such that $f|_{A_i}=i$, while we say it is $S$-\textit{regular} if for any neighborhood $U$ of a point $x$ there exists a function $f\in S $ such that $f(x)=1$ and $\mathrm{supp}(f)\subset U$. A $S$-\textit{partition of unity} is a family $\lbrace \psi_j \rbrace_{j \in J}$ of mappings $S\ni \psi_j:X \rightarrow \mathbb{R}$ with
\begin{itemize}
\item[$(i)$] $\psi_j(x) \geq 0$ for all $j\in J$ and $x\in X$;
\item[$(ii)$] the set $\lbrace \mathrm{supp}(\psi_j): j\in J \rbrace$ is a locally finite covering of $X$,
\item[$(iii)$] $\sum_{j \in J} \psi_j(x)=1$ for all $x \in X$. 
\end{itemize}
When $X$ admits such partition we say it is $S$-\textit{paracompact}.

\begin{proposition}\label{prop_1_mu_caus_top_prop}
The following facts hold true:
\begin{itemize}
\item[$(i)$] Given any $\mathcal{U}\subset \Gamma^{\infty}(M\leftarrow B)$ $CO$-open and any $\varphi_0 \in \mathcal{U}$ there is some $F \in \mathcal{F}_{\mu c}(B,\mathcal{U},g)$ such that $F(\varphi_0)=1$, $0 \leq F|_{\mathcal{U}}\leq 1$ and $F|_{\Gamma^{\infty}(M\leftarrow B) \backslash \mathcal{U}_{\varphi_0}}=0$, \textit{i.e.} $\mathcal{U}$ is $ \mathcal{F}_{\mu c}(B,\varphi_{0},g)$-regular.
\item[$(ii)$] Any $\mathcal{U}\subset \Gamma^{\infty}(M\leftarrow B)$ $CO$-open admits locally finite partitions of unity belonging to $\mathcal{F}_{\mu c}(B,\mathcal{U},g)$, \textit{i.e.} $\mathcal{U}$ is $ \mathcal{F}_{\mu c}(B,\varphi_{0},g)$-paracompact.
\item[$(iii)$] Given any $\mathcal{U}\subset \Gamma^{\infty}(M\leftarrow B)$ $CO$-open, the algebra $\mathcal{F}_{\mu c}(B,\mathcal{U},g)$ separates the points of $\mathcal{U}$ \textit{i.e.} if $\varphi_1\neq \varphi_2 $ there is a microcausal functional $F$ that has $F(\varphi_1)\neq F(\varphi_2)$. 
%\item[(iv)] Given any $\mathcal{U}\subset \Gamma^{\infty}(B)$ CO-open, then any unital $*$-morphism $H:\mathcal{F}_{\mu caus}(B,\varphi_{\alpha},g,\mathcal{U})\rightarrow \mathbb{C}$ is given by the evaluation functional at a unique $\varphi \in \mathcal{U}$.
%\item[(v)] Let $\mathcal{U}$, $\mathcal{W}\subset \Gamma^{\infty}(B)$ be CO-open subsets, then any continuous unital $*$-morphism $\alpha:\mathcal{F}_{\mu caus}(B,\varphi_{\alpha},g,\mathcal{U}) \rightarrow \mathcal{F}_{\mu caus}(B,\varphi_{\alpha},g,\mathcal{W})$, is the pullback of a unique smooth map $\alpha^{*}:\mathcal{W} \rightarrow \mathcal{U}$.
\end{itemize}
\end{proposition}

\begin{proof}
To show $(i)$, take any chart $(\mathcal{U}_{\varphi_0},u_{\varphi_0})$ and consider the open subset $\mathcal{U} \cap \mathcal{U}_{\varphi_0}$, fix some compact set $K \subset M$ and some $\omega \in \Gamma^{\infty}_c(M\leftarrow \varphi^{*}VB'\otimes \Lambda_m(M))$ with $\mathrm{supp}(\omega) \subset K$. Define a functional 
$$
	G_{\omega}:\mathcal{U}\cap \mathcal{U}_{\varphi_0} \ni \varphi \mapsto G_{\omega}(\varphi)=\int_M \omega(u_{\varphi_{0}}(\varphi)). 
$$
By construction $G(\varphi_0)=0$. Let now $\mathcal{W}=\lbrace \varphi \in \mathcal{U}\cap \mathcal{U}_{0} : G(\varphi)<\epsilon^2 \rbrace$ for some constant $\epsilon$, then if $\chi\in C^{\infty}_c(\mathbb R)$ with $0\leq \chi \leq 1$, $\chi(t)=0$ for $|t|\geq 1$, $\chi(t) = 1$ for $t\in [-\frac{1}{2},\frac{1}{2}]$. Consider the new functional $F=\chi \circ (\frac{1}{\epsilon^2}G)$, since $G$ is microlocal, $\chi$ is smooth, by Propositions \ref{prop_1_muloc_into_mucaus} and \ref{prop_1_C-infty_ring}, $F$ is microcausal. Outside $\mathcal{W}$, $F$ is identically zero so we can smoothly extend it to zero over the rest of $\mathcal{U}$ to a new functional, which we denote always by $F$, that has the required properties. We first show that $(ii)$ holds for $U_{\varphi}$. Using the chart $(U_{\varphi},u_{\varphi})$, we can identify $U_{\varphi}$ with an open subset of $\Gamma^{\infty}_c(M\leftarrow \varphi^{*}VB)$. If we show that $U_{\varphi}$ is Lindel\"of and is $ \mathcal{F}_{\mu c}(B,\varphi_{0},g)$-regular, then we can conclude by \cite[Theorem 16.10 pp. 171]{kriegl1997convenient}. $ \mathcal{F}_{\mu c}(B,\mathcal{U},g)$-regularity was point $(i)$ while the Lindel\"of property follows from Theorem \ref{thm_1_Gamma_c_TVS}. Now we observe that any $\mathcal{U}$ can be obtained as the disjoint union of subsets $\mathcal{V}_{\varphi_0} \doteq \lbrace \psi \in \mathcal{U}: \mathrm{supp}_{\varphi_0}(\psi) \ \mathrm{is} \ \mathrm{compact}\rbrace$. Each of those is Lindel\"of and metrizable by Theorem \ref{thm_1_Gamma_c_TVS}, so given the open cover $\lbrace \mathcal{U}_{\varphi}\rbrace_{\varphi \in \mathcal{V}_{\varphi_0}}$, we can extract a locally finite subcover where each elements admits a partition of unity and then construct a partition of unity for the whole $\mathcal{V}_{\varphi_0}$. The fact that $\mathcal{U}= \sqcup \mathcal{V}_{\varphi_0}$ implies that the final partition of unity is the union of all others. Finally for $(iii)$ just take $\mathcal{U}_{\varphi_1}$, $\mathcal{U}_{\varphi_2}$ and $F$ as in $(i)$ constructed as follows: if $\varphi_2 \in \mathcal{U}_{\varphi_1}$ we choose $\epsilon< G(\varphi_2)$ for which $F(\varphi_1)\neq F(\varphi_2)$, if not then any $\epsilon>0$ suffices.
\end{proof}

\begin{definition}\label{def_1_dyn_ideal}
Let $\mathcal{U} \subset \Gamma^{\infty}(M\leftarrow B)$ be $CO$-open and $\mathcal{L}$ a generalized microlocal Lagrangian. We define the on-shell ideal associated to $\mathcal{L}$ as the subspace $\mathcal{I}_{\mathcal{L}}(B,\mathcal{U},g) \subset \mathcal{F}_{\mu c}(B,\mathcal{U},g)$ whose microcausal functionals are of the form 
\begin{equation}\label{eq_1_ideal_gen}
	F(\varphi)=\vec{X}_{\varphi}\left( E(\mathcal{L})_{\varphi}[0]\right)
\end{equation}
with $X:\mathcal{U}  \rightarrow T\mathcal{U}:\varphi\mapsto (\varphi,\vec{X}_{\varphi})$ is a smooth vector field.
\end{definition}

With our usual integral kernel notation we can also write \eqref{eq_1_ideal_gen} as
\begin{equation}\label{eq_1_ideal_gen kernel}
	F(\varphi)=\int_{M}\vec{X}_{\varphi}^i(x) E(\mathcal{L})_{\varphi}[0]_i(x)d\mu_g(x).
\end{equation}
We stress that functionals of the form \eqref{eq_1_ideal_gen} are those which can be seen as the derivation of the Euler-Lagrange derivative by kinematical vector fields over $\mathcal{U}\subset \Gamma^{\infty}(M\leftarrow B)$.

\begin{proposition}\label{prop_1_Poisson_ideal}
$\mathcal{I}_{\mathcal{L}}(B,\mathcal{U},g) $ is a Poisson $*$-ideal of $ \mathcal{F}_{\mu c}(B,\mathcal{U},g)$.
\end{proposition}

\begin{proof}
If $G \in  \mathcal{F}_{\mu c}(B,\mathcal{U},g)$, $G\cdot F(\varphi)=G(\varphi)X(\varphi)\left( E(\mathcal{L})_{\varphi}[0]\right)$. Then $X'= G\cdot X \in \mathfrak{X}(T\mathcal{U})$ and $G \cdot F$ is in the ideal and is associated with the new vector field $X'$. Finally, we have to show that if $F \in \mathcal{I}_{\mathcal{L}}(B,\mathcal{U},g)$ then also $\lbrace F,G \rbrace_{\mathcal{L}}\in \mathcal{I}_{\mathcal{L}}(B,\mathcal{U},g)$. Fix $\varphi \in \mathcal{U}$, $\vec{Y}_{\varphi}\in T_{\varphi}\mathcal{U}$, by the chain rule
$$
	dF_{\varphi}[0](\vec{Y}_{\varphi})= \int_M \left[ \vec{X}^{(1)}_{\varphi}[0]^i(\vec{Y}_{\varphi})\left(E(\mathcal{L})_{\varphi}[0]_i\right) +\vec{X}_{\varphi}^i \left(E^{(1)}(\mathcal{L})_{\varphi}[0]_i(\vec{Y}_{\varphi})\right)\right]d\mu_g
$$
and 
\begin{align*}
	 \lbrace F,G \rbrace_{\mathcal{L}}(\varphi)&= \left\langle dF_{\varphi}[0],\mathcal{G}_{\varphi} dG_{\varphi}[0] \right\rangle\\ &= 
	\int_M\left[ \vec{X}^{(1)}_{\varphi}[0]^{i}_{j}\left(\mathcal{G}_{\varphi}^{jk} g^{(1)}_{\varphi}[0]_k \right)\left(E(\mathcal{L})_{\varphi}[0]_i \right)+\vec{X}_{\varphi}^i \left( E^{(1)}(\mathcal{L})_{\varphi}[0]_{ij}\left(\mathcal{G}_{\varphi}^{jk} g^{(1)}_{\varphi}[0]_k\right)\right)\right]d\mu_g \\ &= 
	\int_M \left\{ \left[ \vec{X}^{(1)}_{\varphi}[0]^{i}_{j}\left(\mathcal{G}_{\varphi}^{jk} g^{(1)}_{\varphi}[0]_k \right) \right]\left(E(\mathcal{L})_{\varphi}[0]_i\right)\right\}d\mu_g
\end{align*}
where we used that $\mathcal{G}_{\varphi}$ associates to its argument a solution of the linearized equations. Defining $\varphi\mapsto \vec{Z}_{\varphi}=\vec{X}^{(1)}_{\varphi}[0]^{i}_j\left(\mathcal{G}_{\varphi}^{jk} g^{(1)}_{\varphi}[0]_k \right)\partial_i \in \Gamma^{\infty}_c(M \leftarrow \varphi^*VB)$ yields a smooth mapping (by smoothness of $X$, the functional $G$ and the propagator $\mathcal{G}_{\varphi}$) defining the desired vector field.
\end{proof}

\begin{definition}\label{def_1_on_shell_ideal}
Let $\mathcal{U} \subset \Gamma^{\infty}(M\leftarrow B)$ be $CO$-open and $\mathcal{L}$ a generalized microlocal Lagrangian. We define the on-shell algebra on $\mathcal{U}$ associated to $\mathcal{L}$ as the quotient
\begin{equation}
	  \mathcal{F}_{\mathcal{L}}(B,\mathcal{U},g) \doteq \mathcal{F}_{\mu c}(B,\mathcal{U},g)/ \mathcal{I}_{\mathcal{L}}(B,\mathcal{U},g).
\end{equation}
\end{definition}

This accounts for the algebra of observables once the condition $E(\mathcal{L})_{\varphi}[0]=0$ has been imposed on $\mathcal{U}$.

\section{Wave maps}\label{section_wave_maps}

Finally we introduce, as an example of physical theory, wave maps. The configuration bundle is $B=M\times N$, where $M$ is an $m$ dimensional Lorentzian manifold and $N$ an $n$ dimensional manifold equipped with a Riemannian metric $h$. The space of sections is canonically isomorphic to $C^{\infty}(M,N)$ and possess a differentiable structure induced by the atlas $\big\{\big(\mathcal{U}_{\varphi}, u_{\varphi}, \Gamma^{\infty}_c(M\leftarrow \varphi^{*}TN) \big)  \big\}_{\varphi\in C^{\infty}(M,N)} $, where $u_{\varphi}$ is given in \eqref{eq_1_trivial_gamma_local_chart}. %with the only difference being that the sections are $N$-valued mappings, thus $\exp$ can be taken as the exponential function induced by a Riemannian metric $h$ on $N$. 
The generalized Lagrangian for wave maps is 
\begin{equation}\label{eq_1_lag_wave_maps}
	\mathcal{L}_{\mathrm{WM}}(f)(\varphi)=\frac{1}{2}\int_M f(x)  \mathrm{Trace}(g^{-1} \circ (\varphi^*h))(x)d\mu_g(x);
\end{equation}
obtained by integration of the standard geometric Lagrangian density $\lambda= \frac{1}{2}g^{\mu\nu}h_{ij}(\varphi)\varphi^i_{\mu}\varphi^j_{\nu}d\mu_g$ smeared with a test function $f\in C^{\infty}_c(M)$. Computing the first functional derivative, as per \eqref{eq_1_euler_der}, we get the associated Euler-Lagrange equations which, written in jets coordinates, reads
\begin{equation}\label{eq_1_wavemaps_EL_eq}
	h_{ij}g^{\mu \nu} \left( \varphi^i_{\mu \nu} +\lbrace h\rbrace^{i}_{kl} \varphi^k_{\mu}\varphi^l_{\nu}-\lbrace g\rbrace^{\lambda}_{\mu \nu}\varphi^i_{\lambda} \right)=0,
\end{equation}
where we denoted by $\lbrace h\rbrace$, $\lbrace g\rbrace$ the coefficients of the Levi-Civita linear connection associated to $h$ and $g$ respectively. Computation of the second Euler derivative of \eqref{eq_1_lag_wave_maps} yields
\begin{equation}
\begin{aligned}
	\delta^{(1)}E(\mathcal{L}_{WM})_{\varphi}[0] &: \Gamma^{\infty}_c(M\leftarrow \varphi^{*}TN) \times \Gamma^{\infty}_c(M\leftarrow \varphi^{*}TN ) \rightarrow \mathbb{R} \\ (\vec{X},\vec{Y}) \mapsto & \int_M\frac{1}{2}\left[g^{\mu\nu}(x)h_{ij}(\varphi(x)) \nabla_{\mu}{X}^i(x) \nabla_{\nu}{Y}^j(x)+A^{\mu}_{ij}(\varphi(x))\left(\nabla_{\mu}{X}^i(x){Y}^j(x)+\nabla_{\mu}{Y}^i(x){X}^j(x)\right)\right.\\
	& \space\ \space\ \space\ \space\ \space\ + \left. B_{ij}(\varphi(x)){X}^i(x){Y}^j(x) \right]d\mu_g(x)\ ,
\end{aligned}
\end{equation}
where we choose $f\equiv 1$ in a neighborhood of $\mathrm{supp}(\vec{X})\cup \mathrm{supp}(\vec{Y})$ as done before. One can show (see \textit{e.g.} \cite[$\S$ 4.6 pp. 130]{carfora2017quantum}) that the coefficients $A^{\mu}_{ij}$ always vanish and
\begin{equation*}
    \begin{split}
    &\delta^{(1)}E(\mathcal{L}_{WM})_{\varphi}[0] (\vec{X},\vec{Y})\\
    &\qquad =\int_M \frac{1}{2} \big( g^{\mu\nu}(x)h_{ij}(\varphi(x))\nabla_{\mu} {X}^i(x) \nabla_{\nu}Y^j(x)+ R^k_{ilj}(\varphi(x))p^{\alpha}_k(\varphi(x))\varphi^l_{\alpha}(x){X}^i(x)Y^j(x)\big)d\mu_g(x)
\end{split}
\end{equation*}
where $R$ are the components of the Riemann tensor of the Riemannian metric $h$, and $p^{\alpha}_k\doteq \frac{\partial \lambda}{\partial y^{k}_{\alpha}}$ is the conjugate momenta of the Lagrangian density $\lambda$. It is evident that the induced differential operator $D_{\varphi}$ can be locally expressed as
\begin{equation}\label{eq_1_wave_maps_diff_op}
	D_{\varphi}(\vec{X})(x) =\big(g^{\mu\nu}(x)h_{ij}(\varphi(x))\nabla_{\mu \nu} \vec{X}^i(x) + R^k_{ilj}(\varphi(x))p^{\alpha}_k(\varphi(x))\varphi^l_{\alpha}(x){X}^i(x)\big) dy^j\big|_{\varphi(x)}\ ,
\end{equation}
where $dy^j\big|_{\varphi(x)}$ is the vertical differential of the bundle $M\times N$ evaluated at $\varphi(x) \in N$. The associated principal symbol is
$$
	\sigma_2(D_{\varphi})= \frac{1}{2} g^{\mu\nu}\frac{\partial}{\partial x^{\nu}}\frac{\partial}{\partial x^{\mu}}\otimes id_{\varphi^{*}TN}. 
$$
Theorem \ref{thm_1_properties_of_Green_functions} then ensures the existence of the advanced and retarded propagators for Wave Maps $\mathcal{G}^{\pm}_{WM}[\varphi]$. Their difference defines the causal propagator and consequently the Peierls bracket as in Definition \ref{def_1_Peierls}. For greater clarity, let us write the expression of the causal propagator when $m=4$.\\ 

We start by writing a parametrix for the differential operator \eqref{eq_1_wave_maps_diff_op}, we shall use its coefficients to determine the form of the (unique) advanced and retarded fundamental solutions which in turn give the causal propagator. We shall assume, in order to be able to write explicitly the above quantities, to have chosen a suitable neighborhood $\mathcal{O}$ of the diagonal of $M^2$ such that 
\begin{itemize}
    \item for each $(x,y) \in \mathcal{O}$ there is a geodesically convex normal open set $\Omega$ containing both $x$ and $y$; 
    \item If $\Omega, \ \Omega'$ are geodesically convex normal sets, then their intersection is still geodesically convex.
\end{itemize}
This is always possible in view of \cite[Theorem 10 pp. 130]{moretti2021global}.\\

Any distributional section $H \in \Gamma^{-\infty}(M^2 \leftarrow \boxtimes^2 \varphi^*TN)$ can be written, by choosing trivializations on the vector bundle, as
$$
    H_{\varphi}(x,y) = H^{ai}(\varphi(x),\varphi(y)) \partial_a|_{\varphi(x)} \boxtimes \partial_i|_{\varphi(y)}\ .
$$
where each $H^{ai}\in \mathcal{D}'(M\times M)$. $H$ is a \textit{parametrix} if  
$$
    D_{\varphi}(x)H_{\varphi}(x,\cdot) - \delta_x \in C^{\infty}(M)\ , \quad D_{\varphi}(y)H_{\varphi}(\cdot,y) - \delta_y  \in C^{\infty}(M)\ .
$$
For notational convenience we shall omit the subscript $\varphi$ from the operator $D_\varphi$ introduced in \eqref{eq_1_wave_maps_diff_op}. One can show, using \textit{e.g.} \cite[Theorem 4.2.1 pp. 131, Theorem 4.3.1 and Lemma 4.3.2 pp. 142]{friedlander1975wave}, that 
\begin{equation}\label{eq_1_wave_maps_parametrix}
    H^{ai}(x,y) = \frac{1}{2\pi} \frac{U^{ai}(x,y)}{\sigma(x,y)}+ \widetilde V^{ai}(x,y)\ ,
\end{equation}
where $\sigma(x,y)$ is the squared geodesic distance between $x$ and $y$, $\widetilde V^{ai}(x,y) = \sum_{l\geq 0} V^{ai}_l(x,y) \frac{\sigma^l(x,y)}{l!}\rho(c_l \sigma(x,y)) $ with $\rho \in C^{\infty}_c(\mathbb R)$, $c_k \nearrow \infty$ and $U^{ai}, V^{a i} \in C^{\infty}(M\times M)$ are the solution of the differential equations involving the so-called van Vleck-Morette determinant $\Delta$:
$$
\begin{aligned}
    g^{\mu\nu}\nabla_{\mu}\sigma\Big( \nabla_{\nu} U^{ai} -\frac{1}{2}\frac{\nabla_{\nu}\Delta}{\Delta}U^{ai}\Big) & =0,\\
    V_0^{ai} +  g^{\mu\nu}\nabla_{\mu}\sigma\Big( \nabla_{\nu} V_0^{ai} -\frac{1}{2}\frac{\nabla_{\nu}\Delta}{\Delta}V_0^{ai}\Big)& =-\frac{1}{2} D U^{ai},\\
    V_l^{ai} +  \frac{1}{l+1}g^{\mu\nu}\nabla_{\mu}\sigma\Big( \nabla_{\nu} V_l^{ai} -\frac{1}{2}\frac{\nabla_{\nu}\Delta}{\Delta}V_l^{ai}\Big)& =-\frac{1}{2l(l+1)} D V_{l-1}^{ai}\ .
\end{aligned}
$$
The calculations show that $H_\varphi$ is symmetric in its variables, moreover we can equivalently write $D_{\varphi}(x)H_{\varphi}(x,\cdot)- \delta_x \in C^{\infty}(M)$ as
$$
    \int_{M} D_{\varphi \ ab}(x) H^{ai}(x,y) \vec X^b(x) d\mu_g(x) = \vec X^i(y)+ \int_{M}  D_{ab}(x) W^{ai}(x,y) \vec X^b(x) d\mu_g(x) \qquad \forall \vec X \in T_{\varphi}C^{\infty}(M,N)\ ,
$$
where $W\in \Gamma^{\infty}(M^2 \leftarrow \boxtimes^2(\varphi^*TN))$ denotes the smooth reminder which forbids $H$ from being a solution. To get the fundamental solutions from the parametrix we denote by $\delta_{\pm}(\sigma), \theta_{\pm}(\sigma) \in \mathcal{D}'(\mathcal{O})$ the distributions defined by
$$
\begin{aligned}
    \delta_{\pm}(\sigma) & = \lim_{\epsilon \searrow 0} \delta(\sigma \mp \epsilon)\\
    \theta_{\pm}(\sigma)(x,y) &= \begin{cases}
        1 \qquad \mathrm{if} \ x \in J^{\pm}(y)\ ,\\
        0 \qquad \mathrm{if } \ x \notin J^{\pm}(y)\ ;
    \end{cases}
\end{aligned}
$$
then we can show (see \textit{e.g.} \cite[Theorem 4.5.1 pp. 154]{friedlander1975wave}) that 
\begin{equation}
    \mathcal{G}^{\pm}_{\varphi} \doteq \frac{1}{2\pi} \big( U\delta_{\pm}(\sigma) + V_{\pm}\big)\ ,
\end{equation}
with $V_{\pm}\in \Gamma^{\infty}(M^2 \leftarrow \boxtimes^2(\varphi^*TN))$, $\mathrm{supp}(V_{\pm}) =\mathrm{supp}\big(\theta_{\pm}(\sigma)\big) $, are the unique advanced and retarded fundamental solutions. Equivalently, for all $\vec X \in \Gamma^{\infty}_c(M \leftarrow \varphi^*TN)$,
$$
    \int_M D_{ab}(x)\mathcal{G}_{\pm}^{ai}(x,y) \vec X^b(x) d\mu_g(x) = \vec X^i(y)\ .
$$
In particular,
\begin{equation}
\begin{aligned}
    V_{\pm}^{ai}(x,y) = \widetilde V^{ai}(x,y) & + \int_{J^{\pm}(y)\cap J^{\mp}(x)} \widetilde V^{aj}(x,z) L_{j}^i(z,y)d\mu_g(z)\\ 
    & + \int_{J^{\pm}(y)\cap \partial J^{\mp}(x)} \widetilde U^{aj}(x,z) L_{j}^i(z,y)d\mu_{\sigma}(z)\ ,
\end{aligned}
\end{equation}
where $d\mu_{\sigma}\equiv d\mu_g|_{\sigma=0}$, $L\in \Gamma^{\infty}(M^2 \leftarrow \varphi^*TN' \boxtimes \varphi^*TN)$ is the inverse kernel of $D W \in \Gamma^{\infty}(M^2 \leftarrow \varphi^*TN \boxtimes \varphi^*TN')$ (see \cite[Lemma 4.4.2 pp. 149 and Theorem 4.4.2 pp. 150]{friedlander1975wave} for the existence of the inverse kernel), \textit{i.e.} it satisfies
$$
\begin{aligned}
    & \vec X^i(y) + \int_M DW^i_a(x,y) \vec X^a(x) d\mu_g(x) = \vec Y^i(y) \equiv \int_M DH^i_a (x,y) \vec X^a(x) d\mu_g(x)\\
    \Leftrightarrow & \quad \vec X^a(x) = \vec Y^a(x) + \int_M L^a_i (x,y) \vec Y^i(y) d\mu_g(y)\ .
\end{aligned}
$$
Finally, we can write the causal propagator as
\begin{equation}\label{eq_1_wave_maps_propagator}
    \mathcal{G}_{\varphi} = \mathcal{G}_{\varphi}^+-\mathcal{G}_{\varphi}^- = \frac{1}{2\pi}\big[ U \big(\delta_+(\Gamma)-\delta_-(\Gamma)\big) + V_+ - V_-\big]\ .
\end{equation}

The results of Sections \ref{section_peierls_bracket} \ref{section_properties_of_muc_functionals} do apply to wave maps: it is therefore possible to obtain a $*$-Poisson algebra generated by microcausal functionals $\mathcal{F}_{\mu{c}}(M\times N,g)$ and endow the latter space with the topology of a nuclear locally convex space.

\section{Conclusions and Outlook}

With the present paper we have partially explored the generalization of \cite{acftstructure} to the space of configurations which are general fiber bundles. We remark that the main technical difficulties are the lack of a vector space structure for images of fields and the fact that while $C^{\infty}(M)$ is a Fréchet space, $\Gamma^{\infty}(M\leftarrow B)$ is not even a vector space. This forces us to use a manifold structure for $\Gamma^{\infty}(M\leftarrow B)$ and an appropriate calculus as well. For the manifold structure we choose locally convex spaces as modelling topological vector spaces. The notion of smooth mappings is however not unique, for instance,  one could have used the convenient calculus of \cite{kriegl1997convenient}; however, this calculus has the rather surprising property that smooth mappings need not be continuous (see \cite{glockner2005discontinuous} for an example). This is rather annoying in view of the heavy usage of distributional spaces in sections 4 and 5, therefore we deemed more fitting Bastiani calculus (see \cite{mb}, \cite{michal1938differential}). 

As in the case of finite dimensional differentiable manifolds -- where geometric properties can equivalently be described locally but are independent from the local chart used --  here too we tried to establish, where possible, independence from the choice of ultralocal charts. It stands out that the characterization of microlocality by Proposition \ref{porop_1_muloc_charachterization} is, to the best of our knowledge, inherently chart dependent. Notice that a sufficient condition to avoid this unpleasant fact is to provide a combination of \textit{topological conditions} on the bundle and \textit{regularity conditions} on the integral kernel of the first derivative of the functional. A noteworthy question would be which additional hypotheses, if any, could one add to Proposition \ref{porop_1_muloc_charachterization} to make the characterization of microlocal functionals intrinsic.

Another possible interesting point to develop in the future is the time-slice axiom \cite{haag1964algebraic,CF09,brunetti2023anomalies} in this classical setting. Here however we would need more structural insights coming from the full equation of motion, not only the linearized version that we use in the paper (see, \textit{e.g.}\ \cite{acftstructure2}). 

Finally we mention that the definition we use of wave front set for distributional sections of vector bundles, see \ref{eq_WF_distributional_sections}, does leave some space for improvement, in particular one could use the refined notion of polarization wave front sets which appeared in \cite{polar} and attempt to re-derive all important result with this finer notion of singularity.

\section*{Acknowledgements}
We wish to thank Lorenzo Fatibene and Pedro Lauridsen Ribeiro for discussions and useful comments.

%%%%%%%%%%%%%%%%%%%%%%%%%%%%%%  APPENDIX   %%%%%%%%%%%%%%%%%%%%%%%%%%%%%

\appendix

\section{Bastiani calculus}

The notion of calculus in locally convex spaces we need is the so-called Bastiani calculus, its origin can be traced back to \cite{mb,michal1938differential}. Here we introduce the basic definitions and properties.

In the sequel we will use complete locally convex spaces and denote them by capital letters $X$, $Y$, $Z$. 

\begin{definition}
    Let $U\subset X$ be an open subset, a mapping $P:U\subset X \to Y$ is Bastiani differentiable if the following conditions hold:
    \begin{itemize}
        \item[$(i)$] $ \lim_{t\to 0} \frac{1}{t}\big(P(x+tv)-P(x)\big)=dP[x](v)$ exists for all $x\in U$, $v\in X$, giving rise to a mapping $dP:U\times X \to Y$ linear in the second entry; 
        \item[$(ii)$] The mapping $dP:U\times X \to Y, \ (x,v) \mapsto dP[x](v)$ is jointly continuous. 
    \end{itemize}
\end{definition}

\begin{theorem}\label{thm_A_Riemann_curve_integral}
    Let $\gamma:[a,b]\subset \mathbb{R} \to X$ be a continuous curve, then there exists a unique object $\int_a^b \gamma(t)dt\in X$ such that
    \begin{itemize}
        \item[$(i)$] for every continuous linear mapping $l:X\to \mathbb{R}$
        $$
            l\bigg(\int_a^b\gamma(t)dt\bigg) =\int_a^bl(\gamma(t))dt;
        $$
        \item[$(ii)$] for every seminorm $p_i$ on $X$, 
        $$
        p_i\bigg(\int_a^b\gamma(t)dt\bigg) \leq \int_a^bp_i(\gamma(t))dt;
        $$
        \item[$(iii)$] for all continuous curves $\gamma$, $\beta$, $\int_a^b \big(\gamma(t)+ \beta(t)\big)dt=\int_a^b \gamma(t)dt+\int_a^b \beta(t)\big)dt$;  
        \item[$(iv)$] for all $\lambda\in \mathbb{R}$, $\int_a^b \big(\lambda \gamma(t)\big)dt=\lambda\bigg(\int_a^b \gamma(t)dt\bigg)$;
        \item[$(v)$] for all $a\leq c \leq b$, $\int_a^b\gamma(t) dt=\int_a^c\gamma(t) dt  +\int_c^b\gamma(t) dt$. 
    \end{itemize}
\end{theorem}
The proof of this result is quite standard (see \textit{e.g.} \cite[Theorem 2.2.1 pp. 71]{hamilton1979inverse}).%, it essentially follows from defined the integral as a Riemann summation in the interval $[a,b]$ whenever $\gamma$ is a piecewise straight continuous curves. Finally, noting that the latter space is dense in $C([a,b];X)$, we can extend the integral by continuity and get uniqueness by the Hahn-Banach theorem.

\begin{lemma}
    Let $U\subset X$, $V\subset Y$ be open subsets, let $P:U\to Y$ and $Q:V\to Z$ be Bastiani differentiable mappings such that $P(U) \subset V$, then $Q\circ P : U \to Z$ is Bastiani differentiable and $d(Q\circ P)[x](v)= dQ[P(x)]\big(dP[x](v)\big)$ for each $x\in U$, $v\in X$. 
\end{lemma}
\begin{proof}
    We claim that $P$ is Bastiani differentiable if and only if there is a continuous mapping $L:U\times U \times X \to Y$ linear in the third entry such that 
    $$
        P(x_1)-P(x_2)= L[x_1,x_2](x_1-x_2).
    $$
    In particular we have $L[x,x](v)=dP[x](v)$. The necessity condition follows by considering the smooth curve $\gamma(t)=x_1+t(x_2-x_1)$ with $t\in [0,1]$, then $dP[\gamma(t)](v)$ is a smooth curve for each $v\in X$, by \ref{thm_A_Riemann_curve_integral}, define
    $$
        L[x_1,x_2](v)\doteq \int_0^1 dP[\gamma(t)](v)dt.
    $$
    It is clear that $L$ is continuous and linear in the third entry, moreover since $\frac{d}{dt}P(\gamma(t))=dP[\gamma(t)](x_1-x_2)$ we get $L[x_1,x_2](x_1-x_2)= dP[x_1](x_1-x_2)$. For the sufficiency condition just note that $\frac{1}{t}\big(P(x+tv)-P(x)\big)=L(x, x+tv)[v]$, so taking the limit we get our claim. Next suppose that 
    $$
        P(x_1)-P(x_2)= L[x_1,x_2](x_1-x_2),
    $$
    $$
        Q(y_1)-Q(y_2)= M[y_1,y_2](y_1-y_2).
    $$
    Then $Q(P(x+tv))-Q(P(x))=M[P(x+tv),P(x)]\big(P(x+tv)-P(x)\big)=t M[P(x+tv),P(x)]\big(L[x+tv,x](v)\big)$ thus dividing by $t$ and taking the limit yields $d(Q\circ P)[x](v) \equiv M[P(x),P(x)]\big(L[x,x](v)\big)= dQ[P(x)]\big(dP[x](v)\big)$.
\end{proof} 

\begin{definition}\label{def_A_Bastiani_smooth_map}
    A mapping $f:U\subset X \to Y$ is $k$ times Bastiani differentiable if $d^{k-1}f:U\times X \cdots \times X \to Y$ is Bastiani differentiable. The $k$th derivative of $f$ at $x$ is defined by recursion
    \begin{equation}\label{eq_A_Bastiani_k_der}
        d^{k}f[x](v_1,\ldots,v_k) \doteq \lim_{t\to 0} \frac{1}{t}\big( d^{k-1}f[x+tv_k](v_1,\ldots,v_{k-1})-d^{k-1}f[x+tv_k](v_1,\ldots,v_{k-1})\big)  .
    \end{equation}
    % Finally, we denote by $C^k_B(U,Y)$ the set of $k$ times Bastiani differentiable functions $f:U \to Y$. 
\end{definition}

Explicit computation of $d\big( d^{k-1}f\big)$ in $(x,v_1,\ldots,v_{k-1})$ yields
$$
    \begin{aligned}
         & d\big(d^{k-1}f\big)[x,v_1,\ldots,v_{k-1}](v_k,w_1,\ldots,w_{k-1})= \sum_{j=1}^{k-1} d^{k-1}f[x](v_1,\ldots, \widehat{v_j}, w_j, \ldots, v_{k-1})\\ & \qquad + \   \lim_{t\to 0} \frac{1}{t}\big( d^{k-1}f[x+tv_k](v_1,\ldots,v_{k-1})-d^{k-1}f[x+tv_k](v_1,\ldots,v_{k-1})\big) ,
    \end{aligned}
$$
we can then see that $d\big( d^{k-1}f\big)$ is Bastiani differentiable if and only if the limit of \eqref{eq_A_Bastiani_k_der} exists and is a continuous mapping $U\times X^k \to Y$. We can thus state 
\begin{lemma}
    A mapping $f : U \subset X \to Y$ is $k$ times Bastiani differentiable if and only if for each $0\leq j \leq k$ all the derivative mappings $d^{j}f:U \times X^j \to Y$ exist and are jointly continuous. 
\end{lemma}

\section{Topologies on spaces of mappings}

Let $M$, $N$ be finite dimensional paracompact Hausdorff topological spaces, denote the space of continuous functions by $C(M,N)$. The \textit{compact open} topology $\tau_{CO}$ or CO-topology is the topology generated by a basis whose elements have the form 
\begin{equation}\label{eq_1_CO_open}
    N(K,V)=\{ \varphi \in C(M,N) : \varphi(K)\subset V\}, 
\end{equation}
where $K\subset M$ is a compact subset and $V\subset N $ is open. Roughly speaking, this topology controls the behaviour of functions only on small regions of $M$, whereas their behaviour ``at infinity" is not specified.

\begin{lemma}\label{lemma_1_CO_is_Hausdorff}
    Let $M$, $N$ as described above, if $N$ is normal, then $\big( C(M,N), \tau_{CO}\big)$ is Hausdorff. 
\end{lemma}
\begin{proof}
Supposing $\varphi \neq \psi$, then at least $\varphi(x)\neq \psi(x) $ for some $x\in M$. By continuity of $\varphi$, $\psi$ there exists an open subset $U_x$ such that $\varphi(y)\neq \psi(y)$ for each $y\in \bar{U}_x$. Without loss of generality we can suppose that $\overline{U_x}$ is compact, then $\varphi(\overline{U_x})$, $\psi(\overline{U_x})$ are compact and therefore closed. Since $N$ is normal, there are disjoint open subsets $V_{\varphi}$, $V_{\psi}$ respectively containing $\varphi(\overline{U_x})$, $\psi(\overline{U_x})$, then
$$
    N\big(U_x,V_{\varphi} \big)\cap N\big(U_x,V_{\psi} \big)=\emptyset.
$$
\end{proof}

\begin{lemma}\label{lemma_1_CO_mertizability}
Let $M$, $N$ be topological spaces, if $N$ is a complete metric space then $\big( C(M,N), \tau_{CO}\big)$ is a complete metric space as well.
\end{lemma}

If $M$ is compact and $(N,d_N)$ is metric, then a neighborhood of $\varphi$ in the compact-open topology can be given as 
$$
    B_{\epsilon}(\varphi)\doteq\{ \psi \in C(M,N): d_N\big(\varphi(x),\psi(x)\big)<\epsilon(x) \  \forall x\in M\big\},
$$
where $\epsilon:M \to \mathbb{R}_+$ is a continuous function.

% \begin{lemma}
% Let $M$, $N$, $L$ be topological spaces with $M$ locally compact and Hausdorff, then 
% $$
%     C(M\times N, L) \simeq C\big( M, C(N,L)\big)
% $$
% where $C(N,L)$ possesses the CO-topology.
% \end{lemma}

Given $\varphi\in C(M,N)$, let $G_{\varphi}:M\to M\times N$ be the graph mapping associated to $\varphi$, set $\mathrm{Im}(G_{\varphi})\equiv \mathrm{gh}(\varphi)=\{(x,\varphi(x))\in M\times N: x\in M\}$. 

\begin{definition}\label{def_1_WO-top}
    The \textit{wholly open topology} $\tau_{\mathrm{WO}}$ or $\mathrm{WO}$-topology on $C(M,N)$ is generated by a subbasis of open subsets of the form 
    \begin{equation}\label{eq_1_WO-open}
        W(V)=\{ \varphi \in C(M,N): \varphi(M)\subset V\},        
    \end{equation}
    where $V\subseteq N$ is open.
\end{definition}

Note that the $\mathrm{WO}$-topology is not Hausdorff, for it cannot separate surjective functions.

\begin{definition}\label{def_1_WO^0-top}
    The \textit{graph topology} $\tau_{\mathrm{WO}^0}$ or ${\mathrm{WO}}^0$-topology on $C^{\infty}(M,N)$ is the one induced by requiring 
    $$
        G:C(M,N) \ni \varphi \mapsto G_{\varphi} \in \big( C(M,M\times N),\tau_{\mathrm{WO}}\big)
    $$
    to be an embedding.
\end{definition}

By Definition \ref{def_1_WO-top} the open subbasis of $C(M,M\times N)$ is given by subsets of the form
$$
    W(\widetilde{V}) = \{ f\in C(M,M\times N): f(M)\subset \widetilde{V}\}
$$
with $\widetilde{V}\subset M\times  N$ open subsets. When $f=G_{\varphi}$ for some $\varphi\in C(M,N)$, then the trace topology on the subset $G\big(C(M,N)\big)$ is generated by a subbasis of elements $W(\widetilde{V})$ where $\widetilde{V}=M\times V$ with $V\subset N$ open subset. Clearly $G$ is an injective mapping, and bijective onto its image. Therefore a subbasis for the $\mathrm{WO}^0$-topology is given by 
\begin{equation}\label{eq_1_WO^0-open}
    W(\widetilde{V})= \{ \varphi \in C(M,N) : G_{\varphi}\subset M\times V\}
\end{equation}
\begin{lemma}
    The $\mathrm{WO}^0$-topology is finer then the CO-topology and is therefore Hausdorff. 
\end{lemma}
\begin{proof}
We show that $\mathrm{id}_{C(M,N)}: \big( C(M,N),\tau_{\mathrm{WO}^0}\big) \to \big( C(M,N),\tau_{\mathrm{CO}}\big)$ is continuous. Let $N(K,V)$ be an open subset as in \eqref{eq_1_CO_open}, $U_1$, $U_2$ be a cover of $M$ such that $K\subset U_1$ and $U_2 =M\backslash K$. Consider the open subset
$$
    W(U_1\times V \cup U_2 \times N)=\{ \varphi \in C(M,N) : G_{\varphi}\subset U_1\times V \cup U_2 \times N\};
$$
the former is a $\mathrm{WO}^0$-open subset, which is however equal to $N(K,V)$.
\end{proof}

The main difference from the compact open topology is that the $\mathrm{WO}^0$ topology does control the behaviour of a mapping over the whole space, while the $\mathrm{CO}$ topology was limited to a compact region. Notice that although in general the $\mathrm{WO}^0$ topology is finer then the $\mathrm{CO}$ topology, when the manifold $M$ is compact the two become equivalent.

\begin{lemma}\label{lemma_1_WO^0-metric-subbasis}
    Let $M$ be paracompact and $(N,d)$ be a metric space, then a basis of neighborhood of $\varphi\in C(M,N)$ for the $\mathrm{WO}^0$-topology is given by 
    \begin{equation}
        W_{\varphi}(\epsilon)=\big\{ \psi \in C(M,N) : d\big(\varphi(x),\psi(x)\big)<\epsilon(x) \ \forall x\in M\big\},
    \end{equation}
    where $\epsilon:M \to \mathbb{R}_+$ is continuous.
\end{lemma}

\begin{proposition}\label{prop_1_WO^0-convergence}
    Let $M$ be paracompact and $(N,d)$ be a metric space, then for any sequence $\{\varphi_n\}\subset C(M,N)$, the following are equivalent:
    \begin{itemize}
        \item[$(i)$] $\varphi_n \to \varphi$ in the $\mathrm{WO}^0$-topology;
        \item[$(ii)$] there exists a compact set $K\in M$ such that $\varphi_n\big\vert_{M\backslash K}\equiv \varphi\big|_{M\backslash K}$ for each $n\in \mathbb{N}$, and $\varphi_n \to \varphi$ uniformly on $K$.
    \end{itemize}
\end{proposition}

Notice that, due to Proposition \ref{prop_1_WO^0-convergence}, the space $C(M,E)$ with $E$ vector space is not a topological vector space, in particular the multiplication mapping cannot be continuous since if $\lambda\in \mathbb R $ goes to $0$, then $\lambda\cdot f \not\to 0$ unless $f=0$ outside some compact subset of $M$.

\begin{proof}
Suppose $\varphi_n \to \varphi $ in the $\mathrm{WO}^0$-topology, however, for all $K\subset M$ compact, either $\varphi_n \not\rightarrow \varphi$ uniformly over $K$ or there is $x\in M\backslash K$ such that $\varphi_n(x) \neq \varphi(x)$. In the first case $\varphi_n \not\rightarrow \varphi$ in the CO-topology as well, contradicting the initial hypothesis. In the second case, let $\{K_n\}$ be an exhaustion of compact subsets of $M$, then for each $n\in \mathbb{N}$ there is $x_n \in M\backslash K_n$ having $\varphi_n(x_n)\neq \varphi(x_n)$. Set $0<\epsilon_n =\sup_{K_n}d(\varphi_n(x),\varphi(x))$. For each $n$ and consider the sequence of open neighborhoods of $\varphi$, $W_{\varphi}(\epsilon_n)$ as per Lemma \ref{lemma_1_WO^0-metric-subbasis}, by construction $\varphi_n \notin W_{\varphi}(\epsilon_n)$ which contradicts the convergence hypothesis. On the other hand let $\epsilon_n:M \to \mathbb{R}_+$ be the constant functions with $\epsilon_n= \sup_{x\in K}d(\varphi_n(x),\varphi(x))$, then $W_{\varphi}(\epsilon_{n_0})\cap \{ \varphi_n\}=\{\varphi\}_{n>n_0}$ by uniform convergence over $K$. Implying $\varphi_n\to \varphi$ in $\tau_{\mathrm{WO}^0}$.
\end{proof}

\begin{corollary} \label{coro_1_WO^0-curves}
    Let $M$ and $N$ as in Proposition \ref{prop_1_WO^0-convergence} and $\gamma:I\subset \mathbb{R}\to \big(C(M,N),\tau_{\mathrm{WO}^0}\big)$ be a continuous mapping with $I$ compact. Then there exists a compact $K\subset M$ such that 
    $$
        \gamma(t) : x \in  M \to N
    $$
    is constant in $M\backslash K$ for each $t\in I$.
\end{corollary}

\begin{proof}
We argue by contradiction, let $K_n$ be an exhaustion of compact subsets of $M$, then for each $n\in \mathbb{N}$ there is some $t_n\in I$, and some $x_n\in M\backslash K_n$ such that $\gamma(t_n)[x_n]\neq \gamma(t)[x_n]$ for at least a $t\in I$. Since $\{t_n\}$ is a sequence on a compact space we may assume, eventually passing to a subsequence, that $t_n\to t_0\in I$, by construction, $\{x_n\}$ does not admit a cluster point in $M$. Finally, by continuity, $t_n \to t_0 \implies \gamma(t_n)\to \gamma(t)$ in the $\mathrm{WO}^0$-topology, by Proposition \ref{prop_1_WO^0-convergence} there has to be a compact subset $K$ such that $\gamma(t_n)\equiv \gamma(t)$ outside $K$ thus the sequence $\{x_n\}$ admits a cluster point.
\end{proof}

From now on we assume that $M$, $N$ are smooth $m$, $n$ dimensional manifolds respectively, consider the $k$th order jet bundle $J^k(M,N)$. Recall as well the mappings $\alpha:J^k(M,N) \to M$, $\beta:J^k(M,N) \to N$.

\begin{definition}\label{def_1_WO^k_topology}
    The \textit{Whitney} $\mathrm{C}^k$-\textit{topology}, or $\mathrm{WO}^k$-topology, on $C^r(M,N)$ for $0\leq k\leq r \leq \infty$ is the topology induced by requiring
    $$
        j^k:C^r(M,N) \to \Big( C\big(M,J^k(M,N)\big),\tau_{\mathrm{WO}^0}\Big)
    $$
    to be topological embedding.
\end{definition}

\begin{proposition}\label{prop_1_WO^k_properties}
    The $\mathrm{WO}^k$-topology on $C^r(M,N)$ enjoys the following properties:
    \begin{itemize}
        \item[$(i)$] A subbasis of open subsets of the topology have the form 
        \begin{equation}\label{eq_1_WO^k-open}
            W(\widetilde{U})=\{ \varphi \in C^r(M,N): j^k\varphi(M)\subset \widetilde{U}\}
        \end{equation}
        where $\widetilde{U}\subset J^k(M,N)$ is an open subset.
        \item[$(ii)$] If $d_k$ is a metric on $J^k(M,N)$, then a basis of neighborhoods for the $\mathrm{WO}^k$-topology of $\varphi\in C^r(M,N)$ is
        $$
            N_{\varphi}^k(\epsilon)=\{ \psi\in C^r(M,N): d_k(j^k_x\psi,j^k_x\varphi) < \epsilon(x)\}
        $$
        where $\epsilon\in C(M,\mathbb{R}_+)$.
        \item[$(iii)$] The sequence $\{\varphi_n\}\subset C^k(M,N)$ converges to $\varphi$ in the $\mathrm{WO}^k$-topology if and only if there is a compact subset $K\subset M$ such that $\varphi_n\equiv \varphi $ in $M\backslash K$ and $j^k\varphi_n \to j^k\varphi$ uniformly over $K$. 
        \item[$(iv)$] If $I\subset \mathbb{R}$ is compact and $\gamma:I \to \big( C^r(M,N), \tau_{\mathrm{WO}^k} \big)$ is continuous, then there is a compact subset such that 
        $$
            \mathrm{ev}_x\gamma: I\ni t \mapsto \gamma(t)[x]
        $$
        is constant for all $x\in M\backslash K$.
%        \item[5)] $\big( C^r(M,N), \mathrm{WO}^k\big)$ is a Baire space. 
        \item[$(v)$] $\mathrm{WO}^{\infty}$ on $C^{\infty}(M,N) $ is the projective limit topology of all $\mathrm{WO}^k$-topologies for $0\leq k\leq \infty$.
        \item[$(vi)$] A basis of open neighborhood of the $\mathrm{WO}^{\infty}$ topology on $C^{\infty}(M,N)$ consists of open subsets 
        \begin{equation}\label{eq_1_WO^infty-open}
            W(\widetilde{U})=\{ \varphi \in C^{\infty}(M,N): j^{\infty}\varphi(M)\subset \widetilde{U}\}
        \end{equation}
        where $\widetilde{U}\subset J^{\infty}(M,N)$ is open.
        \item[$(vii)$] If $\{K_n\}_n$ is an exhaustion of compact subsets of $M$, a basis for the $\mathrm{WO}^{\infty}$ topology on $C^{\infty}(M,N)$ consists of open subsets
        $$
            M(U,n)= \{ \varphi \in C^{\infty}(M,N): j^nf(M\backslash K_n^{\mathrm{o}})\subset U_n\}
        $$
        where $U_n \subset J^n(M,N)$ are open.
    \end{itemize}
\end{proposition}
\begin{proof}
We claim that on the image of $j^k$ in $C(M, J^k(M,N))$ in the $\mathrm{WO}^0$ and $\mathrm{WO}$ topology coincide. Indeed 
$$
    G_{j^k\varphi} (M)\subset M\times J^{k}(M,N)=\{ (x,j^k_x\varphi):x\in M)\}\simeq J^k(M,N),
$$
where the last is a topological embedding, therefore open subsets of $J^k(M,N)$ and $G_{j^k(C^r(M,N))}\subset  M \times J^k(M,N)$ coincide. 
As a result we obtain \eqref{eq_1_WO^k-open} by combining the above result with \eqref{eq_1_WO-open}. $(ii)$ follows by combining $(i)$ with Lemma \ref{lemma_1_WO^0-metric-subbasis}. Using that $\varphi_n\to \varphi $ in $\mathrm{WO}^k$ if and only if $j^k\varphi_n\to j^k\varphi$ in $\mathrm{WO}^0$ over $C(M,J^k(M,N))$ in conjunction with Proposition \ref{prop_1_WO^0-convergence} we get $(iii)$. Similarly $(iv)$ is obtained by combining Corollary \ref{coro_1_WO^0-curves} with the above argument. The argument for $(v)$ and $(vi)$ is the following: the topology on $J^{\infty}(M,N)$ is the coarsest such that each $\pi^{\infty}_k: J^{\infty}(M,N) \to J^k(M,N)$ is continuous. Note that we have an embedding
$$
    J^{\infty}(M,N)\simeq M\times_MJ^{\infty}(M,N) \hookrightarrow M\times J^{\infty}(M,N).
$$
Therefore we construct the following commutative diagram
\begin{center}
\begin{tikzcd}
		 & J^{\infty}(M,N) \arrow[r,hook] \arrow[d,"\pi^{\infty}_k"]  & M\times J^{\infty}(M,N) \arrow[d, "\mathrm{id}_M\times\pi^{\infty}_k"] \\
		 & J^{k}(M,N) \arrow[r,hook] & M\times J^k(M,N),
\end{tikzcd}
\end{center}
where the horizontal mapping are embeddings. Then a subbasis of $C\big( M,J^{\infty}(M,N)\big)$ for the $\mathrm{WO}^0$-topology is 
$$
    W(U_{\infty})=\{ j^{\infty}\varphi \in C^{\infty}(M,N): j^{\infty}(M)\subset U_{\infty}\}
$$
with $U_{\infty}\subset J^{\infty}(M,N)$ open. Finally we show $(vii)$. Let $U_n\subset J^n(M,N)$ be open subsets, then each
$$
    M(U)= \{ \varphi \in C^{\infty}(M,N): \ \forall n\in \mathbb N,\  j^{\infty}\varphi(M\backslash K_n^{\mathrm{o}})\subset (\pi^{\infty}_n)^{-1}U_n\};
$$
is an open subset of the $\mathrm{WO}^{\infty}$ topology. Setting $V_n=(\pi^{\infty}_0)^{-1}(U_0)\cap \cdots \cap (\pi^{\infty}_n)^{-1}(U_n)$ we have that
$$
    \{ \varphi \in C^{\infty}(M,N):  \ \forall n\in \mathbb N,\ j^{\infty}\varphi(K_{n+1}\backslash K_n^{\mathrm{o}})\subset V_n\}= M(U).
$$
The inclusion $\supset$ is clear, for the other, observe that in each region $K_{n+1}\backslash K_n^{\mathrm{o}}$ we have the requirement $ j^{\infty}\varphi(K_{n+1}\backslash K_n^{\mathrm{o}})\subset V_n$ for all $n$ which is way stronger then the corresponding $ j^{\infty}\varphi(M\backslash K_n^{\mathrm{o}})\subset (\pi^{\infty}_n)^{-1}U_n$ for all $n$. Since $J^{\infty}(M,N)$ is a fiber bundle with finite dimensional base and Fréchet space fiber and has the coarsest topology making each $\pi^{\infty}_k$ continuous, we may write
$$
    M(\widetilde{V})=\{ \varphi \in C^{\infty}(M,N): \ \forall n\in \mathbb N,\   j^{\infty}\varphi(K_{n+1}\backslash K_n^{\mathrm{o}})\subset \widetilde{V}_n\}.
$$
where each $\widetilde{V}_n\subset J^{\infty}(M,N)$ is open. We claim that $M(\widetilde{V})$ generates a topology equivalent to the $\mathrm{WO}^{\infty}$ topology. The latter's open subsets posses the form \eqref{eq_1_WO^infty-open}, is thus clear that $M(\widetilde{V})\subset W(\cup_n\widetilde{V}_n)$ thus making the former topology finer then the latter. To see the converse observe that $j^{\infty}\varphi(K_{n+1}\backslash K_n^{\mathrm{o}})$ is a compact subset of a metric space for each $n$, thus there is some $\epsilon_n>0 $ for which the open subset $\{j^{\infty}_x\psi \in \mathbb{R}^{\infty} :d(j^{\infty}_x\varphi,j^{\infty}_x\psi)<\epsilon_n \ \forall x\in K_{n+1}\backslash K_n^{\mathrm{o}} \} \subset \widetilde{V}_n$. Let then $\epsilon\in C(M)$ be a continuous function such that $\epsilon(x)<\epsilon_n$ for all $x\in K_{n+1}\backslash K_n^{\mathrm{o}}$, then $N^{\infty}_{\varphi}(\epsilon)$ is an open subset of the Whitney topology which is contained in $M(\widetilde{V})$.
\end{proof}

As a consequence of Proposition \ref{prop_1_WO^0-convergence}, when $N$ is metrizable and $M$ paracompact and second countable, given any sequence $\lbrace \varphi_{n}\rbrace_{n \in \mathbb{N}}$ we can characterize its convergence in the following way:
\begin{itemize}
\item[$(i)$] $\varphi_n \rightarrow \varphi $ in the $\mathrm{WO}^{\infty}-$topology ,
\item[$(ii)$] $\forall n'\in \mathbb{N} \space\ \exists K_{n'} \subset M$ compact such that if $n\geq n'$ then $\left. \varphi_n \right|_{M \backslash K_{n'}}=\left. \varphi \right|_{M \backslash K_{n'}}$ and $\left. \varphi_n \right|_{K_n'} \rightarrow \left. \varphi \right|_{K_n'}$ uniformly with all its derivatives.
\end{itemize}

This fact has important implications, for if we consider the finite dimensional vector bundle $E \to M$, the vector space $\Gamma^{\infty}(M\leftarrow E)$ will not be a topological vector space due to the failure of continuity for the multiplication by scalar. This can be readily seen from condition $(ii)$ above: if for instance we had $\sigma\in \Gamma^{\infty}(M\leftarrow E)$, $\mathbb{R} \ni \epsilon_n \to 0$ and $\epsilon_n\sigma \to 0$, then each $\epsilon_n\sigma$ must possess compact support, thus $\sigma$ itself ought to be compactly supported. As a consequence we get the following result:

\begin{theorem}\label{thm_1_Gamma_c_TVS}
    Let $(E,\pi,M)$ be a finite dimensional vector bundle, then $\Gamma^{\infty}_c(M\leftarrow E)\subset \Gamma^{\infty}(M\leftarrow E)$, equipped with trace of the Whitney topology on $C^{\infty}(M,E)$. Then $\Gamma^{\infty}_c(M\leftarrow E)$ is the maximal locally convex space contained in $\Gamma^{\infty}(M\leftarrow E)$. Moreover the trace topology coincides with the (natural) final topology induced by the projective limit
\begin{equation}
	\lim_{\substack{\longrightarrow\\K\subset M}}\Gamma^{\infty}_K(M\leftarrow E) =\Gamma^{\infty}_c(M\leftarrow E).
\end{equation}
Consequently, $\Gamma^{\infty}_c(M\leftarrow E) $ is a complete, nuclear and Lindel\"of space, hence paracompact and normal. In particular, for each open cover $\mathcal{U}_i$ of $\Gamma^{\infty}_c(M\leftarrow E)$, there are Bastiani smooth bump functions $\rho_i:\Gamma^{\infty}_c(M\leftarrow E)\to \mathbb{R}$ each of which has $\mathrm{supp}(\rho_i)\subset \mathcal{U}_i$, satisfying
    $$
        \sum_i\rho_i(\sigma)=1 
    $$
    for each $\sigma\in \Gamma^{\infty}_c(M\leftarrow E)$.
\end{theorem}

From a topological standpoint, Theorem \ref{thm_1_Gamma_c_TVS} implies that $\Gamma^{\infty}_c(M\leftarrow E)\subset \Gamma^{\infty}(M\leftarrow E)$ equipped with the Whitney topology is the maximal topological vector subspace. Thus if we want to give a topological manifold structure to spaces such as $\Gamma^{\infty}(M\leftarrow E)$ (or $C^{\infty}(M,N)$) this forces us to use $\Gamma^{\infty}_c(M\leftarrow E)$ as the topological vector space on which to model the manifold (see Definition \ref{def_1_infinite_dim_mfd}). It is therefore natural to seek charts of the form $(\sigma_0+\Gamma^{\infty}_c(M\leftarrow E),u_{\sigma_0})$ where $u_{\sigma_0}(\sigma)=\sigma-\sigma_0$. The undesirable fact is that $\sigma_0+\Gamma^{\infty}_c(M\leftarrow E)$ would then become a closed subset. To remedy this problem we refine the Whitney topology \textit{just enough} to make the above subsets open. To wit consider the following equivalence class: given $M$, $N$ smooth finite dimensional manifolds, set $\varphi \sim \psi$ if $\mathrm{supp}_{\varphi}(\psi)= \overline{\lbrace x \in M : \varphi(x) \neq \psi(x)\rbrace}\subset M$ is compact. 

\begin{definition}\label{def_1_refined_whitney_top}
    The \textit{refined Whitney topology}, or refined $\mathrm{WO}^{\infty}$ topology, is the coarsest topology on $C^{\infty}(M,N)$ which is finer than the $\mathrm{WO}^{\infty}$-topology and for which the sets $\mathcal{U}_{\varphi}=\lbrace \psi \in C^{\infty}(M,N): \psi \sim \varphi \rbrace$ are open.
\end{definition}

The refined Whitney topology has the same converging sequences and smooth curves as the Whitney topology since the proofs of \ref{prop_1_WO^0-convergence} and Corollary \ref{coro_1_WO^0-curves} remains essentially valid. Notice that the refinement we imposed on the topology was made by adding \textit{big} open subsets \textit{i.e.} the trace topology on subspaces of the form $\Gamma^{\infty}_c(M\leftarrow E)$ is not altered, thus the aforementioned properties remain valid. Clearly $\Gamma^{\infty}_c(M\leftarrow E) \subset\Gamma^{\infty}(M\leftarrow E)$ will then become open and moreover $\Gamma^{\infty}(M\leftarrow E)$ becomes a topological affine space with model topological vector space $\Gamma^{\infty}_c(M\leftarrow E)$. The locally convex modeling space $\Gamma^{\infty}_c(M\leftarrow E)$ will however no longer be a Baire space: if $N=\mathbb{R}$ and $K_n$ is an exhaustion of compact subsets of $M$, then $\cup_n C^{\infty}_{K_n}(M)=C^{\infty}_c(M)$, however $C^{\infty}_{K_n}(M)\subset C^{\infty}_{c}(M)$ is not dense for all $n\in \mathbb{N}$.

\begin{proposition}\label{prop_1_continuity_of_push_forward}
    Let $M$, $M'$, $N$, $N'$ be smooth finite dimensional manifolds,
    \begin{itemize}
        \item[$(i)$] if $f:M'\to M$ is a proper smooth mapping, then $f^*:C^{\infty}(M,N) \to C^{\infty}(M',N) $ is continuous in both the Whitney and refined Whitney topology;
        \item[$(ii)$] if $h:N\to N'$ is a smooth mapping, then $h_*:C^{\infty}(M,N) \to C^{\infty}(M,N') $ is continuous in both the Whitney and refined Whitney topology.
    \end{itemize}
\end{proposition}

\begin{proof}
    The proof of $(i)$ can be directly obtained by using Proposition 7.3 in \cite{michor2011manifolds} while keeping in mind that $f^*(\varphi)=\varphi\circ f$. For $(ii)$ consider a Whitney open subset $M'(U)=\{ \varphi\in C^{\infty}(M,N') : \ j^n\varphi(M\backslash K_n)\}\subset U_n \subset J^{n}(M,N') \forall n \in \mathbb N \}$ in $C^{\infty}(M,N')$. The mapping $j^nh : J^n(M,N) \to J^n(M,N')$ is smooth and hence continuous, thus set $V_n=(j^nh)^{-1}(U_n)$. Then $(h_*)^{-1}(M'(U))=M(V)$ which implies continuity of $h_*$ in the Whitney topology. For the refined Whitney topology one notes that if $\varphi\sim \varphi'$, then $h\circ \varphi \sim h \circ \varphi'$ as well, thus $h_*\mathcal{U}_\varphi=\mathcal{U}_{h_*\varphi}$, which in turn implies that $h_*$ remains continuous even when refining the topology. 
\end{proof}

% \begin{proposition}\label{prop_1_smooth_Bastiani_curve}
%     The mapping $C^{\infty}\big(M,\big(C^{\infty}(N,P),WO^{\infty}\big)\big) \to C^{\infty}\big(M\times N,P\big)$ is valued in the set of smooth functions $f:M\times N \to P$ such that for each compact $K\subset M$ there is a compact $L\subset N$ such that $f(x,y) $ is constant for all $x\in K$, $y\in N\backslash L$. If $N$ is not compact, the latter is not an isomorphism.
% \end{proposition}
% \begin{proof}
% \end{proof}

\begin{theorem}[\textbf{Proposition 4.8 pp. 38 \cite{michor2011manifolds}}]\label{thm_A_Bastiani smooth_pushforward}
    Let $(E_i,\pi_i,M)$, $i=1,2$ be finite dimensional vector bundles, suppose that $\alpha: U\subset E_1\to E_2 $ is a smooth fibered morphism projecting to the identity of $M$ and let $\sigma_0\in \Gamma^{\infty}_c(M\leftarrow E_1)$ having $\sigma_0(M)\subset U$, $\alpha(\sigma_0)\in \Gamma^{\infty}_c(M\leftarrow E_2)$. Then the mapping $\alpha_{*}:\mathcal{U}=\{\sigma\in \Gamma^{\infty}_c(M\leftarrow E_1): \sigma(M)\subset U\} \to \Gamma^{\infty}_c(M\leftarrow E_2)$ is a Bastiani smooth mapping, moreover, if $d_v\alpha:VU\to E_2$ is the vertical derivative of $\alpha$, we have $d(\alpha_*)=(d_v\alpha)_*$.
\end{theorem}

We remark that given any connection on the vector bundle $E_1$ it induces a splitting $TE_1=HE_1\oplus VE_1$ into horizontal and vertical vector bundle, the latter is of course independent from the connection chosen. Thus if we use local fibered coordinates on $E_1$ and $E_2$ induced by local frames $e_i$, $f_i$ respectively and study $\alpha $ in a neighborhood of $p\in E_1$, $\alpha(p)= \sum \alpha^i(x,y)f_i$, then $d_v\alpha:V_pE_1\to E_2, \ (p;\sigma) =\sum \frac{\partial \alpha^i(x,y)}{\partial y^j}\sigma^jf_i \in E_2|_{x}$.

\begin{proof}
    Its clear that if $U$ is open, $\mathcal{U}=\{\sigma\in \Gamma^{\infty}_c(M\leftarrow E_1): \sigma(M)\subset U\}$ is open in $\Gamma^{\infty}_c(M\leftarrow E_1)$ with the Whitney topology as well by $(vi)$ of Proposition \ref{prop_1_WO^k_properties} taking $\widetilde{\mathcal{U}}=(\pi^{\infty})^{-1}(M\times U)$. Next we show that $\alpha_*$ is Bastiani differentiable. This is equivalent to show that $d(\alpha_*):\mathcal{U}\times \Gamma^{\infty}_c(M\leftarrow E_1)\to \Gamma^{\infty}_c(M\leftarrow E_2)$ exists and is continuous in the Whitney topology (see Definition \ref{def_1_WO^k_topology}). We thus claim that  
    \begin{equation}
        \lim_{t\to 0} \frac{\alpha_*(\sigma+t\sigma')-\alpha_*(\sigma)}{t} = (d_v\alpha)_*(\sigma;\sigma')
    \end{equation}
    in the Whitney topology. We start by showing that for any neighborhood $U$ of $x_0\in M$, $\frac{\alpha_*(\sigma+t\sigma')-\alpha_*(\sigma)}{t}$ converges uniformly to $(d_v\alpha)_*(\sigma;\sigma')(x)$ in $U$. This is a local problem and we can thus study it using local coordinates. Notice that if $U$ lies outside the support of $\sigma'$ then the claim is trivial. By an abuse of notation we set $\sigma(x)=(x,\sigma(x))$ and likewise for $\sigma'$. Then by Taylor theorem we have
    $$
        \begin{aligned}
            \big(\alpha_*(\sigma+t\sigma')-\alpha_*(\sigma)\big)^i(x)&= \alpha^i(x,\sigma(x)+t\sigma'(x))-\alpha^i(x,\sigma(x)) = t\partial_j\alpha^i(x,\sigma(x))\sigma'^j(x)\\ & \quad +t^2\int_0^1(1-\lambda)\partial_{jk}\alpha^i(x,\sigma(x)+t\lambda\sigma'(x))\sigma'^j(x)\sigma'^k(x)d\lambda,
        \end{aligned}
    $$
    We can then estimate in $U$
    $$
    \begin{aligned}
        \bigg|\frac{\big(\alpha_*(\sigma+t\sigma')-\alpha_*(\sigma)\big)^i(x)}{t}-\big((d_v\alpha)_*(\sigma;\sigma')\big)^i(x) 
    \end{aligned}\bigg| \leq |t| C_{\alpha,\sigma,\sigma',U}
    $$
    for each $x\in U$, establishing uniform convergence. Moreover, since $d_v\alpha: VU\subset VE_1 \to E_2$ remains a smooth fibered morphism, the mapping $(d_v\alpha)_*:\mathcal{U}\times \Gamma_c^{\infty}(M\leftarrow E_1)\to \Gamma_c^{\infty}(M\leftarrow E_2)$ is continuous by Proposition \ref{prop_1_WO^k_properties}. Moreover, if $d^k_v\alpha : \otimes^k_MV_pU$ is the mapping locally defined by
    $$
        d^k_v\alpha[p]: \otimes^k V_pE_1\ni (s_1,\ldots,s_k) \mapsto
        \partial_{j_1 \ldots j_k}\alpha^i(p) s_1^{j_1}\cdots s_k^{j_K},
    $$
    a similar argument to the one above shows that for each $x\in U$ the mapping $d^{(k-1)}(\alpha_*)[\sigma+t\sigma'](\sigma_1,\ldots,\sigma_{k-1})(x)-d^{(k-1)}(\alpha_*)[\sigma](\sigma_1,\ldots,\sigma_{k-1})(x)$ converges uniformly to $(d^k_v\alpha )_*(\sigma;\sigma',\sigma_1,\ldots,\sigma_{k-1})(x)$. Then again, Proposition \ref{prop_1_continuity_of_push_forward} implies the continuity of $d^{(k)}(\alpha_*):\mathcal{U}\times \Gamma_c^{\infty}(M\leftarrow E_1) \cdots \times \Gamma_c^{\infty}(M\leftarrow E_1)\to \Gamma_c^{\infty}(M\leftarrow E_2)$. Finally, by $(iii)$ in Proposition \ref{prop_1_WO^k_properties}, this shows that the mapping $\alpha_*:\mathcal{U} \to \Gamma_c^{\infty}(M\leftarrow E_2)$ is Bastiani smooth.
\end{proof}

\section{Manifolds of mappings and sections}

We begin with the definition of infinite dimensional manifolds. As we mentioned earlier we shall choose to model those on locally convex spaces in view of the results by \cite{ely}, \cite{ely2}, \cite{omori}.

\begin{definition}\label{def_1_infinite_dim_mfd}
    Let $\mathscr{M}$ be a Hausdorff topological space, we say that $\mathscr{M}$ admits a Bastiani smooth manifold structure if 
    \begin{itemize}
        \item[$(i)$] there is a family $\{(\mathcal{U}_i, u_i, E_i)\}_{i\in I}$ where $\{\mathcal{U}_i\}_{i\in I}$ is an open cover of $\mathscr{M}$, $\{E_i\}_{i\in I}$ is a family of complete locally convex spaces and $u_i:\mathcal{U}_i \to E_i$ a family of homeomorphisms onto the open subsets $u_i(\mathcal{U}_i)\subseteq E_i$;
        \item[$(ii)$] for all $i,j \in I$, having $\mathcal{U}_{ij}=\mathcal{U}_i\cap \mathcal{U}_j \neq \emptyset$, the mapping 
        $$
            u_{ij}=u_j\circ u^{-1}_i: u_i(\mathcal{U}_{ij})\subseteq E_i \to u_j(\mathcal{U}_{ij})\subseteq E_j
        $$
        and its inverse $u_{ji}=u_i\circ u^{-1}_j$ are Bastiani smooth. 
    \end{itemize}
    We then call \textit{charts} elements of the family $\{(\mathcal{U}_i, u_i, E_i)\}_{i\in I}$.
\end{definition}

It follows from condition $(ii)$ above that the locally convex spaces $E_i$ linearly isomorphic. A subset $\mathscr{N}\subset \mathscr{M}$ of a differentiable manifold is called a \textit{splitting submanifold} of $\mathscr{M}$ if for each $p\in \mathscr{N}$, there are charts $(\mathcal{U},u,E)$ of $\mathscr{M}$ such that $u(p)=0\in E$ and $U(\mathcal{U}\cap \mathscr{N})=u(\mathcal{U})\cap F$, where $F$ is a closed vector subspace of $E$ for which $E=F\oplus F^c$. The collection of charts $\{(\mathcal{U}_i\cap \mathscr{N}, u_i\vert_{\mathcal{U}_i\cap \mathscr{N}},F_i)\}$ then makes $\mathscr{N}$ a manifold itself as per Definition \ref{def_1_infinite_dim_mfd}. A weaker notion of submanifolds requires that $F$ is just a closed subspace of $E$, in this case we say that $\mathscr{N}$ is a \textit{non-splitting submanifold}. 

Next we define the tangent bundle. Let $\mathscr{M}$ be a Bastiani smooth manifold with atlas $\lbrace(\mathcal{U}_{i},u_{i},E_{i})\rbrace$. A \textit{tangent vector} is an equivalence class of elements $(p,v,U_{i},u_{i},E_{i})$, with $p\in \mathscr{M}$ and $v \in E_{i}$, where $(p,v,\mathcal{U}_{i},u_{i},E_{i})$ and $(p',w,U_{j},u_{j},E_{j})$ are equivalent if $p=p'$ and $d(u_{ij}[u_{i}(p)])(v)=w $. We denote by $T_p\mathscr{M}$ the set of all tangent vectors to $p$, moreover setting $T\mathscr{M} = \bigsqcup_{p \in \mathscr{M}} T_p\mathscr{M} $ we obtain the \textit{space of tangent vectors} of $\mathscr{M}$. It is easy to see that $T\mathscr{M}$ carries a natural structure of Bastiani smooth manifold. To wit, observe that we can always define a canonical projection $\tau : T\mathscr{M} \to \mathscr{M} $. For the family of charts set $\{ (\widetilde{\mathcal{U}}_{i}, \widetilde{u}_{i}, E_{i} \times E_{i}) \} $ where $( {\mathcal{U}}_{i}, {u}_{i}, E_{i}) $ is a chart of $\mathscr{M}$, $\widetilde{\mathcal{U}}_{i}=\tau^{-1}\big( \mathcal{U}_i\big)$, $\widetilde{u}_{i}: \tau^{-1}(\mathcal{U}_i)\ni \widetilde{p}  \mapsto (u_{i}(x),v) \in E_i\times E_i$.  

The topology on $T\mathscr{M}$ is the unique one making each $\widetilde{u}_{i}$ into a homeomorphism, also the transition mapping $\widetilde{u}_{ij}:(x,v) \mapsto (u_{ij}(y), du_{ij}[x](v))$ is Bastiani smooth since $u_{ij}$ is itself smooth in the first place. It is easily shown that $T\mathscr{M}$ is Hausdorff, thus $T\mathscr{M}$ is a differentiable manifold according to Definition \ref{def_1_infinite_dim_mfd}.\\

Next we give a manifold structure to $C^{\infty}(M,N)$ with $M, \ N$ smooth finite dimensional manifolds. We first recall that given any Riemannian $h$ on $N$ there exists the Riemannian exponential $\exp_y: U\subset T_yN \to N, w\mapsto \exp_y(w)$, where $\exp_y(w)$ is the value of the geodesic starting at $y$ with velocity $w$ at time $t=1$. Since $\exp_y(0)=y$, and $T_y\exp_y=id_{T_yN}$, $\exp_y$ is a local diffeomorphism, then we can define a local diffeomorphism $(\tau_N,\exp):\widetilde{U}\subset TN \to \mathcal{O }\subset N \times N : (y,w) \to (y, \exp_{y}(w))$ onto an open subset $\mathcal{O }$ of the diagonal of $N\times N$. 

\begin{theorem}\label{thm_1_mfd_mappings}
    Let $M$, $N$ be smooth finite dimensional manifolds, then $C^{\infty}(M,N)$ is a Bastiani smooth manifold according to Definition \ref{def_1_infinite_dim_mfd}, modelled on the nuclear locally convex space $\Gamma^{\infty}_c(M\leftarrow f^*TN)$.
\end{theorem}

\begin{proof}
    Let $\varphi\in C^{\infty}(M,N)$, then define $\mathcal{U}_{\varphi}$ to be the subset of all $g \in C^{\infty}(M,N)$ with compact support with respect to $\varphi$, such that $(\varphi,\psi)(M)\subset (\tau_N,\exp)\big(\widetilde{U}\big)$, then $\mathcal{U}_{\varphi}$ is an open subset for example by $(vii)$ in Proposition \ref{prop_1_WO^k_properties}. Let then
    $$
        u_{\varphi}: \mathcal{U}_{\varphi} \ni \psi \mapsto u_{\varphi}(\psi)\in \Gamma^{\infty}_c(M\leftarrow \varphi^*TN)
    $$
    defined as follows:
    \begin{equation}\label{eq_1_trivial_gamma_local_chart}
        u_{\varphi}(\psi)(x)=(\tau_N,\exp)^{-1} (\varphi(x),\psi(x))\simeq \big(\varphi(x),\exp^{-1}_{\varphi(x)}(\psi(x))\big) .
    \end{equation}
    It is clear that $u_{\varphi}(\psi)$ is a smooth mapping, its image is valued in $\varphi^*TN$, moreover since $\varphi\neq \psi$ only in a compact subset of $K\subset M$, then $\exp^{-1}_{\varphi(x)}(\psi(x))\neq 0 $ if and only if $x \in K$. Thus $u_{\varphi}$ is valued into $\big\{\vec{X} \in \Gamma^{\infty}_c(M\leftarrow \varphi^*TN) : \vec{X}(M) \subset (\tau_N,\exp)^{-1}\big(\varphi^*\widetilde{U}\big)\big\}$. Therefore the mapping $u_{\varphi}$ becomes a homeomorphism between $\mathcal{U}_{\varphi}$ with the trace of the refined Whitney topology and an open subset of $ \Gamma^{\infty}_c(M\leftarrow \varphi^*TN) $ with the usual limit Fréchet topology. If $\mathcal{U}_{\varphi\psi}= \mathcal{U}_{\varphi}\cap \mathcal{U}_{\psi}\neq \emptyset$, then we can consider the transition mapping $u_{\varphi\psi}\doteq u_{\psi}\circ u_{\varphi}^{-1}:\Gamma^{\infty}_c(M\leftarrow \varphi^*TN) \to  \Gamma^{\infty}_c(M\leftarrow \psi^*TN)$. This mapping can be constructed as the push forward of
    $$
        T_{\varphi(x)}N \ni ( x,w) \mapsto \big(x, \exp_{\psi(x)}^{-1}\big(\exp_{\varphi(x)}(w)\big) \big)\in T_{\psi(x)}N,
    $$
    which is a smooth global fibered isomorphism $\varphi^*\widetilde U\subset \varphi^*TN\to \psi^*\widetilde U\subset \psi^*TN$. Then by Theorem \ref{thm_A_Bastiani smooth_pushforward} we also have that $u_{fg}$ is a smooth mapping together with its inverse $u_{\varphi\psi}$. Finally if one chooses a different metric $h'$ on $N$ inducing the exponential $\exp'$ and new charts $u'_{\varphi}$, then again, the transition mapping $u'_{\varphi}\circ u_{\varphi}$ can be obtained as the push forward of the local fibered isomorphism
    $$
        \varphi^*\widetilde U\subset \varphi^*TN \ni ( x,w) \mapsto \big(x, (\exp'_{\varphi(x)})^{-1}\big(\exp_{\varphi(x)}(w)\big) \big)\in \varphi^*\widetilde U'\subset \varphi^*TN.
    $$
    which by Theorem \ref{thm_A_Bastiani smooth_pushforward} is smooth. Therefore the smooth structure on $C^{\infty}(M,N)$ does not depend on the choice of the exponential mapping.
\end{proof}

We remark that in \cite{michor2011manifolds}, the role of the mapping $(\tau_N,\exp)$ is played by the so-called \textit{local addition}, that is a mapping $A:TN\to N\times N$ which is a local diffeomorphism onto an open subset of the diagonal for which $A(0_y)=y$ for all $y\in N$. One can show \cite[Lemmas 10.1 and 10.2 pp. 90]{michor2011manifolds} that $(\tau_N,\exp)$ is a local addition, then the proof of Theorem \ref{thm_1_mfd_mappings} can be repeated along the same lines.

In the general case of a non-trivial bundle $\pi:B\to M$, $\Gamma^{\infty}(M\leftarrow B)$ can be topologized as follows: first we give $C^{\infty}(M,B)$ the refined Whitney topology, then we note that $\varphi \in \Gamma^{\infty}(M\leftarrow B)  \subset C^{\infty}(M,B)$ if and only if $\pi_{*}(\varphi)=\pi\circ \varphi = \mathrm{id}_M$. By (ii) in Proposition \ref{prop_1_continuity_of_push_forward}, $\pi_{*}$ is continuous, so the equation $\pi_{*}(\cdot) =\mathrm{id}_M $ in $C^{\infty}(M,B)$ defines a closed subset in the refined Whitney topology. We wish to show that $\Gamma^{\infty}(M\leftarrow B)$ is a splitting submanifold of $C^{\infty}(M,B)$. First, notice that if $\varphi\in \Gamma^{\infty}(M\leftarrow B)\subset C^{\infty}(M,B)$, then $\mathcal{U}_{\varphi}\cap \Gamma^{\infty}(M\leftarrow B)$ is the set of all $\psi \in \Gamma^{\infty}(M\leftarrow B)$ such that $\psi$, $\varphi$ differ only on a compact subset of $M$. Secondly, observe that $TB$ can be split by the choice of a connection as $HB \oplus VB$, this induces the splitting $\Gamma^{\infty}_c(M\leftarrow \varphi^*TB)=\Gamma^{\infty}_c(M\leftarrow  \varphi^*HB)\oplus \Gamma^{\infty}_c(M\leftarrow \varphi^*VB)$ at the level of locally convex spaces. Finally, if $\varphi\in \Gamma^{\infty}(M\leftarrow B)$ and $(\mathcal{U}_{\varphi},u_{\varphi})$ is a chart of $C^{\infty}(M,B)$, then 
$$
    u_{\varphi}\vert_{\mathcal{U}_{\varphi}\cap \Gamma^{\infty}(M\leftarrow B)}: \mathcal{U}_{\varphi}\cap \Gamma^{\infty}(M\leftarrow B)\to \Gamma^{\infty}_c(M\leftarrow \varphi^*TB)
$$
with
\begin{equation}\label{eq_1_gamma_local_chart}
    u_{\varphi}(\psi)(x)=\big(x, \exp_{\varphi(x)}^{-1}(\psi(x))  \big).
\end{equation}
Notice that even though the image of $x\in M$ through $\psi$, $\varphi$, lies in the same fiber $\pi^{-1}(x)$ we are not guaranteed that $\exp_{\varphi(x)}^{-1}(\psi(x))\in V_{\varphi(x)}B$, for $\pi^{-1}(x)$ might fail to be totally geodesic for the Riemannian metric chosen on $B$; therefore, in general, the geodesic joining $\varphi(x)$ and $\psi(x)$ might travel outside $\pi^{-1}(x)$. If we were able to solve this issue and show that $\exp_{\varphi(x)}^{-1}(\psi(x))\in V_{\varphi(x)}B$, we could then proceed by noticing that although the splitting depends on the connection chosen, the vertical subbundle $VB = \mathrm{ker}(\tau_B:TB\to B)$ does not, therefore $\Gamma^{\infty}_c(M\leftarrow \varphi^*TB)=\Gamma^{\infty}_c(M\leftarrow  \varphi^*HB)\oplus \Gamma^{\infty}_c(M\leftarrow \varphi^*VB)$ depends on the connection chosen just on the horizontal part. We are thus left with solving the issue of $\pi^{-1}(x)$ being not totally geodesic. By \cite[Lemma 10.9]{michor2011manifolds} to each section $\varphi$ we can find a tubular neighborhood, \textit{i.e.} a vector bundle $(E_{\varphi},\widetilde \pi, \varphi(M))$ with $E_{\varphi}\subset B$ and $\widetilde \pi = \pi|_{E_{\varphi}}$. Moreover, by \cite[Lemma 10.6]{michor2011manifolds}, we can modify (through smooth diffeomorphisms) the local addition $(\tau_B,\exp_B)$ defining the chart $(\mathcal{U}_{\varphi},u_{\varphi})$ of $C^{\infty}(M,B)$ to a local addition $(\tau_{E_{\varphi}}, \widetilde \exp)$ on $E_{\varphi}$ for which 
$$
    \widetilde \exp_{\varphi}^{-1}(\psi) \in \Gamma^{\infty}_c(M\leftarrow \varphi^*TE_{\varphi}).    
$$
By construction, if $\psi(x)\in E_{\varphi}$, for all $x\in M$
\begin{equation}\label{eq_1_tilde_exp}
    \widetilde \exp_{\varphi(x)}^{-1}(\psi(x)) \in T_{\varphi(x)}\big( E_{\varphi}|_{\varphi(x)} \big) \simeq V_{\varphi(x)}E_{\varphi} \simeq V_{\varphi(x)}B.
\end{equation}
 In the sequel we will write the charts of $\Gamma^{\infty}(M\leftarrow B)$ as $(\mathcal{U}_{\varphi},u_{\varphi})$ understanding that their are the slice charts induced above. We shall call them \textit{ultralocal charts}\footnote{The term ultralocal has been introduced in \cite{forger2004currents} to signify that the mapping $u_{\varphi}$ does just depend on the point values of the mappings $\psi$, $\varphi$ without dependence on higher derivatives of the two.} in order to differentiate them from the local chart of finite dimensional manifolds that were mentioned before.\\

\begin{lemma}\label{lemma_A_locality&bastiani_implies_w-reg}
    Let $P:\Gamma^{\infty}(M \leftarrow B)\to \Gamma^{\infty}(M \leftarrow C)$ be a differential operator and $\Gamma^{\infty}(M \leftarrow B),\ \Gamma^{\infty}(M \leftarrow C)$ be endowed with the infinite dimensional structure described in Theorem \ref{thm_1_mfd_mappings}; then if $P$ is smooth it is weakly regular\footnote{To define weak regularity, consider a jointly smooth family of mappings $\Phi:\mathbb{R}\times M \to B$ such that
    \begin{itemize}
        \item for each $t\in \mathbb{R}$ fixed, $\Phi_t:M\to B in \Gamma^{\infty}(M\leftarrow B)$;
        \item there is a compact subset $H \subset M$ such that for all $x \notin H$, $\Phi_t(x)$ is constant in $t$.
    \end{itemize}
    We call such a family of mappings a one parameter \textit{compactly supported variation}. We say that a differential operator $P:\Gamma^{\infty}(M \leftarrow B)\to \Gamma^{\infty}(M \leftarrow C)$ is \textit{weakly regular} if given any compactly supported variation $\Phi$, the mapping $P[\Phi]:\mathbb{R}\times M \to C$, defined by $P[\Phi](t,x)=P[\Phi_t](x)$, is again a compactly supported variation.}.
\end{lemma}

\begin{proof}
    Suppose that $\varphi_s$ is a compactly supported variation of $\varphi_0$. We suppose also that $s\in I\subset \mathbb{R}$ with $I$ compact, but the general case is a straightforward generalization. We claim that $s\mapsto \varphi_s$ is a smooth curve in $\Gamma^{\infty}(M \leftarrow B)$, then again by Lemma \ref{lemma_1_interpolation_of_sections}, we can assume that the image of this path lies in a chart $\mathcal{U}_{\varphi}$, therefore our claim is equivalent to smoothness of 
    $$
        s \mapsto u_{\varphi_0}(\varphi_s)\in \Gamma^{\infty}_c(M \leftarrow \varphi^*VB),
    $$
    where $u_{\varphi}$ is the chart mapping defined in \eqref{eq_1_gamma_local_chart}. If $K\subset M$ is the compact where $\varphi_s\neq \varphi_0$, then $u_{\varphi}(\varphi_s)\neq u_{\varphi}(\varphi_0)$ in $K$ as well. Then it is enough to test differentiability at each order in the Fréchet space $\Gamma^{\infty}_K(M \leftarrow \varphi^*VB)$. 
    \begin{itemize}
        \item If we see $u_{\varphi}$ as a mapping from a neighborhood of the diagonal $\mathcal{O}\subset B\times B $ to the tangent space of $B$, then it is smooth, thus by (ii) in Proposition \ref{prop_1_continuity_of_push_forward}, we conclude that $u_{\varphi}(\varphi_s)=(u_{\varphi})_{*}\circ\varphi_s$ is continuous.
        \item To show differentiability, note that $\varphi(\pi(\varphi_s(x)))=\varphi(x)$ for all $x\in M$, therefore $u_{\varphi}(\varphi_s)=\widetilde\exp_{\varphi}^{-1}(\varphi_s)$, then for all $x\in M$
        $$
            \frac{d}{ds}u_{\varphi}(\varphi_s)(x)= T_{\varphi_s(x)}\widetilde\exp_{\varphi(x)}^{-1}(\dot\varphi_s);
        $$
        thus the derivative $ \frac{d}{ds}u_{\varphi}(\varphi_s)$ exists and by (ii) in Proposition \ref{prop_1_continuity_of_push_forward} is continuous due to smoothness of $T_{\bullet}\widetilde\exp_{\bullet}^{-1}: T\mathcal{O} \to TTB$.
        \item Iterating this argument we have shown smoothness of $u_{\varphi_s}$. 
    \end{itemize}
    Since $P$ is smooth, then $P(\varphi_s)$ is a smooth curve in $\Gamma^{\infty}(M \leftarrow C)$, eventually shrinking $I$ we can 
    assume that $P(\varphi_s)\subset \mathcal{V}_{P(\varphi_0)}$ \textit{i.e.} it lies inside a chart $(\mathcal{V}_{P(\varphi_0)}, v_{P(\varphi_0)})$ of $\Gamma^{\infty}(M \leftarrow C)$. Then $v_{P(\varphi_0)}\big(P(\varphi_s)\big)\subset \Gamma^{\infty}_c(M \leftarrow \sigma^*VC)$ is smooth. By $(iii)$, $P(\varphi_s)\in C^{\infty}(M,C)$ is a smooth curve and there is a compact subset $K'$ such that $P(\varphi_s)\equiv P(\varphi_0)$ outside $K'$. Therefore 
    $$
        v_{P(\varphi_0)}\big(P(\varphi_s)\big)\subset \Gamma^{\infty}_{K'}(M \leftarrow \sigma^*VC),
    $$
    the latter being a Fréchet space we can apply \cite[(2) Lemma 3.9]{kriegl1997convenient} and conclude.
\end{proof}

The tangent space at each point $\varphi$ is $T_{\varphi}\Gamma^{\infty}(M\leftarrow B)\equiv \Gamma^{\infty}_c(M\leftarrow \varphi^{*}VB)$. The tangent bundle $(T\Gamma^{\infty}(M\leftarrow B),\tau_{\Gamma}, \Gamma^{\infty}(M\leftarrow B))$ is defined in analogy with the finite dimensional case, and carries a canonical infinite dimensional bundle structure with trivializations
$$
	t_{\varphi}: \tau^{-1}_{\Gamma}(\mathcal{U}_{\varphi})\rightarrow \mathcal{U}_{\varphi} \times \Gamma^{\infty}_c(M\leftarrow \varphi^{*}VB).
$$
As usual, we can identify points of $T\Gamma^{\infty}(M\leftarrow B)$ by elements $t^{-1}_{\varphi}\left(\varphi,\vec{X}_{\varphi}\right)$. With those trivializations a tangent vector to $\Gamma^{\infty}(M\leftarrow B)$, \textit{i.e.} an element of $T_{\varphi}\Gamma^{\infty}(M\leftarrow B)$, can equivalently be seen as a section of the vector bundle $\Gamma^{\infty}_c(M\leftarrow \varphi^{*}VB)$. When using the latter interpretation, we will write the section in local coordinates as $\vec{X}(x)=\vec{X}^i(x)\partial_i\big\vert_{\varphi(x)}$. %To avoid confusion between the two we adopt different notations: we shall write the arrow symbol over capital Roman letters, \textit{e.g.} $(\vec{X}, \vec{Y}, \vec{Z}, \dots)$, to denote elements in $T_{\varphi}\Gamma^{\infty}(M\leftarrow B)$; we shall use instead use $(\mathcal{X}, \mathcal{Y}, \mathcal{Z}, \dots) $ to denote sections of $\Gamma^{\infty}_c(M\leftarrow \varphi^{*}B)$. 
Finally we will use Roman letters, \textit{e.g.} $(s,\vec{u}, \dots)$ to denote elements of the topological dual space $\Gamma^{-\infty}_c(M\leftarrow \varphi^{*}VB)\equiv\big(\Gamma^{\infty}_c(M\leftarrow \varphi^{*}VB)\big)'$.
\begin{definition}\label{def_1_connection}
A connection over the (possibly infinite dimensional) bundle $(C,\pi,X)$ is a vector-valued one form $\Phi \in \Omega^1(C;VC)$ satisfying
	\begin{itemize}
	\item[$(i)$] $\mathrm{Im}(\Phi) = VC$,
	\item[$(ii)$] $\Phi \circ \Phi=\Phi$.
	\end{itemize}
\end{definition}
The mapping $\Phi $ represents the projection onto the vertical subbundle of $TC$. Given a connection $\Phi$ it is always possible to associate its canonical Christoffel form $\Gamma\doteq \mathrm{id}_{TC}-\Phi$ which will define the projection onto the space of horizontal vector fields. In our case, we consider $C=T\Gamma^{\infty}(M\leftarrow B)$, the latter has canonical trivialization
$$
	Tt_{\varphi}: \tau^{-1}_{T\Gamma}\circ \tau^{-1}_{\Gamma}(\mathcal{U}_{\varphi})\rightarrow \mathcal{U}_{\varphi} \times \Gamma^{\infty}_c(M\leftarrow \varphi^{*V}B) \times \Gamma^{\infty}_c(M\leftarrow \varphi^{*}VB) \times \Gamma^{\infty}_c(M\leftarrow \varphi^{*}VB).
$$
Therefore given $ Tt_{\varphi}^{-1}\left(\vec{Y}_{\varphi}, s_{\vec{X}}\right) \in TT\Gamma^{\infty}(M\leftarrow B)$, we can write the connection locally as
$$
	t_{\varphi}^{*}\Phi\left(\vec{Y}_{\varphi}, s_{\vec{X}}\right)=\left(\vec{0}_{\varphi}, s_{\vec{X}}-\Gamma_{\varphi}(\vec{X},\vec{Y}) \right),
$$
where the Christoffel form
$$
	\Gamma_{\varphi}\equiv t_{\varphi}^{-1}\Gamma: \Gamma^{\infty}_c(M\leftarrow \varphi^{*}VB) \times \Gamma^{\infty}_c(M\leftarrow \varphi^{*}VB) \rightarrow\Gamma^{\infty}_c(M\leftarrow \varphi^{*}VB) 
$$
can be chosen to be linear in the first two entries. For additional details about connections see \cite[Chapter VI, $\S$ 37]{kriegl1997convenient}. Instead of using the abstract notion provided by Definition \ref{def_1_connection}, in the case of manifolds of mappings, there is a more intuitive way of generating a connection. For simplicity's sake we shall do the easier case of $C^{\infty}(M,N)$, since the generalization to general bundles is almost immediate. Let $\widetilde{\Gamma}$ be a connection on the finite dimensional manifold $TN$, then we induce a connection $\Phi$ on $TC^{\infty}(M,N)$ as follows: fix $f\in C^{\infty}(M,N)$, $\vec{X},\vec{Y}\in T_fC^{\infty}(M,N)\simeq \Gamma^{\infty}_c(M\leftarrow f^{*}TN)$, $s\in T_{\vec{X}}T_fC^{\infty}(M,N)\simeq\Gamma^{\infty}_c(M\leftarrow f^{*}TN)$, then
\begin{equation}\label{eq_1_def_connection}
\big(t_{f}^{-1}\Phi\big)\big(f,\vec{X},\vec{Y},s\big)\doteq \big( \vec{0}_f, s-\Gamma_f(\vec{X},\vec{Y})\big)   
\end{equation}
where $\Gamma_f(\vec{X},\vec{Y})\in \Gamma^{\infty}_c(M\leftarrow f^{*}TN)$ is defined by 
\begin{equation}\label{eq_1_connection_induced}
    \Gamma_f(\vec{X},\vec{Y})(x)\doteq \widetilde{\Gamma}_{jk}^i(f(x))\vec{X}^j(x)\vec{Y}^k(x) \partial_i\vert_{f(x)}.    
\end{equation}
Equivalently we are setting $\Gamma_f = \widetilde{\Gamma}_*$, by Theorem \ref{thm_A_Bastiani smooth_pushforward}, the mappings $\Gamma_f$, $\Phi_f$ are Bastiani smooth, moreover they induce a connection $\Phi$. In the sequel we shall use \eqref{eq_1_conn} to induce a connection as in \eqref{eq_1_def_connection}.

\printbibliography

\end{document}